\documentclass[10pt,conference]{IEEEtran}

\usepackage[utf8]{inputenc}

\usepackage{float}
\usepackage[colorlinks,bookmarksopen,bookmarksnumbered,citecolor=black,urlcolor=black,linkcolor=black]{hyperref}
\usepackage{mathtools}
\usepackage{amsthm}
\usepackage{amssymb}
\usepackage{amsmath}
\usepackage{yhmath}
\usepackage{amsfonts}
\usepackage{url}
\usepackage{listings}
\usepackage{color}
\usepackage[dvipsnames]{xcolor}
\usepackage{graphicx}
\usepackage{caption}
\usepackage{subfigure}
\usepackage[noadjust]{cite}
\usepackage{breqn}
\usepackage{makecell}
\usepackage{enumitem}
\usepackage[titlenumbered,ruled,linesnumbered]{algorithm2e}
\usepackage{balance}
\usepackage{amssymb}
\usepackage{multirow}
\usepackage{multicol} 
\usepackage{soul}
\usepackage{adjustbox}
\usepackage{arydshln}

\usepackage{booktabs} 
\usepackage{array}    
\usepackage{graphicx} 
\usepackage{placeins}
\usepackage{diagbox}

\newtheorem{definition}{Definition}
\newtheorem{theorem}{Theorem}
\newtheorem{lemma}{Lemma}

\newtheorem{example}{Example}

\definecolor{cmdcolor}{RGB}{34, 153, 84}

\newcommand{\eg}{{\it e.g.}\xspace}
\newcommand{\ie}{{\it i.e.}\xspace}

\usepackage{tikz}

\usepackage{pifont}

\newcommand{\add}[1]{#1}

\newcommand{\task}{\tau}
\newcommand{\asso}{\kappa}          

\newcommand{\cnt}{\delta}           
  
\newcommand{\cp}{\lambda}           
\newcommand{\CR}{\mathcal{C}}
\newcommand{\ctx}{r}             
\newcommand{\ctxset}{R}          
\usepackage{tabularx}

\IEEEoverridecommandlockouts

\begin{document}

\thispagestyle{plain}
\pagestyle{plain}

\title{Tight Cache Contention Analysis for\\WCET Estimation on Multicore Systems}

\author{
\IEEEauthorblockN{Shuai Zhao\IEEEauthorrefmark{2}, Jieyu Jiang\IEEEauthorrefmark{2}, Shenlin Cai\IEEEauthorrefmark{2}, Yaowei Liang\IEEEauthorrefmark{2}, Chen Jie\IEEEauthorrefmark{2}, Yinjie Fang\IEEEauthorrefmark{3}, Wei Zhang\IEEEauthorrefmark{4},\\ Guoquan Zhang\IEEEauthorrefmark{6}, Yaoyao Gu\IEEEauthorrefmark{6}, Xiang Xiao\IEEEauthorrefmark{6}, Wei Qin\IEEEauthorrefmark{6}, Xiangzhen Ouyang\IEEEauthorrefmark{6}, Wanli Chang\IEEEauthorrefmark{3}
}
\IEEEauthorblockA{
\IEEEauthorrefmark{2}Sun Yat-sen University, China
\IEEEauthorrefmark{3}Hunan University, China
\IEEEauthorrefmark{4}Shandong University, China
\IEEEauthorrefmark{6}Xiaomi Corporation, China
}
}

\maketitle

\begin{abstract}

WCET (Worst-Case Execution Time) estimation on multicore architecture is particularly challenging mainly due to the complex accesses over cache shared by multiple cores.
Existing analysis identifies possible contentions between parallel tasks by leveraging the partial order of the tasks or their program regions.
Unfortunately, they overestimate the number of cache misses caused by a remote block access without considering the actual cache state and the number of accesses.
This paper reports a new analysis for inter-core cache contention.
Based on the order of program regions in a task, we first identify memory references that could be affected if a remote access occurs in a region.
Afterwards, a fine-grained contention analysis is constructed that computes the number of cache misses based on the access quantity of local and remote blocks.
We demonstrate that the overall inter-core cache interference of a task can be obtained via dynamic programming.
Experiments show that compared to existing methods, the proposed analysis reduces \add{inter-core cache interference and WCET estimations by 52.31\% and 8.94\% on average}, without significantly increasing computation overhead.
\end{abstract}

\section{Introduction} 
\label{sec:intro}

The worst-case execution time (WCET) of tasks is essential for ensuring the timing predictability of real-time systems, \ie, being able to finish before deadlines~\cite{wilhelm2008worst,hahn2018static}. 
In WCET estimations, the cache contention is a major source of latency, where cache misses can impose a significant delay on the execution of a task~\cite{liang2012timing,zhang2022precise,colin2000worst}.
For uniprocessor systems, the cache analysis focuses on the intra-core contention caused by memory accesses of a task itself, providing the maximum cache miss number among the execution paths of the task~\cite{lv2016survey,reineke2007timing,sen2007wcet}.
Matured techniques are available for analysing the intra-core cache contention~\cite{lv2016survey,chattopadhyay2009unified}, \eg, the Must/May~\cite{alt1996cache} and Persistence analysis~\cite{cullmann2013cache,zhang2015improving,stock2019cache} that compute the age of memory blocks across different paths and loop structures.

With the growing demand for performance and complex functionalities, multicore architectures are widely adopted in real-time systems, which often contain a shared cache that is accessed from multiple cores in parallel~\cite{davis2011survey, chattopadhyay2014unified}.
In such systems, the inter-core cache contention can occur when memory blocks required by a task are evicted due to the block accesses of the remote tasks~\cite{yan2008wcet, dharishini2021precise}.
Such contention can lead to substantial misses on the shared cache, resulting in significant latency in task execution.
Therefore, for multicore real-time systems, the inter-core cache contention analysis is necessary to provide correct WCET estimations of tasks~\cite{maiza2019survey}.

Ideally, a precise bound on the inter-core cache interference of a task can be obtained by analysing all potential contentions imposed from remote tasks at each program point, with the cache state taken into account (\ie, the blocks in the shared cache and their ages). 
However, due to the non-deterministic arrivals of parallel tasks~\cite{fischer2023analysis}, this can cause severe scalability issues that fundamentally undermine the applicability of the analysis. In addition, the complexity of the cache and the intricacies of the programs (\eg, loop structures) further exacerbate the problem~\cite{maiza2019survey}, making it even challenging to obtain a precise analytical bound efficiently.

Therefore, rather than formulating an optimisation problem across all program points~\cite{nagar2016precise}, existing approaches explore the polynomial-time solutions that identify contentions which can occur in the worst case~\cite{liang2012timing,zhang2022precise}.
Specifically, the analysis in~\cite{liang2012timing} leverages the lifetimes of tasks to determine the ones that cause a contention, thereby reducing the number of interfering tasks.
A more recent method in~\cite{zhang2022precise} exploits the partial order between program regions to identify contentions that are mutually exclusive in time, and applies a dynamic program to compute the number of misses based on the feasible contentions.

However, while reducing the computation complexity, these methods often rely on a pessimistic assumption that a remote access can impose a contention on all accesses of every block, neglecting the actual cache state of blocks (\ie, age) or the access behaviours (\ie, number of accesses) of the parallel tasks (see Sec.~\ref{sec:model-existing} for details).
Consequently, they can overestimate the number of cache misses caused by a remote access,
which leads to overly conservative analytical results.

\textbf{Main Contribution:} 
This paper presents a novel inter-core cache contention analysis. 
Similar to~\cite{zhang2022precise}, the constructed analysis is performed at the program regions of tasks. 
However, for each region, we determine the number of cache misses by examining every possible age and access number of a memory block, providing tighter analytical bounds that address the limitation in existing analysis.
To achieve this, the following contributions are proposed.
\begin{itemize}
\item 
A sequence of contention regions (CRs) where each contains the memory references can be interfered if a remote access occurs in the corresponding program region.
\item A fine-grained contention analysis that quantifies the cache miss number of memory references in a CR with the actual cache accessing behaviours taken into account. 
\item A complete inter-core cache analysis over a sequence of CRs via dynamic programming that leverages the partial order between the regions of parallel tasks.
\end{itemize}

We integrate the proposed analysis into an existing WCET analysis tool named LLVM-TA~\cite{hahn2022llvmta}, forming a complete solution for analysing WCETs on multicore real-time systems\footnote{The implementation of the proposed inter-core cache analysis is available at \url{https://github.com/RTS-SYSU/Timing-Analysis-Multicores}.}.
To the best of our knowledge, this is the first inter-core cache analysis that considers both the actual cache state and accessing behaviours of parallel tasks. 
Experiments show that compared to the existing analysis, the proposed analysis reduces \add{inter-core cache interference and WCET estimations by 52.31\% and 8.94\%} on average under tasks in TacleBench without significantly increasing the computation cost.

\textbf{Organisation:} Sec.~\ref{sec:model} presents the preliminaries of the proposed analysis. Sec.~\ref{sec:crm} constructs the CRs based on the program regions of a task. Sec.~\ref{sec:contention} presents the contention analysis on a CR. Sec.~\ref{sec:multi} forms the complete inter-core cache analysis of a task. Finally, Sec.~\ref{sec:result} describes the experimental results and Sec.~\ref{sec:conclusion} draws the conclusion.
\section{Preliminaries}
\label{sec:model}

\begin{table}[t]
\vspace{4pt}
\centering
\caption{Notations Introduced in Sec.~\ref{sec:model}.}
\label{tab:notations_model}
\renewcommand{\arraystretch}{1.1}
\begin{tabular}{p{.18\columnwidth}m{.72\columnwidth}}
\hline
\textbf{Notation} & \textbf{Description} \\
\hline
$\asso$ & The associativity of the shared cache.\\
$\task$~/~$\task'$ & The task under analysis / the interfering task.\\
\hline
$b_{i}$ & A memory block $b_i$.\\
$\omega_i$ & The block address of $b_i$. \\
$U_x$ & A unordered region with an index $x$.\\

$\cnt_x$ & The number of executions of $U_x$.\\

\hline
$\cp$ & A path that contains a sequence of out-most URs.\\
$\Omega_x$ & The set of unique block addresses in $U_x$.\\
$anc(U_x)$ & The complete set of the outer URs that contain $U_x$. \\
\hline
\end{tabular}
\vspace{-3pt}
\end{table}


\subsection{System Model} \label{sec:model-system}
\textbf{Hardware Platform.} 
To simplify the presentation of the analysis, we focus on a single-level fully associative cache shared between two cores.
The associativity (\ie, the number of cache lines) of the cache is denoted as $\asso$, where a cache line can store one memory block loaded from the main memory. 
\add{As with~\cite{wilhelm2008worst,hahn2018static,hardy2009using,nagar2016fast,zhu2024fine,stock2019cache}, the Least Recently Used (LRU) policy~\cite{reineke2007timing} is applied as the replacement policy.}
Sec.~\ref{sec:multi-multi} describes the application of the analysis for a multi-level set-associative cache hierarchy shared among multiple cores. 
Notations introduced in this section are summarised in Tab.~\ref{tab:notations_model}.

\textbf{Memory Blocks.} 
The system contains two parallel tasks $\task$ and $\task'$ competing for the cache, where $\task$ is under analysis and $\task'$ is an interfering task.
As with~\cite{zhang2022precise,stock2019cache}, a task is represented by a Control Flow Graph (CFG) that accesses a set of memory blocks at different program points.
Notation $b_i$ indicates a memory block being accessed at a distinct program point.
A $b_i$ has a block address $\omega_i$ that indicates the data location.
For one access to $b_i$, its data can be found either in the cache (cache hit) or the memory (cache miss). We note that different $b_i$ can have the same address, indicating the same piece of data accessed from different program points in $\tau$.
For example, Fig.~\ref{fig:cfg} presents a CFG that accesses six blocks at different program points, in which $\omega_1 = \omega_5$ and $\omega_2 = \omega_6$. 

Given a series of block accesses in $\tau$ (\ie, the traces in~\cite{stock2019cache}), the age of $b_i$ is defined as the maximum number of \textit{unique blocks} (\ie, ones with different addresses) being accessed since the most recent access to $\omega_i$ (if it exists)~\cite{stock2019cache,liang2012timing}.
For example, in the access sequence $\{b_1, b_2, b_3, b_5\}$ with $\omega_1 = \omega_5$, the age of $b_1$ is $\infty$ and the age of $b_5$ is 2.
In this work, we assume that the ages of blocks are known in advance based on the existing intra-core cache analysis~\cite{alt1996cache,stock2019cache}.

\textbf{Unordered Program Regions.} For a CFG, obtaining a complete block access sequence of a path (\ie, the control flow) requires unrolling all loops~\cite{stock2019cache}, which can be time-consuming due to the intricate loop structures. 
To avoid this, the unordered program region (UR) in~\cite{zhang2022precise} is applied to organise block accesses in loop structures while maintaining the partial order between the URs.
A UR is denoted as $U_x$, which has an execution number $\cnt_x$ based on the loop count,
\eg, $U_2 = \{ b_2, b_3 \}$ with $\cnt_2 = 3$ and $U_4 = \{ b_5, b_6 \}$ with $\cnt_4 = 5$ in Fig.~\ref{fig:cfg}.
For a block that does not reside in any loop, a UR is constructed as a singleton set with an execution number of one, \eg, $U_1 = \{ b_1 \}$ with $\cnt_1 = 1$ in Fig.~\ref{fig:cfg}. 
Notation $\Omega_x$ denotes the set of unique block addresses being accessed in a $U_x$, \ie, $\Omega_x=\{\omega_i~|~b_i \in U_x\}$.



A path of CFG (denoted as $\cp_q$) contains a sequence of the \textit{out-most URs} (the ones without any outer URs) that are executed in order, \eg, $\cp_1 = \{ U_1, U_2, U_4\}$ in Fig.~\ref{fig:cfg}. 
Hereinafter, we omit the index $q$ of $\cp_q$ if there exists only one path in the context. 
For nested UR structures (\ie, nested loops), function $anc(U_x)$ returns the complete set of the outer URs that contain $U_x$ directly or transitively, \eg, $anc(U_5) = \{U_4\}$.

\begin{figure}
    \centering
    \includegraphics[width=0.8\linewidth]{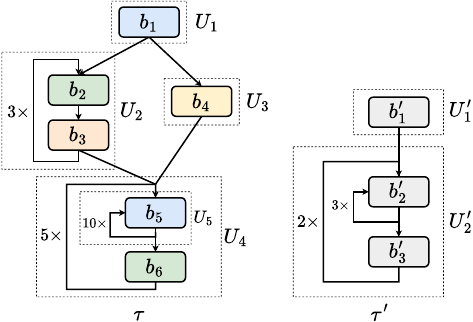}
    \caption{A control flow graph organised by unordered regions \textit{\add{(the same colour in $\tau$ indicates a same block addresses)}}.}
    \label{fig:cfg}
    \vspace{-7pt}
\end{figure}

\subsection{Existing Analysis on Inter-core Cache Contention } \label{sec:model-existing}

Extensive efforts have been conducted to estimate the inter-core cache contention, aiming to provide efficient solutions with tight bounds~\cite{hardy2009using,nagar2016fast,nagar2016precise,fischer2023analysis,zhu2024fine}. 
In~\cite{nagar2016precise}, an optimisation problem is constructed that finds the maximum number of cache misses of $\tau$ by examining the contention from remote accesses across different program points of $\tau$.
The analysis in~\cite{fischer2023analysis} determines whether a block can be evicted before its next access based on the minimum time required for the remote tasks to access a sufficient number of unique blocks.
However, these methods rely on integer linear programming, which can incur severe scalability issues~\cite{maiza2019survey}.
Below we detail two analyses with polynomial-time complexity.

In~\cite{liang2012timing}, a cache analysis is constructed that identifies the interfering tasks of $\task$ by examining their lifetimes, \ie, the duration from the earliest start to the latest finish of a task. The initial lifetime of a task is obtained based on the intra-core analysis.
The analysis then computes the inter-core cache interference and extends the lifetime based on the estimated latency in an iterative manner. 
During an iteration, it classifies accesses to a block in $\tau$ as cache misses as long as the number of unique addresses accessed by the interfering tasks exceeds the difference between $\asso$ and the block's age.
This analysis provides an efficient approach for computing the inter-core cache interference. However, it assumes that the accesses of interfering tasks can affect every block access of $\task$ at any time of its execution, which leads to overly pessimistic results.

\begin{table}[t]
\vspace{4pt}
\centering
\caption{Notations Introduced in Sec.~\ref{sec:crm}.}
\label{tab:notations_crm}
\renewcommand{\arraystretch}{1.1}
\begin{tabular}{p{.25\columnwidth}m{.65\columnwidth}}
\hline
\textbf{Notation} & \textbf{Description} \\
\hline

$age(b_i, U)$ & The age of $b_i$ in the scope of the program. \\
$age(b_i, U_x)$ & The age of $b_i$ in the scope of $U_x$. \\
\mbox{$\ctx^k=(\omega^k, \cnt^k, age^k)$} & A memory reference that contains $\cnt^k$ accesses to block address $\omega^k$ with an age of $age^k$. \\
$R^x$ & The set of all $r^k$ in $U_x$. \\
$U(b_i)$ & The set of nested URs that include $b_i$. \\
\hline
$\alpha^k$, $\beta^k$ & The indices of the first and last UR of a UR sequence in which a remote access can affect $r^k$. \\
\add{$\check{U}(r^k)$} & The index of the UR that contains the reference immediately following the previous access to $\omega^k$. \\
$\mathcal{C}_s$ & A contention region with an index s.\\
$\mathcal{C}$ & The sequence of contention regions for a UR path.\\
\hline
\end{tabular}
\vspace{-7pt}
\end{table}

More recently, an analysis is proposed in~\cite{zhang2022precise}, \add{which introduces the notion of UR to organise the memory blocks accessed in a path of CFG as a sequence of out-most URs.
The analysis constructs the paths by merging the basic blocks representing loop structures into a monolithic block (\ie, an out-most UR). 
This effectively simplifies the timing relations by transforming them from instruction-level to UR-level, alleviating the path explosion problem.}
Then, given one UR path of $\tau$ and $\tau'$, impossible contention scenarios can be identified based on the partial order between these UR paths.
\add{
For $\task$ and $\task'$ in Fig.~\ref{fig:cfg}, if $U_1$ incurs contention from $U'_2$, then $U_2$ cannot be affected by $U'_1$ due to their temporal order. 
The analysis then applies dynamic programming to find the maximum number of cache misses of $\task$, excluding contention scenarios that violate such temporal constraints. By eliminating the contention scenarios that are temporally impossible, it can tighten the WCET estimations.}
However, for a $U_x$ and $U'_x$, the analysis treats all accesses of a block $b_i \in U_x$ as cache misses if there exists a block in $U'_x$ that maps to the same cache set as $b_i$, overlooking the actual cache state (\eg, the $\asso$ and the ages) and the number of block accesses in $U'_x$. 
Hence, it still incurs a certain degree of pessimism, especially with a large $\asso$ or $\cnt_x$ of a $U_x$.

\vspace{-1pt}
\subsection{Motivation}\label{sec:model-mot}

To address these limitations, this paper presents a new inter-core cache analysis that provides fine-grained quantification of cache misses by considering both the cache state and the accessing behaviour of the parallel tasks. 
By examining the block accesses in each (nested) UR, we construct memory references that explicitly account for the different ages and the associated access numbers of a $b_{i}$  (Sec.~\ref{sec:crm-ref}).
Then, a sequence of contention regions (CR) is constructed to identify memory references that can be affected if a remote block access occurs at each $U_x \in \lambda$ (Sec.~\ref{sec:crm-region}).
For a CR, a fine-grained contention analysis is proposed to bound the number of misses given a sequence of interfering URs, with the age and the access number of blocks taken into account (Sec.~\ref{sec:contention}).
Finally, dynamic programming is applied to find the worst-case inter-core cache interference of $\tau$ across program regions, forming the complete cache contention analysis (Sec.~\ref{sec:multi}).


\section{Construction of Contention Regions}
\label{sec:crm}

To obtain precise inter-core cache interference estimations, it is necessary to (i) identify references to a $b_i$ with a different age (\eg, in nested URs) - reflecting different likelihoods of eviction for accesses to $b_i$; and (ii) determine accesses that can be affected if a remote access occurs at a $U_x$ - identifying ones that are prone to cache misses across different URs.
To achieve this, 
this section first identifies references of each $b_i \in U_x$ with different ages and access numbers (Sec.~\ref{sec:crm-ref}). 
Then, based on the temporal order between URs, a sequence of contention regions (CR) is constructed, where each CR contains the references that can be affected by a remote access (Sec.~\ref{sec:crm-region}). Notations introduced in this section are summarised in Tab.~\ref{tab:notations_crm}.

\begin{figure}[t]
\centering
\includegraphics[width=.9\linewidth]{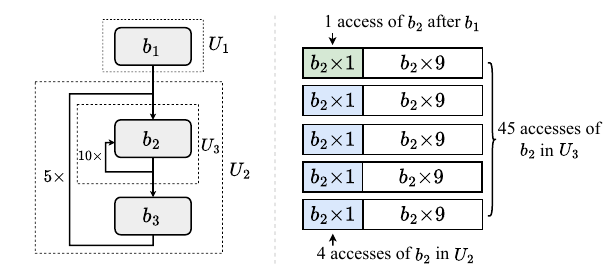}
\caption{An example of different references of $b_2$. The ages of $b_2$ are $age(b_{2},U)=\infty$, $age(b_{2},U_2)=1$ and $age(b_{2},U_3)=0$.}
\label{fig:sre}
\vspace{-7pt}
\end{figure}

\subsection{Constructing Memory References}
\label{sec:crm-ref}

We start by identifying different ages of a block, especially when it is accessed in nested URs.
For a block $b_i$ in a $U_x$ with $\cnt_x > 1$, the age of $b_i$ can differ significantly between its first access (\ie, the scope of the program) and the subsequent iterations (\ie, the scope within $U_x$)~\cite{huynh2011scope}.
First, for the scope of the program, the age of $b_i$ (denoted as $age(b_i, U)$) is defined as \add{the maximum reuse distance among all memory access traces in the preceding URs}, which can be computed using the Must/May analysis~\cite{alt1996cache}.
Then, the age of $b_i$ within the scope of $U_x$ (denoted as $age(b_i, U_x)$) is computed based on the repeating access traces inside $U_x$, which can be obtained based on the conflict set by the persistence analysis~\cite{stock2019cache}.
\add{In Fig.~\ref{fig:sre}, as $b_2$ resides in nested URs, it is accessed from the scope of the program ($U$), $U_2$ and $U_3$ with different ages.}

Based on the above, Def.~\ref{def:sre} describes a memory reference $r^k$ to a block $b_i$, which is characterised by a block address $\omega^k$, an access number $\delta^k$, and the age $age^k$ of these accesses. Accordingly, $b_2$ in Fig.~\ref{fig:sre} has three references, each of which has a specific age and an associated access number, \ie, $r^1 = \{ \omega_2, 1, \infty \}$, $r^2 = \{ \omega_2, 4, 1 \}$, and $r^3 = \{ \omega_2, 45, 0 \}$. 
\begin{definition}\label{def:sre}
A memory reference to $b_i$ is defined as $\ctx^k=(\omega^k=\omega_i, \cnt^k, age^k)$, indicating $\cnt^k$ accesses to the block address $\omega^k$ with an age of $age^k$.
\end{definition}

\begin{algorithm}[t]
\caption{Construction of memory references.}
\label{alg:sre}
\KwIn{$\forall U_x\in\lambda$;}
$R^x = \varnothing$;\\

\For{\normalfont{each} $ b_i \in U_x$}{
         {\footnotesize\ttfamily/* {The reference of the first access} */}\\
         $\omega^k = \omega_i$;~~~$\delta^k = 1$;~~~$age^k=age(b_i,U)$;\\
        $\ctxset^x =\ctxset^x \cup  \{r^k\}$;\\
    

    \For{\normalfont{each} $U_l \in U(b_i)$}{
        
        {\footnotesize\ttfamily/* {Access count in scope of $U_l$} */}\\
        $\delta^k= \Pi_{U_{o}\in anc(U_l)} \cnt_{o} \times (\cnt_l-1) $;\\

        {\footnotesize\ttfamily/* {The reference of accesses in $U_l$} */}\\
        \If{$\cnt^k > 0$}{
            $\omega^k=\omega_i$;~~~$age^k=age(b_i,U_l)$;\\
            $\ctxset^x =  \ctxset^x \cup \{r^k \}$;\\
        }

    }
}


\Return $\ctxset^x$;\\
\end{algorithm}

Let $R^x$ denote the set of $r^k$ in a $U_x$, Alg.~\ref{alg:sre} iterates over every out-most $U_x \in \lambda$ and identifies the references of $b_i \in U_x$.  
For a $b_i$, the reference of the first access in $U_x$ is constructed, which has an age of $age(b_i, U)$ in the scope of the program (lines 4-5).
Then, we identify the remaining references of $b_i$ by examining (i) the set of nested URs that contain $b_i$ and (ii) the execution numbers of these URs (lines 6-14). 
Function $U(b_i)$ at line 6 returns the nested URs that contain $b_i$, including $U_x$ itself. 
For each $U_l$ in $U(b_i)$ (line 6), the number of accesses to $b_i$ in the scope of $U_l$ (except its first access) is computed by $\cnt^k= \Pi_{U_{o}\in anc(U_l)} \cnt_{o} \times (\cnt_l-1)$, where $\Pi_{U_{o}\in anc(U_l)} \cnt_{o}$ gives the executions number of $U_l$ invoked by its outer URs and $\cnt_l-1$ indicates the iterative accesses to $b_i$ in $U_l$.
For instance, $b_2$ is accessed $\cnt_2 \times (\cnt_3 - 1)=45$ times in the scope of $U_3$ with $age(b_2, U_3) =0 $.
Then, if $\cnt^k > 0$, a $r^k$ to $b_i$ is constructed with $age(b_i, U_l)$ (at lines 10-13).

Finally, Alg.~\ref{alg:sre} terminates with $R^x$ returned for every $U_x \in \cp$. With the above, all the references to every $b_i$ accessed in $\lambda$ can be obtained. 
The time complexity of Alg.~\ref{alg:sre} is quadratic, as constructing references of every $b_i$ in $\lambda$ across scopes requires examining all the nested URs that contain $b_i$.

\subsection{Constructing the Contention Regions}
\label{sec:crm-region}

With the references constructed for each $U_x \in \lambda$, 
we then identify the set of references in which their ages can be affected by a remote access occurring at every $U_x \in \lambda$.
\add{As illustrated in Fig.~\ref{fig:CR}, $r^1$ can be affected by a remote access at $U_2$ or $U_3$, whereas $r^2$ can be affected by a remote access at $U_3$ or $U_4$. Accordingly, a remote access at $U_3$ can affect both $r^1$ and $r^2$ simultaneously.}
Thus, to identify references that can be affected by a remote access occurring at every $U_x \in \lambda$, we first determine the sequence of URs at which $r^k$ can be affected. 
Note, only the references with $age^k < \asso$ will be considered, as ones with $age^k \geq \asso$ will inevitably miss the shared cache regardless of inter-core contention.

\begin{figure}
\centering
\includegraphics[width=.95\linewidth]{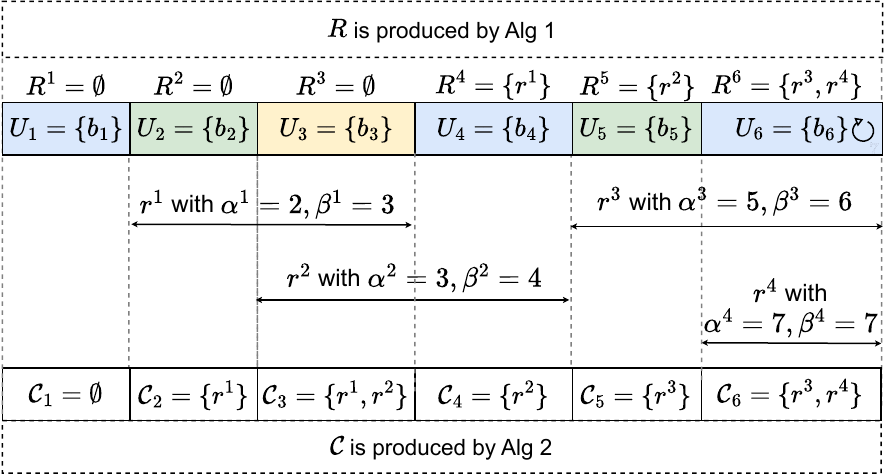}
\caption{\add{The construction of CRs \textit{(blocks in the same colour have the same address; $U_6$ is a loop that accesses $b_6$ iteratively; references with $age^k\geq\asso$ are omitted)}.}}
\label{fig:CR}
\vspace{-7pt}
\end{figure}


\add{For a reference $r^k$, notations $\alpha^k$ and $\beta^k$ denote the indices of the first and the last URs in the segment. That is, $\alpha^k$ and $\beta^k$ identify a contiguous UR segment between $r^k$ and the previous access to $r^k$'s address, in which the access to other addresses would increase the age of $r^k$.}
Intuitively, $\alpha^k$ and $\beta^k$ of a $r^k$ in $U_x$ can be identified based on the UR containing the previous access to $\omega^k$ and $U_x$, respectively. 
However, for a $U_x$ with $\delta_x > 1$, it may contain multiple references that access $\omega^k$ with different ages.
To cope with such situations, $\alpha^k$ and $\beta^k$ are decided based on the following three cases. Function \add{$\check{U}(r^k)$} returns the index of the UR that contains the memory reference immediately following the previous access to $\omega^k$.
\begin{itemize}
\item If $\cnt_x =1$, it indicates that $U_x$ only contains one $r^k$ that access $\omega^k$ once (\ie, $U_x$ does not contain any loops). In this case, the previous access to $\omega^k$ is allocated in a UR before $U_x$ in $\lambda$, and hence, $\alpha^k = \check{U}(r^k)$ and $\beta^k = x-1$.
\item If $\cnt_x > 1$ and $\cnt^k = 1$, it means $r^k$ could be the first access to $\omega^k$ in $U_x$, where the previous access to $\omega^k$ resides in a UR before $U_x$. \add{In addition, as $U_x$ may contain references that access the cache before $r^k$, a remote access occurs at $U_x$ can affect $r^k$}. Hence, $\alpha^k = \check{U}(r^k)$ and $\beta^k = x$.
\item If $\cnt_x > 1$ and $\cnt^k > 1$, it indicates $r^k$ is not the first access to $\omega_k$ in $U_x$. Hence, a remote access can affect $r^k$ only if it occurs at $U_x$, and therefore, $\alpha^k = x$ and $\beta^k = x$.
\end{itemize}

Based on the above, Eq.~\ref{eq:alpha} and~\ref{eq:beta} present the computations of $\alpha^k$ and $\beta^k$ for a $r^k$ in $U_x$, respectively. Note, if $\cnt_x=1$, then $\cnt^k=1$ as there exists no loop in $U_x$. \add{Example~\ref{exa:alpha_beta} illustrates the computations of $\alpha^k$ and $\beta^k$ based on references in Fig.~\ref{fig:CR}.}
\begin{equation} \label{eq:alpha}
\alpha^k =
\begin{cases}

\check{U}(r^k), & \text{if}~  \cnt^k=1 \\
x, & \text{if}~ \cnt^k>1 \\
\end{cases}
\end{equation}
\begin{equation} \label{eq:beta}
\beta^k = 
\begin{cases}
x-1, &  \text{if}~  \cnt_{x}=1 \\
x, & \text{if}~  \cnt_{x}>1\\
\end{cases}
\end{equation}

\begin{example}\label{exa:alpha_beta}
\normalfont
\add{
For $r^1$ to $r^4$ in Fig.~\ref{fig:CR}, $r^1$ and $r^2$ correspond to $U_4$ and $U_5$, respectively; while both $r^3$ and $r^4$ are references to $b_6$ in $U_6$ based on Alg.~\ref{alg:sre}. Since $U_4$ does not contain any loops (\ie, $\delta_4 = 1$), $r^1$ falls into the first case described above.
In addition, as the previous access of the same address occurs in $U_1$ (coloured in blue), $\alpha^1 = \check{U}(r^1) = 2$ and $\beta^1 = 4 - 1 = 3$. Similarly, we have $\alpha^2 = 3$ and $\beta^2 = 4$. 
For $r^3$, since $U_6$ contains a loop (\ie, $\delta_6 > 1$) and $\delta^3 = 1$ (\ie, the first access to $b_6$), it falls into the second case. In addition, as the previous access to the address occurs in $U_4$, so $\alpha^3 = \check{U}(r^3) = 5$ and $\beta^3 = 6$ accordingly.
Finally, as $\delta^4 > 1$ and $\delta_6 > 1$ for $r^4$, we have $\alpha^4 = \beta^4 = 6$ based on the third case. 
}
\end{example}

\begin{algorithm}[t]
\caption{Construction of contention regions.}
\label{alg:build_CR}
\KwIn{$\forall R^x \in U_x, \forall U_x \in \lambda$;}
$\mathcal{C} = \varnothing$;\\

{\footnotesize\ttfamily/* {The region where $r^k$ can be affected.} */}\\
\For{\normalfont{each} $r^k \in R^x, \forall U_x \in \lambda$}{
    compute $\alpha^k$ and $\beta^k$ by Eq.~\ref{eq:alpha} and~\ref{eq:beta};\\
          
}
\For{\normalfont{each} $U_x \in \lambda$}{ 
    $\mathcal{C}_{s} = \varnothing$;\\

{\footnotesize\ttfamily/* {Identify the references in a CR.} */}\\
    \For{\normalfont{each} $\ctx^k  \in \ctxset^x, \forall U_x \in \lambda$}{
        \If{$\alpha^k \leq x  \leq \beta^k ~\wedge~ age^k < \asso$}{
            $\mathcal{C}_{s} = \mathcal{C}_{s} \cup \{\ctx^k\}$;\\
        }
    }

    $\mathcal{C} =\mathcal{C}  \cup \{ \mathcal{C}_s \} $;\\
}

\Return $\mathcal{C}$;\\

\end{algorithm}

Following the above, if a remote access occurs at $U_x$, a contention region can be constructed (denoted as $\CR_s$) that identifies the references that can be affected by examining the value of their $\alpha^k$ and $\beta^k$, \ie, ones satisfying $\alpha^k \leq x \leq \beta^k$. \add{For instance, $\mathcal{C}_3=\{r^1,r^2\}$ is constructed for $U_3$.}
Then, by traversing all $U_x \in \lambda$, a sequence of CRs (defined as $\mathcal{C}$) can be obtained if a remote access occurs at every $U_x$.

Alg.~\ref{alg:build_CR} presents the construction process of contention regions for a UR path $\lambda$ in $\tau$.
The algorithm starts by computing the $\alpha^k$ and $\beta^k$ of every $r^k$ using Eq.~\ref{eq:alpha} and Eq.~\ref{eq:beta} (lines 3-5).
Then, for every $U_x \in \lambda$, a $\CR_s$ is constructed by identifying the $r^k$ that could be affected by a remote access occurred at $U_x$, \ie, $\alpha^k \leq x \leq \beta^k \wedge age^k < \kappa$ (lines 6-15).  
Finally, the algorithm terminates with the $\CR$ returned (line 16).
As with Alg.~\ref{alg:sre}, the time complexity of Alg.~\ref{alg:build_CR} is quadratic. 
We note that several optimisations can be applied to reduce the number of CRs in $\CR$, by removing the empty CRs \add{(\eg, $\CR_1$ in Fig.~\ref{fig:CR})} and merging two consecutive identical CRs.
This effectively speeds up the constructed analysis, as described in Sec.~\ref{sec:ana_cost}.

To this end, we identify the set of $r^k$ that can be affected by a remote access occurring at a $U_x$. Similar to URs, $\CR_s \in \CR$ follows a strict temporal order, enabling the identification of contention scenarios that are mutually exclusive in time.



\newcommand{\nkx}{n^k_x}
\newcommand{\nk}{n^k}

\newcommand{\swx}{\sum_{r^k \in R_s(\omega)} n^k_x}
\newcommand{\sw}{N^s(\omega)}

\newcommand{\sx}{N^s_x}
\newcommand{\s}{N^s}

\vspace{-1pt}
\section{Analysing Inter-core Contention of a CR}
\label{sec:contention}

\begin{figure*}[t]
\centering
\subfigure[An illustrative example of $\CR_1$ and $U'$.]
{\label{fig:queue_a}
{\includegraphics[width=0.334\textwidth]{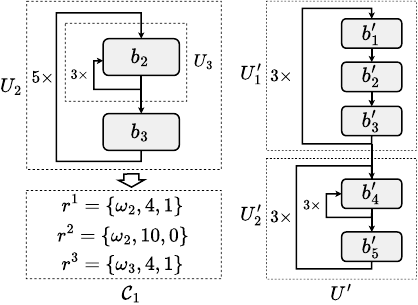}}}
\hspace{12pt}
\subfigure[A worst-case scenario.]
{\label{fig:queue_b}
{\includegraphics[width=.212\textwidth]{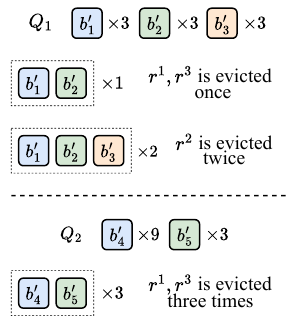}}}
\hspace{12pt}
\subfigure[Computation based on the aggregated queue.]
{\label{fig:queue_c}
{\includegraphics[width=.3531\textwidth]{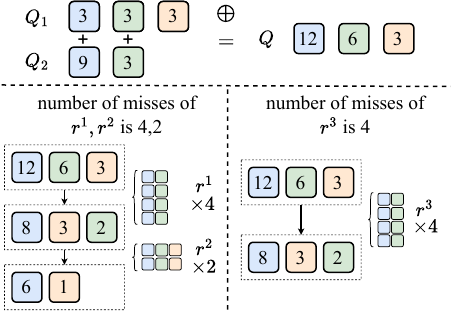}}}
\vspace{-2pt}
\caption{An illustrative example of the worst-case inter-core cache interference and the upper bound computations based on the aggregated queue \textit{(the cache associativity is $\asso = 3$; \add{the addresses in each queue are different, marked by different colours)}}.}
\vspace{-7pt}
\label{fig:queue}
\end{figure*}

With $\mathcal{C}_s\in\mathcal{C}$, this section analyses the cache misses number of references in $\CR_s$ imposed by a sequence of out-most URs in a $\lambda'$ of $\tau'$ (denoted as $U' \subseteq \lambda'$).
To achieve this, we first examine the situation where the references in $\CR_s$ are interfered by one  $U'_x \in U'$ (Sec.~\ref{sec:contention-oneU}), and then compute the maximum number of cache misses of $\CR_s$ imposed by $U'$ (Sec.~\ref{sec:contention-allU}). Notations introduced in this section are summarised in Tab.~\ref{tab:notations_contention}.

\begin{table}[t]
\vspace{4pt}
\centering
\caption{Notations introduced in Sec.~\ref{sec:contention}.}
\label{tab:notations_contention}
\renewcommand{\arraystretch}{1.1}
\begin{tabular}{p{.15\columnwidth}m{.74\columnwidth}}
\hline
\textbf{Notation} & \textbf{Description} \\
\hline

$\rho^k$ & The minimal number of unique block addresses accessed remotely that can cause $r^k$ to incur a cache miss. \\

$R_s(\omega)$ & The set of $r^k \in \CR_s$ with $\omega^k=\omega$ ordered non-decreasingly by $\rho^k$. \\
\hline
$Q_x$ & A list of access numbers to each unique address in $U'_x$. \\
$Q_x^v$ &  The $v^{\text{th}}$ element in $Q_x$.\\

$\Phi(Q_x, \{r^k\})$ & 
An operator that returns the number of iterations by taking one access from the first
$\rho^k$ addresses in $Q_x$ iteratively. \\
\hline
$\nkx$ & The number of misses of $r^k$ imposed by a $U'_x$. \\
$\nk$ & The number of misses of $r^k$ imposed $U'$. \\
$N_s^*(w)$ & Number of misses in $R_s(\omega)$ due to carry-on contention.\\
$N(\CR_s,U')$ & The number of misses in $\CR_s$ imposed by $U'$. \\
\hline
\end{tabular}
\end{table}

\subsection{Inter-core Contention Imposed by an Interfering UR}\label{sec:contention-oneU}

We first compute the number of cache misses of a $r^k \in \CR_s$ imposed by a $U'_x$ (denoted as $n^k_x$), which serves as the key for the constructed analysis. 
For a $r^k \in \CR_s$, Lemma~\ref{lem:miss} describes the case where $r^k$ can incur a cache miss due to the remote block accesses in an interfering UR $U'_x$.
\begin{lemma} \label{lem:miss}
For a $r^k$, it can incur a cache miss if there exists \add{one remote access to each of} $\rho^k = \kappa - age^k$ unique addresses.
\end{lemma}

The proof of Lemma~\ref{lem:miss} is straightforward, and hence is omitted. 
Accordingly, to analyse the number of cache misses incurred by the $\delta^k$ accesses of $r^k$,
Eq.~\ref{eq:q} first constructs a \textit{non-zero, non-increasing} access queue that contains the number of accesses to every unique address in $U'_x$ (\ie, $\forall \omega\in\Omega'_x$). For instance, $Q_1 = \{3, 3, 3\}$ and $Q_2 = \{9, 3\}$ for $U'_1$ and $U'_2$ in Fig~\ref{fig:queue_a}. The $v^{\text{th}}$ element in $Q_x$ is denoted as $Q_x^v$.
\begin{equation}\label{eq:q}
Q_x = \{ \sum_{b'_i \in U'_x \land \omega_i=\omega}  \Pi_{U'_{l}\in U(b'_i)} \cnt_l ~|~ \forall \omega \in \Omega'_x \}
\end{equation}

Based on Lemma~\ref{lem:miss} and $Q_x$, \add{the
$\nkx$ is calculated as the number of iterations being performed by an iterative operation $\Phi(Q_x, \{r^k\})$, which takes one access from the first $\rho^k$ elements in $Q_x$ at each iteration (\ie, one iteration indicates a miss for $r^k$).} 
The $\Phi(Q_x, \{r^k\})$ terminates when $\nkx = \delta^k$ or $|Q_x| < \rho^k$.
For $r^1$ in Fig.~\ref{fig:queue_a}, $n^1_1$ is computed by taking $\{1, 1\}$ from first two elements in $Q_1=\{3,3,3\}$ iteratively, where $Q_1$ becomes $\{3, 2, 2\}$ and $\{ 2,2,1 \}$ after the first and the second iteration, respectively. Finally, the computation ends after the fourth iteration with $n^1_1=4$, where $Q_1=\{1,0,0\}$.
Theorem~\ref{the:n} provides a theoretical bound for $\nkx$.

\begin{theorem}\label{the:n}
For a $r^k$ and an interfering $U'_x$, it follows that $\nkx \leq \delta^k$ and $\nkx \times \rho^k \leq \sum_{v=1}^{|Q_x|} \min(Q_x^v, \nkx)$.
\end{theorem}
\begin{proof}
First, the condition $\nkx \leq \delta^k$ holds as $r^k$ has $\delta^k$ accesses.
Second, based on Lemma~\ref{lem:miss}, a cache miss of $r^k$ is caused by one access to $\rho^k$ addresses in $U'_x$.
Given number of accesses to each address in $U'_x$ (\ie, $Q_x$), the maximum number of accesses in $U'_x$ that can interfere $r^k$ is $\sum_{v=1}^{|Q_x|} \min(Q_x^v,\nkx)$. Hence, the theorem follows.
\end{proof}



The above computes the cache misses of a $r^k$.
For two references $r^a$ and $r^b$ in $\CR_s$, different interference scenarios can occur depending on whether they access the same address, as shown in Lemma~\ref{lemma:one_miss} and~\ref{lemma:all_miss}. 
Let $R$ denote a set of $r^k$ ordered non-decreasingly by $\rho^k$, $\Phi(Q_x, R)$ applies the above iterative process on each $r^k \in R$ in order, extracting accesses from the same $Q_x$ until no further accesses can be taken. 

\begin{lemma}\label{lemma:one_miss}
For $r^a$ and $r^b $ in $ \CR_s$ with $\omega^a = \omega^b$, the maximum possible number of cache misses due to $U'_x$ is $n^a_x + n^b_x = \Phi(Q_x, \{r^a, r^b\})$.
\end{lemma}
\begin{proof}
If $r^a$ and $r^b$ access the same address,
one access in $U_x'$ can only interfere with either $r^a$ or $r^b$, where the age of the block is updated immediately after the access.
Thus, with $\Phi(Q_x, R)$ defined, $ n^a_x + n^b_x = \Phi(Q_x, \{r^a, r^b\})$ follows.
\end{proof}

\begin{lemma}\label{lemma:all_miss} 
For $r^a$ and $r^b$ in $\CR_s$ with $\omega^a \neq \omega^b$, the maximum number of cache misses imposed by $U'_x$ is $n^a_x + n^b_x = \Phi(Q_x, \{r^a\}) + \Phi(Q_x, \{r^b\})$.
\end{lemma}
\begin{proof}
If $\omega_a \neq \omega_b$, the blocks of $r^a$ and $r^b$ are \add{present} in the shared cache at the same time. Therefore, a remote access in $U'_x$ can affect the age of both blocks simultaneously, and hence, the lemma follows.
\end{proof}

Following Lemma~\ref{lemma:one_miss}, 
let $R_s(\omega)= \{ r^k | \omega^k = \omega, \forall r^k \in \CR_s \}$ denote the references in $\CR_s$ that have the same block address,
the number of cache misses that $R_s(\omega)$ can incur due to the interference of $U'_x$ is computed as Eq.~\ref{eq:swx}. 
\begin{equation}\label{eq:swx}
\swx =  \Phi(Q_x, R_s(\omega))
\end{equation}

Furthermore, the maximum number of cache misses of a $\CR_s$ imposed by a $U'_x$ can be computed as $\sum_{\omega \in \Omega_s}\swx$, based on Eq.~\ref{eq:swx} and Lemma~\ref{lemma:all_miss}.

\subsection{Inter-core Contention Imposed by Multiple Interfering URs} \label{sec:contention-allU}

The above produces the number of misses in a $\CR_s$ imposed by a $U'_x$. However, for a sequence of interfering URs (\ie, $U'$), simply applying $\Phi(Q_x, R_s(\omega))$ over every $U'_x \in U'$ cannot yield a correct bound, as illustrated by Example~\ref{example:worst_wrong}.
\begin{example}\label{example:worst_wrong}
\normalfont
\add{ Fig.~\ref{fig:queue_a} shows a CR of $\tau$ (\ie, $\CR_1 = \{r^1, r^2, r^3\}$) and the interfering URs $U' = \{U'_1, U'_2\}$.}
A direct application of $\Phi(Q_x, R_s(\omega))$ on $U'_1$ and $U'_2$ would lead to the following results.
\add{First, as $r^1$ and $r^2$ have the same address while $r^3$ has a different address, we have $n_x^1 + n_x^2 + n_x^3 = \Phi(Q_x, \{r^1, r^2\}) + \Phi(Q_x, \{r^3\})$ based on Eq.~\ref{eq:swx}.}
For $U'_1$, we have $\Phi(Q_1, \{r^1,r^2\}) = 4$ and $\Phi(Q_1, \{r^3\}) = 4$. As for $U'_2$, it cannot impose any cache miss as (i) every access in $r^1$ and $r^3$ is already accounted for and (ii) $U'_2$ cannot cause a cache miss on $r^2$ since $|Q_2|< \rho^2$. This leads to eight cache misses.
However, as shown in Fig.~\ref{fig:queue_b}, a worst-case scenario occurs where $U'_1$ causes one miss for both $r^1$ and $r^3$ while imposing two misses for $r^2$. Then, $U'_2$ imposes three misses for both $r^1$ and $r^3$. This leads to ten cache misses in total.
\end{example}

As illustrated by the example, the worst case occurs when $r^1$ and $r^3$ are interfered by $U'_1$ and $U'_2$ in a collaborative manner.
Let $\overline{\nkx}$ denote the number of misses of $r^k$ caused by $U'_x$ that leads to the maximum $\sum_{U'_x\in U'}\sum_{r^k \in R_s(\omega)} \overline{\nkx}$, Eq.~\ref{eq:bound_true} provides the theoretical bound for the total number of misses.
Therefore, finding every $\overline{\nkx}$ is an optimisation problem which can be solved by integer linear programming. 
Yet, this can bring significant computation cost that undermines the application of the analysis, especially when $|\CR_s|$ and $|U'|$ are high.
\begin{equation}\label{eq:bound_true}
\small
\sum_{U'_x\in U'}\sum_{r^k \in R_s(\omega)}\overline{\nkx} \times \rho^k
\;\le\;
\sum_{U'_x\in U'}\sum_{v = 1}^{|Q_x|} \min(Q^v_x, \sum_{r^k \in R_s(\omega)} \overline{\nkx})
\end{equation}

However, we notice that as the URs in $U'$ are executed sequentially in order, 
such collaborative contention effects from different $U'_x$ can be effectively computed by aggregating the $Q_x$ of every $U'_x \in U'$ using an element-wise addition, denoted as $Q = \bigoplus_{U'_x \in U'} Q_x$. For instance, $Q = Q_1 \oplus Q_2 = \{12, 6, 3\}$ in Fig.~\ref{fig:queue_c}. 
Below we first present the computations based on $Q$ and then prove the correctness of the analysis.

Based on the aggregated $Q$ from every $U'_x \in U'$, 
the number of cache misses incurred by $R_s(\omega)$ due to $U'$ can be computed as $\Phi(Q, R_s(\omega))$.  
This effectively abstracts away the computation order of every $U'_x \in U'$, forming a holistic approach that analyses the number of misses that $r^k$ can incur due to $U'$. 
Let $\nk$ denote the number of misses of $r^k$ obtained by $\Phi(Q, R_s(\omega))$ for every 
$r^k \in R_s(\omega)$, Eq.~\ref{eq:bound_upper} provides the  bound on the maximum value of $\sum_{r^k\in R_s(\omega)} \nk$. \add{Example~\ref{ex:aggregated_queue} illustrates the computations of $n^k$ according to $\Phi(Q, R_s(\omega))$.}
\vspace{-2pt}
\begin{equation}\label{eq:bound_upper}
\sum_{r^k\in R_s(\omega)} \nk \times \rho^k \leq \sum_{v=1}^{|Q|} \min \Big(Q^v, \sum_{r^k\in R_s(\omega)} \nk \Big)
\end{equation}

\begin{example}\label{ex:aggregated_queue}
    \normalfont
\add{
Fig.~\ref{fig:queue_c} illustrates the computations of $n^k$ based on the aggregated queue $Q = Q_1 \oplus Q_2 = \{12, 6, 3\}$, where $n^1 + n^2 + n^3 = \Phi(Q, \{r^1, r^2\}) + \Phi(Q, \{r^3\})$.
For $r^1$ and $r^2$ (with an age of 0 and 1, respectively), $r^1$ is examined first as only two unique interfering blocks can cause an eviction. Thus, $\Phi(Q,\{r^1,r^2\})$ is performed that subtracts one from the first two elements in $Q$ iteratively, where $Q = \{12, 6, 3\}$ is updated to $Q = \{8, 3, 2\}$ after 4 iterations, indicating that all four accesses of $r^1$ can be evicted. 
Then, a similar process is applied for $r^2$, which requires three unique interfering blocks to be evicted. Therefore, only 2 iterations can be performed, with $Q = \{8, 3, 2\}$ updated to $Q = \{6, 1\}$ (\ie, $r^2$ can be evicted twice). Hence, we have $\Phi(Q,\{r^1,r^2\})=6$. Following the similar process, we have $\Phi(Q, \{r^3\})=4$. Therefore, the total number of cache misses is $n^1 + n^2 + n^3 = 10$.}
\end{example}

Theorem~\ref{the:bound_sw} describes that $\sum_{r^k\in R_s(\omega)} \nk=\Phi(Q, R_s(\omega))$ bounds the number of misses in a $\CR_s$ due to each $U'_x \in U'$.
\begin{theorem}\label{the:bound_sw}
For a given $R_s(\omega)$ and $U'$, it follows that
$\sum_{r^k\in R_s(\omega)} \nk \geq   \sum_{U'_x \in U'} \sum_{r^k \in R_s(\omega)} \overline{\nkx}$.
\end{theorem}
\begin{proof}
We prove this theorem via contradiction by assuming that $\sum_{r^k\in R_s(\omega)} \nk < \sum_{U'_x \in U'} \sum_{r^k \in R_s(\omega)}  \overline{\nkx}$. In this case, the condition in Eq.~\ref{eq:bound_upper} no longer holds if $ \nk$ is replaced by $\sum_{U'_x \in U'} \overline{\nkx}$, because $\sum_{r^k \in R_s(\omega)} \nk$ is the maximum value that satisfies the condition. This leads to the following computations, where the last derivation holds as ``$\min\sum \geq \sum\min$".
\begin{equation*}
\small
\begin{split}
 \sum_{U'_x \in U'} \sum_{r^k \in R_s(\omega)}\overline{\nkx} \times &\rho^k
 > \sum_{v = 1}^{|Q|} \min(Q^v, \sum_{U'_x \in U'}\sum_{r^k \in R_s(\omega)}  \overline{\nkx}) \\
&~~~ = \sum_{v=1}^{|Q|} \min( \sum_{U'_x \in U'} Q^v_x, \sum_{U'_x \in U'}\sum_{r^k \in R_s(\omega)} \overline{\nkx}) \\
&~~~ \ge \sum_{U'_x \in U'} \sum_{v=1}^{|Q_x|} \min(Q^v_x, \sum_{r^k \in R_s(\omega)}  \overline{\nkx} )
\end{split}
\end{equation*}
\normalsize

However, this contradicts with relationship depicted in Eq.~\ref{eq:bound_true} for $\sum_{U'_x\in U'}\sum_{r^k \in R_s(\omega)} \overline{\nkx}$. Thus, the situation of
$\sum_{r^k\in R_s(\omega)} \nk < \sum_{U'_x \in U'}\sum_{r^k \in R_s(\omega)}  \overline{\nkx}$ does not exist, and therefore, the theorem holds.
\end{proof}

\textbf{Carry-on Contention.} The above presents an efficient approach that considers the collaborative contention effects from $U'_x \in U'$ on $\CR_s$, in which an access of $r^k \in \CR_s$ can be evicted by any $U'_x \in U'$. However, it is performed by assuming the contentions caused by the $U'_x \in U'$ are independent of each other. 
By doing so, no cache misses would be identified for a $r^k$ if every $U'_x$ satisfies $\rho^k > |Q_x|$.
In reality, 
although $U'_x$ is not able to cause an eviction directly, 
it can increase the age of $r^k$ by $|Q_x|$.
Thus, for next UR (say $U'_{x+1}$), it only needs to access $\rho^k - |Q_{x}|$ unique addresses to impose a cache miss for $r^k$.
Example~\ref{eg:carry} shows the carry-on contention imposed from the consecutive interfering URs on $r^1$ and $r^3$ in Fig.~\ref{fig:queue_a}.

\begin{figure}
    \centering
    \includegraphics[width=0.7\linewidth]{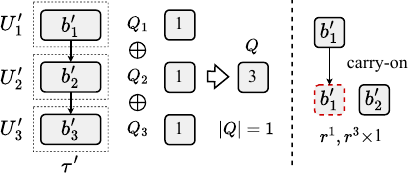}
    \caption{A carry-on contention imposed from $U'_1$ and $U'_2$ on $r^1$ and $r^3$ from Fig.~\ref{fig:queue_a} with $age^1 = age^3 = 1$ and $\kappa=3$.}
    \label{fig:carry}
    \vspace{-7pt}
\end{figure}

\begin{example}\label{eg:carry}
\normalfont
In Fig.~\ref{fig:carry}, neither $U'_x$ contain a sufficient number of addresses that can cause a cache miss on $r^1$ and $r^3$ directly. However, $U'_1$ can increase the age of $r^1$ and $r^3$ by one as $|Q_1|=1$, \ie, $age^1 = age^3 = 2$ now.
Then, $U'_2$ can impose one cache miss for these references by one block access.
However, with the aggregated queue $Q = \{3\}$, the computation based on $Q$ yields zero cache miss since $|Q| = 1$, which fails to consider such carry-on contention effects between consecutive URs.
\end{example}

However, as the carry-on contention can occur at most once between two consecutive URs in $U'$, the maximum number of cache misses caused by such contention effects on $R_s(\omega)$ can be effectively bounded as $N_s^*(w) =min\{\sum_{r^k \in R_s(\omega)} (\delta^k - \nk),|U'|-1\}$, where $(\delta^k-n^k)$ is the number of accesses in $r^k$ that have not incurred a cache miss yet.
We note that a finer analysis of the carry-on contention is possible, \eg, by examining the unique addresses in each UR. However, we decide to halt here to avoid over-complicated computations.

Based on the above, the number of cache misses incurred by the references in $\CR_s$ due to the contention from $U'$ can be obtained (denoted as $N(\CR_s,U')$), as shown in Eq.~\ref{Eq:totalbound}.
\begin{equation}\label{Eq:totalbound}
N(\CR_s,U')=
\sum_{\omega \in \Omega_s} \Big(
\Phi(Q, R_s(\omega))+ N_s^*(w) \Big)
\end{equation}

This concludes the proposed cache contention analysis for the memory references in a $\CR_s$ when facing a sequence of interfering URs. By examining the block address of each references, the proposed analysis identifies finer contention scenarios that are overlooked by existing methods.
In addition, a precise quantification of the number of cache misses is provided based on the ages and access numbers of blocks, leading to tighter analytical results (see Sec.~\ref{sec:result}).


\section{Complete Inter-core Cache Analysis}\label{sec:multi}

Based on the analysis constructed for a $\CR_s \in \CR$ and a sequence of interfering URs $U'$ (\ie, $N(\CR_s, U')$), this section presents the complete inter-core cache contention analysis that produces the maximum number of misses of a task imposed by the interfering ones.
We first describe the computations for a task $\tau$ and an interfering task $\tau'$ (Sec.~\ref{sec:multi-dual}). 
Then, we show that the proposed analysis can effectively support more generalised scenarios which contain a multi-level set-associative cache hierarchy shared among multiple cores (Sec.~\ref{sec:multi-multi}).
Notations introduced in this section are summarised in Tab.~\ref{tab:notations_multi}. 

\subsection{Bounding the inter-core latency of $\tau$ imposed by $\tau'$}\label{sec:multi-dual}

As described, both $\tau$ and $\tau'$ can be represented by a CFG that contains a set of execution paths. 
Given the sequence of CRs in a path of $\tau$ (\ie, $\CR$) and the sequence of URs in a path of $\tau'$ (\ie, $\lambda'$), the number of cache misses of $\CR$ imposed by $\lambda'$ can be computed as $\sum_{\CR_s \in \CR} N(\CR_s, \lambda')$ based on the analysis in Sec.~\ref{sec:contention}.
However, this can lead to overly pessimistic results due to the following two reasons. First, both the CRs in $\CR$ and the URs in $\lambda'$ follow a strict temporal order during execution. Hence, contentions that are mutually exclusive in time can be accounted for without considering their partial order~\cite{zhang2022precise}. Second, For the CRs in $\CR$, they may contain the same $r^k$ (\eg, all $\CR_2$, $\CR_3$ and $\CR_4$ contain the same $r^3$ in Fig.~\ref{fig:CR}) with the number of misses being accounted for multiple times. 

To address the above, we construct an analysis that identifies contention scenarios that satisfy the partial order between $\CR$ and $\lambda'$.
In addition, when examining every $\CR_s \in \CR$, we identify the $r^k$ in which every access incurs a cache miss, avoiding the repetitive computations on the same $r^k$ across different CRs.

\begin{table}[t]
\centering
\caption{Notations Introduced in Sec.~\ref{sec:multi}.}
\label{tab:notations_multi}
\renewcommand{\arraystretch}{1.1}
\begin{tabular}{p{.12\columnwidth}m{.78\columnwidth}}
\hline
\textbf{Notation} & \textbf{Description} \\
\hline

$\check{x}_s$~/~$\hat{x}_s$ & Index of the first / last UR in a UR sequence that affects $\CR_s$. \\
${V}_{s,\hat{x}}^h$ & A state of $\CR_s$ with an index of $h$ when $\{\CR_1,\cdots,\CR_s\}$ is interfered by $\{U'_1,\cdots,U'_{\hat{x}}\}$. \\
$\zeta_s^h$ & Set of $r^k$ in $h^{\text{th}}$ state of $\CR_s$ that every access incurs a miss. \\
$V_s$ & The set of all states for $\CR_s$. \\
\hline
\end{tabular}
\vspace{-7pt}
\end{table}

For each $\CR_s$, it can incur contentions from a sequence of consecutive URs in $\lambda'$. Let $\check{x}_s$ and $\hat{x}_s$ denote the indices of the first and the last UR of the UR sequence that interfere with $\CR_s$, Lemma~\ref{lem:boundary} describes the relationship of $\check{x}_s$ and $\hat{x}_s$ across different $\CR_s$ based on the partial order.
\begin{lemma}\label{lem:boundary}
For $\CR_{s-1}$ and $\CR_s$ in $\CR$, it follows $\hat{x}_{s-1} \le \check{x}_s \le \hat{x}_s$. 
\end{lemma}

A similar proof of Lemma~\ref{lem:boundary} can be found in~\cite{zhang2022precise}, and hence, is omitted in this paper.
Given the constraints of $\check{x}_s$ and $\hat{x}_s$ described in Lemma~\ref{lem:boundary}, 
the problem of finding the maximum number of cache misses of $\CR$ caused by $\lambda'$ can be formulated as an optimisation problem over all possible values of $\check{x}_s$ and $\hat{x}_s$ for every $\CR_s \in \CR$, as shown in Eq.~\ref{eq:optimise}, in which $\{  U'_{\check{x}_s}, \ldots, U'_{\hat{x}_s} \}$ is the sequence of interfering URs in $\lambda'$ that can impose contentions on $\CR_s$. 
\begin{equation} \label{eq:optimise}
    \begin{aligned}
    & {\text{Maximise }} 
    & & \sum_{\CR_s \in \CR} N(\CR_s, \{ U'_{\check{x}_s}, \ldots, U'_{\hat{x}_s} \}) \\
    & \text{on} 
    & &  \check{x}_{s},\hat{x}_{s} \in [1,|\lambda'|] ,~ \forall \CR_{s}\in \CR\\
    & \text{subject to} 
    & &  \check{x}_{s} \leq \hat{x}_{s} \land \hat{x}_{s-1} \leq\check{x}_{s}, ~\forall \CR_{s-1}, \CR_{s}\in \CR
    \end{aligned}
\end{equation}

Similar to the analysis in~\cite{zhang2022precise}, this optimisation problem can be solved via dynamic programming. By guaranteeing the conditions of $\check{x}_{s}$ and $\hat{x}_s$ in Lemma~\ref{lem:boundary} for every $\CR_s \in \CR$, only the feasible contentions will be accounted for during the computations~\cite{zhang2022precise}. This effectively avoids the pessimism in the first point described above. 
In addition, $\CR_{s-1}$ and $\CR_s$ may contain the same $r^k$ (\ie, $\CR_{s-1} \cap \CR_{s} \neq \varnothing$). For a $r^k \in \CR_{s-1} \cap \CR_{s}$ where every access incurs a cache miss in $\CR_{s-1}$, it should not be considered again when examining $\CR_s$.
This can be addressed in the construction of the analysis, as described below.

\SetKwInOut{Initialise}{Initialise}
\begin{algorithm}[t]
\caption{Constructing Inter-core Cache Analysis.}
\label{alg:dpmax}
\KwIn{$\forall \CR_s \in \CR$, $\forall U'_x \in \lambda'$;}


{\footnotesize\ttfamily/* {Create the states for $\CR_1$} */}\\

\For{\normalfont{each} $\hat{x} \in [1,~|\cp'|]$}{
    $U'=\{U'_1,...,U'_{\hat{x}}\}$; \\

    {\footnotesize\ttfamily/* {The $h^{th}$ state of $\CR_1$} */}\\
    $V_{1,\hat{x}}^h= N(\CR_1, U')$;~~~$V_1 = V_1 \cup \{ V_{1,\hat{x}}^h \}$;\\
    {\footnotesize\ttfamily/* {Identify refs. with full misses} */}\\
    $\zeta_{1}^h = \{r^k \in \CR_1~|~ n^k= \cnt^k\}$;\\
}

{\footnotesize\ttfamily/* {Create states for $\CR_s$ based on $V_{s-1}$} */}\\
\For{\normalfont{each} $s \in [2,~|\mathcal{C}|]$}
{
    \For{\normalfont{each} $V_{s-1, x}^l \in  V_{s-1}$}
    {
            \For{\normalfont{each} $\hat{x} \in [x,~|\cp'|]$}{
            
            $U'=\{U'_{x},...,U'_{\hat{x}}\}$;\\

            {\footnotesize\ttfamily/* {The $h^{\text{th}}$ state of $\CR_s$ via ${V}_{s-1,x}^l$} */}\\
            $V={V}_{s-1,x}^l+ N(\CR_s\setminus \zeta_{s-1}^l, U')\}$;\\
            ${V}_{s,\hat{x}}^h = max\{\mathrm{V}_{s,\hat{x}}^h,V\}$;~~~$V_s=V_s\cup\{V_{s,\hat{x}}^h\}$;\\
            {\footnotesize\ttfamily/* {Identify refs. with full misses} */}\\
            $\zeta_{s}^h = \{r^k\in \CR_s~|~r^k\in \zeta_{s-1}^l  \vee  n^k= \cnt^k\} $;\\
            
            }          
        
    }

}
\Return the maximum of all $V_{s,\hat{x}}^h$; \\

\end{algorithm}

Alg.~\ref{alg:dpmax} presents the analysis process given a $\CR$ and a $\lambda'$ via dynamic programming. Since the computation always focuses on the current $\CR_s \in \CR$, the index $s$ is omitted for $\check{x}_s$ and $\hat{x}_s$ in the algorithm.
For a $\CR_s \in \CR$, the $h^{th}$ state of $\CR_s$ (denoted as $V_{s,\hat{x}}^h$) represents a contention scenario when a sequence of CRs $\{\CR_1,\cdots,\CR_s\}$ is interfered by $\{U'_1,\cdots, U'_{\hat{x}}\}$, in which a CR in $\{\CR_1,\cdots,\CR_s\}$ can be interfered by a different UR segment in $\{U'_1,\cdots, U'_{\hat{x}}\}$. 
Notation $\zeta_s^h$ gives the set of $r^k\in\CR_s$ that every access of $r^k$ incurs a cache miss in the $h^{\text{th}}$ state of $\CR_s$.
The set of all states of $\CR_s$ is denoted as $V_s$.

The algorithm starts from $\CR_1$ to construct all its states.
For $\CR_1$, its states are constructed by the cache miss number when it is interfered by $U'=\{U'_1,...,U'_{\hat{x}}\}$, with $\hat{x} \in [1,|\lambda'|]$ (\ie, $V_{1,\hat{x}}^h=N(\CR_1, U')$ at lines 2–8). 
For a $V_{1,\hat{x}}^h$, $\zeta^h_1$ is constructed by $r^k \in \CR_1$ that each of its accesses incurs a cache miss (\ie, $n^k=\delta^k$ at line 7), where $n^k$ is the cache miss number of $r^k$ caused by $U'$. The $n^k$ is obtained by the operator $\Phi(\cdot)$ when computing $N(\CR_1, U')$, as described in Sec.~\ref{sec:contention}.

For $\CR_s$ with $s \geq 2$ (line 10), their states $V_{s,\hat{x}}^h$ depend on the states of the previous CRs.
Given a state of $\CR_{s-1}$ (denoted as $V_{s-1,x}^l$ to differentiate with the current state) and a $\hat{x} \in [x,~|\cp'|]$ (lines 11-12), a $V_{s,\hat{x}}^h$ of $\CR_s$ is obtained as follows.
First, the sequence of the interfering URs is constructed by $U'=\{U'_x,\cdots, U'_{\hat{x}}\}$ at line 13. 
Then, the number of misses in $\{\CR_1,\cdots, \CR_{s} \}$ is computed as $V_{s-1, x}^l + N(\CR_s\setminus \zeta_{s-1}^l, U')$ (line 15), based on the following two aspects. 
\begin{itemize}
\item[(i)] the number of misses incurred by  $\{\CR_1,\cdots, \CR_{s-1} \}$ from $\{U'_1,\cdots, U'_{x} \}$, \ie, the $V_{s-1, x}^l$ of $\CR_{s-1}$;
\item[(ii)] the number of misses incurred by $\CR_s$ from $U'=\{U'_{x},\cdots, U'_{\hat{x}}\}$, with all $r^k \in \zeta_{s-1}^l$ excluded from the computation, \ie, $N(\CR_s\setminus \zeta_{s-1}^l, U')$.
\end{itemize}


In addition, as $V_{s,\hat{x}}^h$ could also be computed based on another state of $\CR_{s-1}$, the maximum value of $V_{s,\hat{x}}^h$ is always taken as the worst case (line 16).
Finally, $\zeta_s^h$ is constructed to identify the set of $r^k\in\CR_s$ that every access incurs a cache miss, including ones in $\zeta_{s-1}^l \cap \CR_s$ (line 18).

The algorithm terminates with the maximum $V_{s,\hat{x}}^h$ returned as the cache miss number of $\CR$ imposed by $\lambda'$ (line 22).
The time complexity of Alg.~\ref{alg:dpmax} is $\mathcal{O}(|\CR| \cdot |V| \cdot |\lambda'|^2)$, where $|V|$ is the maximum number of states among $\CR_s \in \CR$. 
Given that the cache holds at most $\asso$ unique blocks, a $\CR_s$ contains no more than $\asso$ references with different addresses, limiting the state number to $2^\asso$. Thus, the time complexity is pseudo-polynomial as $\asso\leq16$ in most COST architectures of real-time systems~\cite{hennessy2017computer,fischer2023analysis,stock2019cache}.
Example~\ref{ex:dp} illustrates the state transitions
across the CRs in Fig.~\ref{fig:dp} based on the analysis.

\begin{figure}
    \centering
    \includegraphics[width=0.72\linewidth]{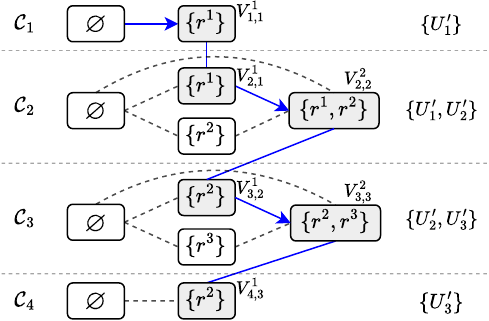}
    \caption{An example of state transition across four CRs with references that access a block address only once \textit{(grey blocks: worst-case states; white blocks: potential states)}.}
    \label{fig:dp}
    \vspace{-7pt}
\end{figure}

\begin{example}\label{ex:dp}
\normalfont
\add{ In Fig.~\ref{fig:dp}, $\{U'_1\}$, $\{U'_1, U'_2\}$, and $\{U'_2, U'_3\}$ can cause $r^1$, $r^2$, and $r^3$ to be evicted, respectively.}
First, $\CR_1$ has state $V_{1,1}^1= 1$, which is obtained when $\CR_1$ is interfered by $\{U'_1\}$ with $\zeta_1^1 = \{r^1\}$. 
For $\CR_2$, as $r^1$ already incurs a cache miss, we have \add{$V_{2,1}^1=V_{1,1}^1 + N(\{r^2\}, \{ U_1'\}) =1$} with $\zeta_2^1=\{r^1\}$. In addition, given that $r^2$ can incur a cache miss by $\{U'_1, U'_2\}$, \add{$V_{2,2}^2= V_{1,1}^1 + N(\{r^2\}, \{ U_1', U_2'\}) = 2$} is constructed with $\zeta_2^2 = \{r^1, r^2\}$.
The same process repeats for the following CRs, which terminate at $\CR_4$ with $V_{4,3}^1=3$ and $\zeta_4^1=\{r^2\}$. 
\end{example}

By examining every execution path of $\tau$ and $\tau'$, the worst-case inter-core cache latency incurred by $\tau$ can be obtained by multiplying the maximum number of misses among all paths of $\tau$ and the latency of a single cache miss~\cite{zhang2022precise}.

\subsection{Supporting Multi-level Set-associative Cache on Multicores} \label{sec:multi-multi}

This section describes the application of the proposed analysis to systems which contain a multi-level set-associative cache hierarchy shared among multiple cores.

\textbf{Supporting Multi-level Set-associative Cache.} 
In a multi-level cache, the processor queries each level of the cache hierarchy until the requested data is found~\cite{mueller1997timing}.
This implies only the memory accesses of $\tau$ and $\tau'$ that can reach the shared cache could interfere with each other.
Whether a memory access will reach a particular level of cache can be effectively determined by Cache Access Classification (CAC)~\cite{hardy2008wcet,zhang2015precise}.
According to CAC, if an access to $b_i$ is classified as Always Hit at the $(L{-}1)$-level cache, then it will not reach the $L$-level cache.
Based on this property, the proposed analysis can effectively support a multi-level cache hierarchy.
Hence, for a specified cache set and cache level, the $\CR$ of $\tau$ and interfering URs $U'$ of $\tau'$ only contain the memory accesses that (i) can access a given level of the shared cache and (ii) maps to the same cache set based on the block address~\cite{white1997timing}.
By doing so, the inter-core cache contention in every cache set across multiple cache levels can be effectively quantified.

\textbf{Supporting Multiple Interfering Tasks.}
\add{
For multiple tasks running on different cores, the interfering tasks on each remote core (\ie, ones that can interfere $\tau$) are first identified by the analysis in~\cite{liang2012timing}, where the URs of such tasks are merged into a set UR sequence (\ie, a $\lambda'$) for every remote core based on their execution order.}
\add{Then, the same process in Alg.~\ref{alg:dpmax} can be applied for the $\mathcal{C}$ of $\tau$ and the $\lambda'$ on every remote core, identifying the sequence of URs (\ie, $U' \subseteq \lambda'$) on each core that can interfere a $\CR_s \in \CR$ based on the temporal constraints.
However, as the $U'$ of each remote core requires a $\hat{x}$ to handle the temporal constraints, the dimensions of ${V}_{s,\hat{x}}^h$ would increase with the number of cores, leading to significant computation overhead. 
To improve efficiency, the temporal constraints of URs can be relaxed in the implementation by merging certain URs into a single UR for the interfering tasks.}

In addition, to quantify the number of cache misses incurred by a $\CR_s$ due to the $U'$s on remote cores, the contention analysis in Sec.~\ref{sec:contention} can be applied (\ie, the $\Phi$ operator) using a global access queue that concatenates the access queue (\ie, $Q$) of every $U'$. This is realised in the reference implementation provided in the paper.
However, it is worth noting that this approach may affect the accuracy of the analysis by overlooking the partial order among different $U'$.
To improve the accuracy, the global access queue can be constructed via a search algorithm based on the partial order of the different $U'$. However, this would further complicate the analysis with additional computation costs. 

 \begin{figure*}[t]
 \centering
 \includegraphics[width=.96\linewidth]{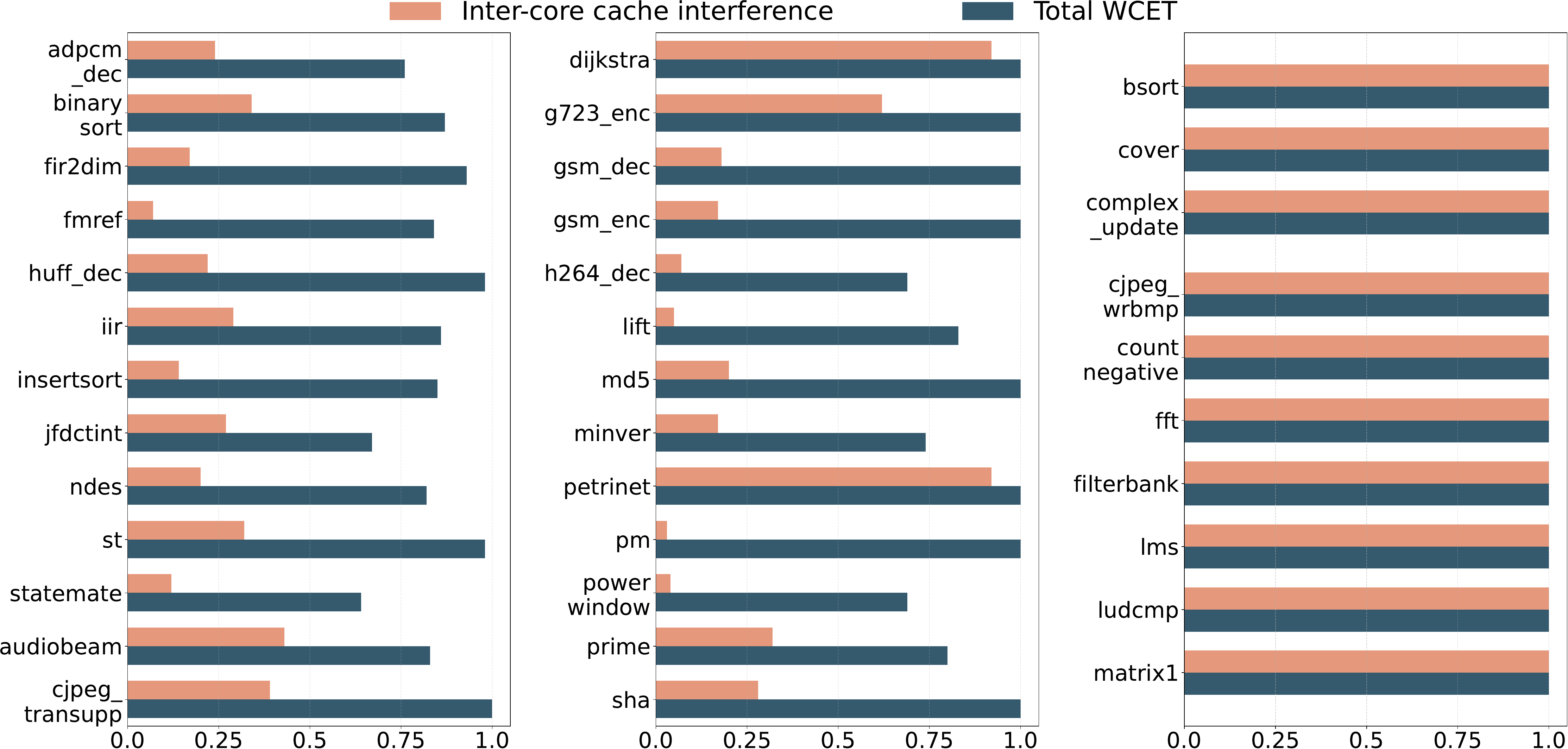}
 \caption{\add{The average ratio of inter-core cache interference and total WCET of proposed analysis to that of Zhang2022. Each task is analysed under the 12 interfering tasks. The complete results provided in Tab.~\ref{table:full_inter} and~\ref{table:full_wcet} in~\cite{zhao2025wcet}.}}
 \label{fig:Overall}
 \vspace{-4pt}
 \end{figure*}

\section{Experimental Results} 
\label{sec:result}

This section evaluates the proposed analysis against the existing approaches in~\cite{liang2012timing} (denoted as Liang2012) and~\cite{zhang2022precise} (denoted as Zhang2022), in terms of the resulting inter-core cache contention under various parallel tasks (Sec.~\ref{sec:ana_precision}) and different cache settings (Sec.~\ref{sec:scale_ana}). The computation cost of the proposed analysis is reported in Sec.~\ref{sec:ana_cost}.

\newcommand{\dd}[2]{%
  \diagbox[dir=SW, innerleftsep=1pt, innerrightsep=1pt, linewidth=0.3pt]{#1}{#2}%
}

\subsection{Experimental Setup}\label{sec:exp_setup}

We implemented the proposed and the existing analyses (\ie, Liang2012~\cite{liang2012timing} and Zhang2022~\cite{zhang2022precise}) in the WCET analysis framework LLVM-TA~\cite{hahn2022llvmta}, extending the LLVM-TA to support multicore platforms.  
The LLVM-TA can produce the CFG of a task under evaluation, as well as the ages of its blocks at different cache levels based on the intra-core cache analysis.
To compare the competing analysing methods, we modified the LLVM-TA to produce both the total WCET estimations of a task and the associated inter-core cache latency. 

\textbf{Hardware Platform.} 
The analysis is conducted on a dual-core in-order processor with a standard five-stage pipeline and a two-level cache hierarchy with the LRU replacement policy~\cite{touzeau2019fast, reineke2007timing} applied.
The L1 cache is dedicated to each core, whereas the L2 cache is shared among cores.
As with~\cite{stock2019cache,touzeau2019fast}, the size of a cache line for both the L1 and L2 cache is set to 16 bytes.
The L1 cache is configured with a total size of 512 bytes, which contains a separate instruction and data cache with 8 sets and 2 ways, \ie, 2 lines for each cache set. 
The shared L2 cache has 32 sets and $\asso$ ways, with $\asso \in \{2, 4, 8\}$ ($\asso=2$ by default).
Similar with~\cite{liang2012timing, zhang2022precise, nagar2016precise}, the latency of an L1 cache hit, an L2 cache hit and an L2 cache miss is set to 1 cycle, 5 cycles and 100 cycles, respectively.

\textbf{Benchmarks.} 
\add{In total, 36 tasks from the TACLeBench suite~\cite{falk2016taclebench} are analysed by the competing methods, where each task can incur the inter-core cache contention from interfering tasks running on other cores.}
\add{Tab.~\ref{table:l2-stats} lists the benchmark tasks that are used as the interfering tasks in the evaluation, along with their cache-access statistics obtained from the intra-core analysis in LLVM-TA ($\asso=2$)~\cite{hahn2022llvmta}, including the total number of accesses, number of unique blocks, L2 access rate, and L2 hit rate.}
As shown in the table, these benchmarks cover a wide range of access numbers (from 419 in \texttt{binarysearch} to 8,027,527 in \texttt{fmref}) and the number of unique blocks (from 32 to 688), 
reflecting diverse cache-accessing behaviours \add{that impose different degrees of inter-core cache interference on the shared cache for the evaluated tasks}.
The tasks are compiled with an optimisation level of O0.

\begin{table}[t]
    \centering
    \caption{Statistics of the interfering tasks with $\kappa=2$.}
    \label{table:l2-stats}
    \setlength{\tabcolsep}{5pt} 
\resizebox{\columnwidth}{!}{
    \begin{tabular}{llccc}
        \toprule
        Benchmark&\makecell[l]{Total \\Access}&\makecell[c]{Unique \\ blocks}& \makecell[c]{L2 Access \\Ratio} & \makecell[c]{L2 Hit\\Ratio} \\
        \midrule
        \texttt{adpcm\_dec}      & 7,052     & 300 & 39.60\% & 7.10\% \\
        \texttt{cover}           & 102,283   & 237 & 34.13\% & 0.00\% \\
        \texttt{fir2dim}         & 132,54    & 152 & 38.19\% & 1.80\% \\
        \texttt{huff\_dec}       & 5,403,338  & 263 & 37.60\% & 0.26\% \\
        \texttt{iir}             & 435      & 32  & 44.74\% & 4.24\% \\
        \texttt{ndes}            & 364,305   & 465 & 38.05\% & 6.30\% \\
        \texttt{st}              & 3,319,950  & 118 & 32.85\% & 0.47\% \\
        \texttt{binarysearch}    & 419      & 37  & 42.32\% & 3.38\% \\
        \texttt{fmref}           & 8,027,527  & 688 & 39.00\% & 14.91\% \\
        \texttt{insertsort}      & 13,638    & 32  & 40.43\% & 11.56\% \\
        \texttt{jfdctint}        & 2,289     & 83  & 37.02\% & 20.44\% \\
        \texttt{statemate}       & 315,361   & 656 & 36.54\% & 4.02\% \\
        \bottomrule
    \end{tabular}}
    \vspace{-7pt}
\end{table}

\subsection{Analytical Results under Different Parallel Tasks}\label{sec:ana_precision}

\add{
This section compares the proposed analysis and Zhang2022 in terms of the inter-core cache interference and the total WCET estimations (including the inter-core cache interference) of the 36 tasks, where each task (\ie, $\tau$) is analysed with the inter-core cache contention imposed by every interfering task (\ie, $\tau'$). 
Fig.~\ref{fig:Overall} presents the average ratio of the inter-core cache interference and the total WCET estimations of each task, \ie, the average ratio of the proposed method to Zhang2022 for that task across the 12 interfering tasks.
For the sake of completeness, detailed results of each evaluated task are available in Tab.~\ref{table:full_inter} and Tab.~\ref{table:full_wcet} of the technical report~\cite{zhao2025wcet}. 
Compared to Zhang2022, the proposed analysis reduces the inter-core cache interference and WCET estimations by 52.31\% and 8.94\% on average for all 36 tasks.}

\add{In particular, we observed three different types of  behaviours across the 36 tasks for accessing the L2 cache, as described below. (i) We observed that 26 out of 36 tasks can incur evictions on the L2 cache due to the interfering task (\ie, from \texttt{adpcm_dec} to \texttt{sha} in the figure). 
For such tasks, the proposed analysis is in general effective, which outperforms Zhang2022 by reducing the inter-core cache interference and WCET estimations by 72.43\% and 12.38\% on average, respectively. 
(ii) 3 tasks (from \texttt{bsort} to \texttt{complex_update}) can always hit the L1 cache except the first miss, \ie, all data required by the task can be fitted in the L1 cache. As these tasks would not incur delay on the L2 cache due to the interfering task, both methods return an inter-task interference of zero with identical WCET estimations. 
(iii) There are 7 tasks (from \texttt{cjpeg_wrbmp} to \texttt{matrix1}) that always miss the L2 cache, based on the intra-core cache analysis without considering any inter-core interference. For these tasks, both methods yield the same inter-core cache interference of zero.}

\add{From the observations, although there exist certain tasks in which the inter-core cache interference is not a dominant factor for the WCET, our analysis reduces WCET estimations by 12.38\% for the first 26 tasks (\ie, the first category where tasks incur L2 evictions). This demonstrates the effectiveness of the proposed analysis in the general case.
}

\subsection{Results under Varied Associativity \add{and Number of Cores}}\label{sec:scale_ana}

To further evaluate the effectiveness of the proposed analysis, \add{this section compares the analytical results of the competing analysis under varied cache associativity ($\kappa = 2, 4, 8$ for the L2 cache) and number of cores (2, 4, 6 and 8 cores).}

\begin{figure}[t]
\vspace{-4pt}
\centering
\subfigure[Results under $\asso = 2$.]
{\label{fig:ndes_2}
{\includegraphics[width=0.48\textwidth]{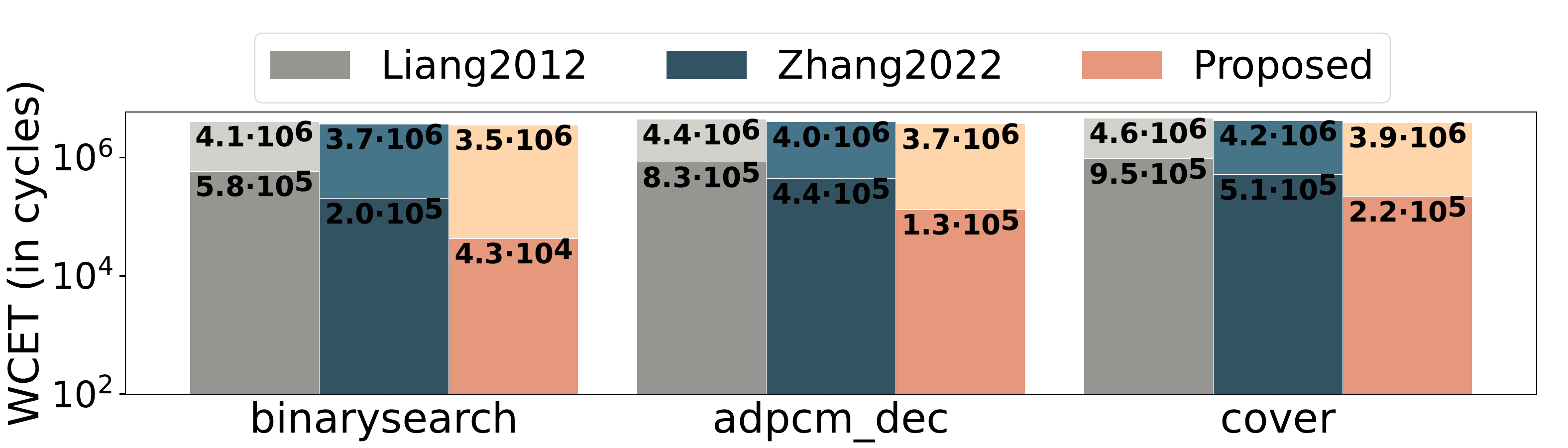}}}
\hspace{12pt}
\subfigure[Results under $\asso = 4$.]
{\label{fig:queue_4}
{\includegraphics[width=.48\textwidth]{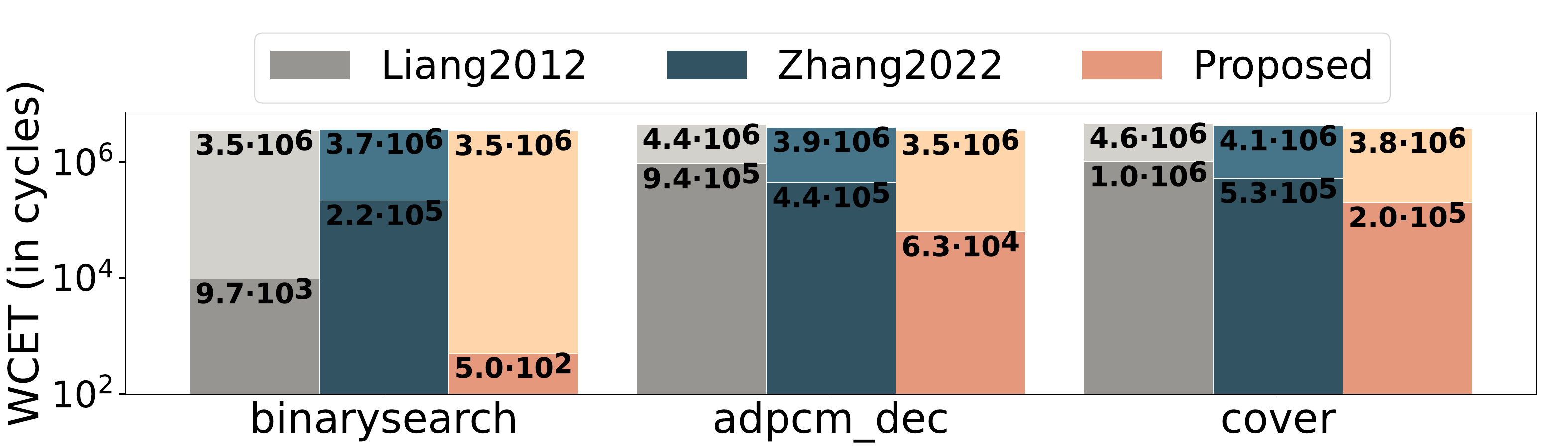}}}
\hspace{12pt}
\subfigure[Results under $\asso = 8$.]
{\label{fig:ndes_8}
{\includegraphics[width=.48\textwidth]{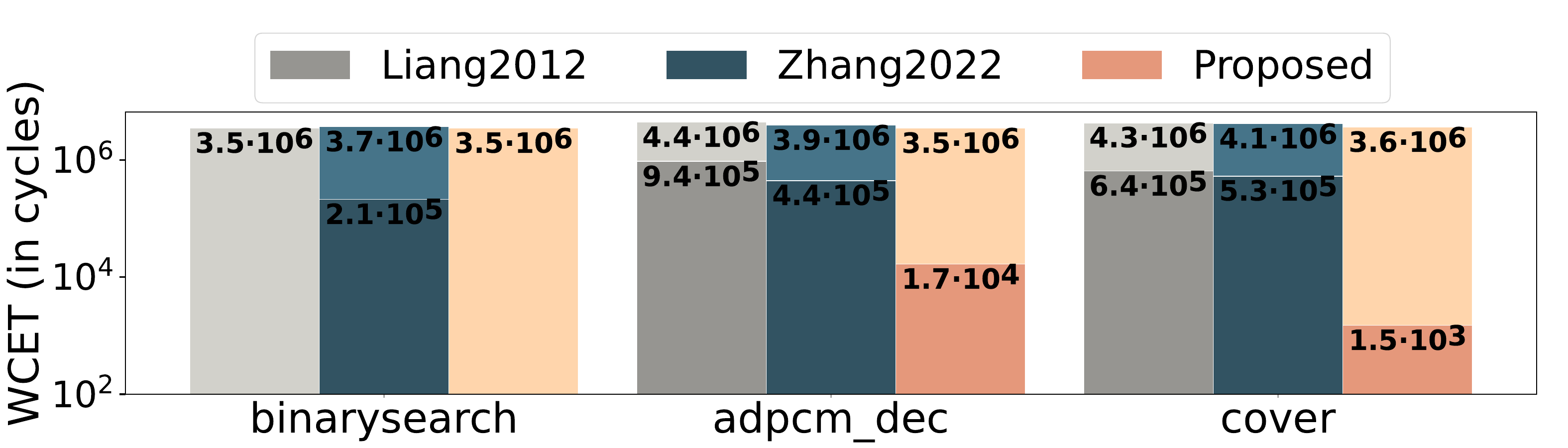}}}
\caption{Results (in log-scale) of \texttt{ndes} when being interfering by \texttt{binarysearch}, \texttt{adpcm\_dec} and \texttt{cover} \textit{\add{(dark colours: inter-core cache latency; light colours: total WCET)}}.}
\label{fig:ndes_248}
\vspace{-4pt}
\end{figure}

Fig.~\ref{fig:ndes_248} presents the WCET and the associated inter-core cache latency (in log-scale) of \texttt{ndes} obtained by each analysis under different $\asso$, with an interference task \texttt{binarysearch}, \texttt{cover}, or \texttt{adpcm\_dec}.
The \texttt{ndes} has a moderate amount of memory access, whereas \texttt{binarysearch}, \texttt{cover}, and \texttt{adpcm\_dec} provide different degrees of shared cache contention (see Tab.~\ref{table:l2-stats}).
As Liang2012 provides a holistic approach that analyses the WCET directly, its inter-core cache latency is obtained by subtracting the intra-core latency from the WCET estimations.
In addition, we note that as \texttt{ndes} are compiled together with each interfering task, the virtual addresses of \texttt{ndes} can vary with different interfering tasks, hence leading to a different intra-core latency.

As shown in the figure, both the WCET and the inter-core cache latency of \texttt{ndes} show a decreasing trend in a general case when $\asso$ is increased.
When $\kappa = 2$, the proposed analysis achieves lower estimations than both the competing analyses, whereas Zhang2022 produces tighter results than Liang2012. This is because the proposed method computes the number of cache misses by considering both the age and the number of memory accesses of the parallel tasks, whereas Zhang2022 utilises the partial order of URs to reduce possible contentions, thereby leading to tighter bounds. 

When $\kappa = 4$, we observe that the inter-core cache contention from \texttt{binarysearch} is significantly reduced as it only accesses 37 unique block addresses.
This is captured by both the proposed analysis and Liang2012.
However, as Zhang2022 does not account for the cache associativity, the resulting inter-core cache latency is similar to the case of $\asso=2$.
In addition, for \texttt{cover} and \texttt{adpcm\_dec} with a large number of accesses to different memory blocks, both Liang2012 and Zhang2022 report slightly higher inter-core cache latency compared to the results in $\kappa = 2$.
This is because with a higher $\asso$, more accesses in \texttt{ndes} can hit the L2 cache based on the intra-core analysis. In this case, the contention effects caused by a remote task can be amplified, thereby leading to an increased inter-core cache latency. However, as described, the total WCET estimations generally decrease under the competing analyses as $\kappa$ increases.

Finally, when $\kappa$ is increased to 8, the \texttt{binarysearch} can hardly cause any misses with a negligible inter-core cache latency, as suggested by both the proposed analysis and Liang2012.
In addition, we observe that under both the proposed analysis and Liang2012, the contention imposed by \texttt{cover} becomes lower than that of \texttt{adpcm_dec}.
This is because both analyses consider the $\asso$ and the number of unique block addresses in the interfering task, where
\texttt{adpcm\_dec} accesses a larger number of unique blocks than \texttt{cover} (see Tab.~\ref{table:l2-stats}).
However, while Liang2012 produces a lower bound for \texttt{cover}, the results for \texttt{adpcm\_dec} remain unchanged under $\asso=8$ as it overlooks the access number of the interfering task.  
As for our analysis, it achieves much lower estimations for both interfering tasks by considering both the access count and the number of unique addresses, where the reduction is more significant for \texttt{cover}.  
In addition, the estimations obtained by Zhang2022 remain unchanged in general across the experiment due to the obliviousness of $\kappa$.

\begin{figure}[t]
    \centering
    \includegraphics[width=\linewidth]{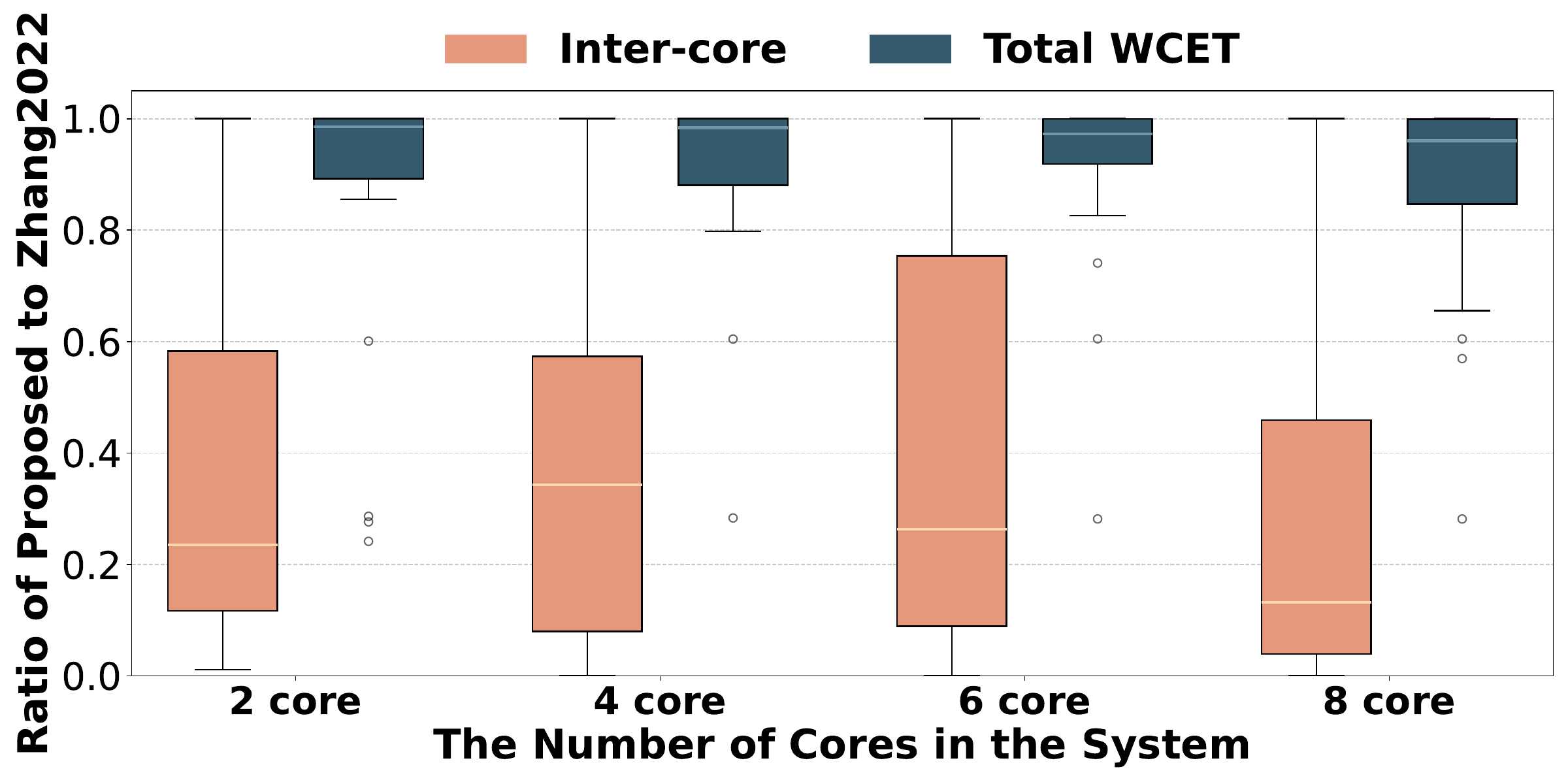}
    \caption{\add{The average ratio of inter-core cache interference and total WCET estimations of our analysis to Zhang2022 under a varied number of cores. \textit{(configuration: 2 core (L2: 1KB), 4 core (L2: 4KB), 6 core (L2: 8KB), 8 core (L2: 16KB))}.}}
    \label{fig:muti-core}
    \vspace{-7pt}
\end{figure}

\add{Fig.~\ref{fig:muti-core} presents the ratio of the inter-core cache interference and the total WCET estimations of the proposed analysis to that of Zhang2022 for the first 26 tasks (\ie, tasks in the first category described in Sec.~\ref{sec:ana_precision}) under a different number of cores.
In this experiment, the L2 cache size scales with the number of cores: 2 cores with 1 KB, 4 cores with 4 KB, 6 cores with 8 KB, and 8 cores with 16 KB.
For a system setting with $m$ cores, the first $m-1$ tasks in Tab.~\ref{table:l2-stats} are taken as the interfering tasks, where each task runs on a remote core. For instance, with $m=4$, the first three tasks (\texttt{adpcm\_dec} to \texttt{fir2dim}) in Tab.~\ref{table:l2-stats} are applied as the interfering tasks.}

\add{From the results, we can observe that compared to Zhang2022, the proposed analysis reduces the inter-core cache interference and WCET estimations by 57.84\% and 3.29\% on average, respectively. 
In addition, it is worth noting that the advantages of the proposed analysis become more obvious when the number of cores is increased. For instance, the proposed analysis reduces the inter-core cache interference and WCET estimation by 66.95\% and 4.13\% when the number of cores is eight, respectively.
This is because Zhang2022 does not take into account the cache-accessing behaviours (\eg, the cache associativity, see Sec.~\ref{sec:model-existing}) and sums the independently-computed interference of each parallel task, which fails to consider the actual cache state of the parallel tasks.
In contrast, the proposed analysis can produce tighter results by explicitly accounting for the cache-accessing behaviours and the number of cache accesses of the parallel interfering tasks in a holistic manner, especially when the number of cores is high.
}

\begin{table}[t]
\vspace{4pt}
\centering
\caption{\add{The average computation cost (in seconds) \textit{(binaryS: binarysearch; complex\_up: complex\_update; cjpeg\_tr: cjpeg\_transupp; countN: countnegative)}. }}
\label{table:exe_time}
\setlength{\tabcolsep}{1pt} 
\begin{tabular}{lcc@{\hskip 10pt}lcc} 
\toprule
Benchmark & Proposed & Zhang2022 & Benchmark & Proposed & Zhang2022 \\
\midrule
\texttt{adpcm\_dec}      & 1.49     & 0.83      & \texttt{lift}            & 1.07     & 0.80      \\
\texttt{binaryS}    & 0.20     & 0.45      & \texttt{md5}             & 0.66     & 24.67     \\
\texttt{fir2dim}         & 0.80     & 2.57      & \texttt{minver}          & 70.61    & 56.71     \\
\texttt{fmref}           & 60.86    & 19.10     & \texttt{petrinet}        & 0.01     & 1.14      \\
\texttt{huff\_dec}       & 5.26     & 11.01     & \texttt{pm}              & 242.20   & 449.89    \\
\texttt{iir}             & 0.22     & 0.35      & \texttt{powerwindow}     & 5.04     & 12.07     \\
\texttt{insertsort}      & 0.70     & 2.44      & \texttt{prime}           & 8502.78  & 1738.37   \\
\texttt{jfdctint}        & 9.27     & 0.64      & \texttt{sha}             & 2.72     & 69.44     \\
\texttt{ndes}            & 94.55    & 8.58      & \texttt{bsort}           & 0.01     & 1.00      \\
\texttt{st}              & 163.71   & 69.25     & \texttt{cover}           & 0.01     & 1.42      \\
\texttt{statemate}       & 1.38     & 1.81      & \texttt{complex\_up}& 0.01     & 0.58      \\
\texttt{audiobeam}       & 45.91    & 15256.67  & \texttt{cjpeg\_wrbmp}    & 0.01     & 161.99    \\
\texttt{cjpeg\_tr} & 86.34    & 13.98     & \texttt{countN}   & 0.01     & 1.43      \\
\texttt{dijkstra}        & 0.01     & 2.34      & \texttt{fft}             & 0.01     & 0.78      \\
\texttt{g723\_enc}       & 0.11     & 2.35      & \texttt{filterbank}      & 0.01     & 0.72      \\
\texttt{gsm\_dec}        & 0.56     & 6.56      & \texttt{lms}             & 0.01     & 0.43      \\
\texttt{gsm\_enc}        & 0.11     & 5.60      & \texttt{ludcmp}          & 0.01     & 4.50      \\
\texttt{h264\_dec}       & 337.13   & 254.36    & \texttt{matrix1}         & 0.01     & 0.25      \\
\bottomrule
\end{tabular}
\vspace{-6pt}
\end{table}

\subsection{Computation Cost}\label{sec:ana_cost}

Finally, we evaluate the average computation cost of the proposed analysis and Zhang2022 on an AMD EPYC 7763 processor running at 1.50 GHz with 503 GB of memory, as shown in
Tab.~\ref{table:exe_time} for \add{the 36 tasks analysed in Fig.~\ref{fig:Overall}}.

In general, the proposed analysis demonstrates a slight increase in the computation cost compared to Zhang2022.
This is expected due to the finer-grained quantification of inter-core cache contention in the proposed analysis, especially when the task contains a large number of URs (\eg, \texttt{st} and \texttt{fmref} in Tab.~\ref{table:exe_time}).
\add{However, we observe that the proposed analysis achieves a lower cost than Zhang2022 for tasks where the number of the CRs is significantly lower than that of the URs 
(\eg, \texttt{md5} with 2 CRs and 70 URs, and \texttt{sha} with 2 CRs and 408 URs). 
For these tasks, as only a limited number of CRs are constructed, this leads to fewer computations with a low cost for our analysis.}
\add{Especially, for tasks in the second and third categories described in Sec.~\ref{sec:ana_precision} (\eg, \texttt{bsort}), the cost of the proposed analysis is negligible} as no access can hit the L2 cache based on the intra-core analysis and thus no CR is constructed. 
Thus, although the proposed analysis introduces additional overhead for certain tasks, the increase remains modest due to its pseudo-polynomial-time complexity. Based on the above, our analysis effectively tightens the inter-core cache contention estimations while not significantly increasing the overhead, validating the strength of the proposed approach.

\section{Conclusion} \label{sec:conclusion}

This paper proposed an inter-core cache contention analysis for estimating the cache miss latency on the shared cache caused by parallel tasks. While preserving the partial order of program regions, the proposed analysis computes the number of cache misses by explicitly considering the age and the number of accesses to memory blocks of the tasks.
Experimental results show that compared to existing approaches, our analysis yields lower estimations across different parallel tasks and cache settings without significantly increasing the computation cost, providing an effective solution for the timing verification of multicore real-time systems.

\bibliographystyle{IEEEtran}
\bibliography{ref}


\section{Appendix}

\begin{table*}[h]
\centering
\caption{Ratio of the inter-core cache interference of proposed analysis to Zhang2022 \textit{(binaryS: binarysearch; insertS: insertsort)}.}
\label{table:full_inter}
\setlength{\tabcolsep}{1.5pt} 
\renewcommand{\arraystretch}{1.2}
\resizebox{.98\textwidth}{!}{
\begin{tabular}{|c|l|c|c|c|c|c|c|>{\hspace{5pt}}c<{\hspace{5pt}}|c|c|>{\hspace{4pt}}c<{\hspace{4pt}}|>{\hspace{5pt}}c<{\hspace{5pt}}|c|>{\hspace{4pt}}c<{\hspace{4pt}}|}

\hline
Classification &Benchmark & adpcm\_dec & binaryS & cover & fir2dim & fmref & huff\_dec & iir & insertS & jfdctint & ndes & st & statemate & \textbf{Avg.} \\
\hline
\multirow{26}{*}{\makecell[c]{Tasks that\\incur evictions on\\the L2 cache due to\\the interfering task}} 
&adpcm\_dec & 0.48 & 0.31 & 0.28 & 0.07 & 0.06 & 0.10 & 0.19 & 0.21 & 0.38 & 0.11 & 0.04 & 0.70 & 0.24 \\
\cline{2-15}  &
binarysearch & 1.0 & 0.01 & 0.33 & 0.03 & 0.07 & 0.07 & 1.0 & 1.0 & 0.17 & 0.06 & 0.05 & 0.25 & 0.34 \\
\cline{2-15}  &
fir2dim & 0.24 & 0.06 & 0.53 & 0.05 & 0.07 & 0.09 & 0.04 & 0.06 & 0.10 & 0.12 & 0.02 & 0.70 & 0.17 \\
\cline{2-15} &
fmref & 0.15 & 0.01 & 0.23 & 0.02 & 0.10 & 0.10 & 0.01 & 0.01 & 0.02 & 0.02 & 0.02 & 0.20 & 0.07 \\
\cline{2-15}  &
huff\_dec & 0.12 & 1.0 & 1.0 & 0.01 & 0.09 & 0.11 & 0.01 & 0.01 & 0.01 & 0.08 & 0.01 & 0.25 & 0.22 \\
\cline{2-15}  &
iir & 1.0 & 1.0 & 0.20 & 0.14 & 0.08 & 0.14 & 0.08 & 0.06 & 0.24 & 0.11 & 0.03 & 0.33 & 0.29 \\
\cline{2-15}  &
insertsort & 0.23 & 0.01 & 0.25 & 0.14 & 0.01 & 0.07 & 0.01 & 0.33 & 0.06 & 0.14 & 0.05 & 0.40 & 0.14 \\
\cline{2-15}  &
jfdctint & 0.81 & 0.19 & 0.39 & 0.11 & 0.11 & 0.16 & 0.22 & 0.25 & 0.13 & 0.16 & 0.06 & 0.64 & 0.27 \\
\cline{2-15}  &
ndes & 0.30 & 0.21 & 0.43 & 0.11 & 0.10 & 0.23 & 0.01 & 0.01 & 0.22 & 0.15 & 0.03 & 0.60 & 0.20 \\
\cline{2-15}  &
st & 0.92 & 0.04 & 0.51 & 0.15 & 0.20 & 0.34 & 0.01 & 0.34 & 0.27 & 0.27 & 0.01 & 0.82 & 0.32 \\
\cline{2-15}  &
statemate & 0.12 & 0.03 & 0.18 & 0.03 & 0.07 & 0.07 & 0.01 & 0.01 & 0.05 & 0.09 & 0.05 & 0.70 & 0.12 \\
\cline{2-15}  &
audiobeam & 0.61 & 0.13 & 0.40 & 0.18 & 0.07  & 0.14 & 1.0 & 1.0 & 0.50 & 0.11 & 0.08 & 1.0 & 0.43 \\
\cline{2-15}  
&cjpeg\_transupp&0.46&0.12&0.75&0.42&0.16&0.47&0.18&0.02&0.42&0.48&0.23&1.0&0.39 \\
\cline{2-15}  &
dijkstra & 1.0 & 1.0 & 0.03 & 1.0 & 1.0 & 1.0 & 1.0 & 1.0 & 1.0 & 1.0 & 1.0 & 1.0 & 0.92 \\
\cline{2-15}  &
g723\_enc & 0.20 & 0.01 & 0.18 & 0.05 & 0.05 & 1.0 & 1.0 & 1.0 & 1.0 & 1.0 & 1.0 & 1.0 & 0.62 \\
\cline{2-15}  &
gsm\_dec & 0.62 & 0.01 & 0.25 & 0.01 & 0.08 & 0.07 & 0.02 & 0.01 & 0.02 & 0.09 & 0.01 & 1.0 & 0.18 \\
\cline{2-15}  &
gsm\_enc & 0.57 & 0.10 & 0.24 & 0.09 & 0.05 & 0.10 & 0.01 & 0.01 & 0.17 & 0.12 & 0.08 & 0.50 & 0.17 \\
\cline{2-15}  &
h264\_dec & 0.12 & 0.01 & 0.08 & 0.01 & 0.03 & 0.06 & 0.01 & 0.01 & 0.04 & 0.04 & 0.01 & 0.39 & 0.07 \\
\cline{2-15}  &
lift & 0.01 & 0.01 & 0.07 & 0.01 & 0.08 & 0.08 & 0.01 & 0.02 & 0.05 & 0.04 & 0.01 & 0.25 & 0.05 \\
\cline{2-15}  &
md5&1.0&1.0&0.07&0.01&0.11&0.12&0.01&0.01&0.01&0.01&0.04&0.09&0.20\\
\cline{2-15}  &
minver & 0.52 & 0.16 & 0.23 & 0.06 & 0.02 & 0.05 & 0.15 & 0.10 & 0.15 & 0.16 & 0.05 & 0.41 & 0.17 \\
\cline{2-15}  &
petrinet & 1.0 & 0.01 & 1.0 & 1.0 & 1.0 & 1.0 & 1.0 & 1.0 & 1.0 & 1.0 & 1.0 & 1.0 & 0.92 \\
\cline{2-15}  &
pm & 0.05 & 0.01 & 0.02 & 0.01 & 0.11 &0.05 & 0.01 & 0.01 & 0.01 & 0.02 & 0.01 & 0.11 & 0.03 \\
\cline{2-15}  &
powerwindow & 0.01 & 0.01 & 0.05 & 0.01 & 0.06 & 0.09 & 0.01 & 0.01 & 0.01 & 0.02 & 0.01 & 0.29 & 0.04 \\
\cline{2-15}  &
prime & 0.92 & 0.13 & 0.41 & 0.14 & 0.14 & 0.19 & 0.28 & 0.26 & 0.34 & 0.22 & 0.06 & 0.76 & 0.32 \\
\cline{2-15}  &
sha & 0.67 & 0.01 & 0.14 & 0.10 & 0.06 & 0.14 & 0.01 & 0.01 & 1.0 & 0.15 & 0.12 & 1.0 & 0.28 \\

\hline\hline
\multirow{3}{*}{\makecell[c]{Tasks that always\\hit the L1 cache\\except the first miss}}

&bsort & 1.0 & 1.0 & 1.0 & 1.0 & 1.0 & 1.0 & 1.0 & 1.0 & 1.0 & 1.0 & 1.0 & 1.0 & 1.0 \\
\cline{2-15}  &
cover & 1.0 & 1.0 & 1.0 & 1.0 & 1.0 & 1.0 & 1.0 & 1.0 & 1.0 & 1.0 & 1.0 & 1.0 & 1.0 \\
\cline{2-15}  &
complex\_updates & 1.0 & 1.0 & 1.0 & 1.0 & 1.0 & 1.0 & 1.0 & 1.0 & 1.0 & 1.0 & 1.0 & 1.0 & 1.0 \\

\hline\hline
\multirow{7}{*}{\makecell[c]{Tasks that cannot\\hit the L2 cache\\according to\\the intra-core\\cache analysis}} 
&cjpeg\_wrbmp & 1.0 & 1.0 & 1.0 & 1.0 & 1.0 & 1.0 & 1.0 & 1.0 & 1.0 & 1.0 & 1.0 & 1.0 & 1.0 \\
\cline{2-15}  &
countnegative & 1.0 & 1.0 & 1.0 & 1.0 & 1.0 & 1.0 & 1.0 & 1.0 & 1.0 & 1.0 & 1.0 & 1.0 & 1.0 \\
\cline{2-15}  &
fft & 1.0 & 1.0 & 1.0 & 1.0 & 1.0 & 1.0 & 1.0 & 1.0 & 1.0 & 1.0 & 1.0 & 1.0 & 1.0 \\
\cline{2-15}  &
filterbank & 1.0 & 1.0 & 1.0 & 1.0 & 1.0 & 1.0 & 1.0 & 1.0 & 1.0 & 1.0 & 1.0 & 1.0 & 1.0 \\
\cline{2-15}  &
lms & 1.0 & 1.0 & 1.0 & 1.0 & 1.0 & 1.0 & 1.0 & 1.0 & 1.0 & 1.0 & 1.0 & 1.0 & 1.0 \\
\cline{2-15}  &
ludcmp & 1.0 & 1.0 & 1.0 & 1.0 & 1.0 & 1.0 & 1.0 & 1.0 & 1.0 & 1.0 & 1.0 & 1.0 & 1.0 \\
\cline{2-15}  &
matrix1 & 1.0 & 1.0 & 1.0 & 1.0 & 1.0 & 1.0 & 1.0 & 1.0 & 1.0 & 1.0 & 1.0 & 1.0 & 1.0 \\

\hline
\end{tabular}}
\end{table*}

This appendix provides the supplementary results from the evaluations described in Sec.~\ref{sec:result}, which compare the proposed analysis against Zhang2022. The results include four main aspects: (i) the comparison of the inter-core cache interference of the 36 tasks under every interfering task; (ii) the comparison of the total WCET the 36 tasks with the inter-core cache interference included; (iii) a detailed comparison of each evaluated task; and (iv) additional experiments using Papabench.

\textbf{Extended Evaluation.} First, Tab.~\ref{table:full_inter} presents the detailed results of the inter-core cache interference estimation comparisons shown in Fig.~\ref{fig:Overall}, where each task is analysed under 12 interfering tasks. 
Compared to Zhang2022, the proposed analysis reduces the inter-core cache interference estimations by 52.31\% on average for all 36 tasks.
As shown in the table, the tasks are organised into three groups according to their behaviours on the L2 cache: 
\begin{itemize}
\item Tasks that can incur evictions on the L2 cache due to the interfering task. 
For such tasks (26 out of 36 tasks), the proposed analysis outperforms Zhang2022~\cite{zhang2022precise} by reducing the inter-core cache interference estimations by 72.43\% on average. 

\item Tasks that always hit the L1 cache except the first miss, \ie, all data required by the task can be fitted in the L1 cache. For such tasks (3 out of the 36 tasks), they would not incur additional delay on the L2 cache due to the interfering task. In this case, both competing methods return an inter-task interference of zero.  
\item Tasks that always miss the L2 cache, based on the intra-core cache analysis without considering any inter-core interference. For these tasks (7 out of 36 tasks), both methods yield the inter-core cache interference of zero. 
\end{itemize}

\textbf{Comparison of WCET Estimations.} Second, Tab.~\ref{table:full_wcet} shows the comparison of the total WCET of 36 tasks under Zhang2022~\cite{zhang2022precise} and the proposed analysis. 
The tasks are organised by the same approach as described above.
From the results, we observed that for certain tasks (\eg, \texttt{huff_dec}), the inter-core cache interference is indeed negligible to the total WCET, \ie, the reduction in the inter-core cache interference estimations barely affects the total WCET estimations. However, the results also show there are 17 tasks in which the inter-core cache interference is non-trivial to the WCET, where our analysis can lead to reduced WCET estimations.
Compared to~\cite{zhang2022precise}, our analysis reduces the total WCET estimations by 12.38\% on average for tasks in the first group (26 tasks) and 8.94\% on average for all evaluated tasks, respectively. This observation illustrates that for certain tasks, the impact of the inter-core cache interference on the total WCET is trivial. However, the results justify that the proposed analysis is in general effective in reducing the WCET estimations. 

\begin{table*}[t]
\centering
\caption{Ratio of total WCET of proposed analysis to Zhang2022 \textit{(binaryS: binarysearch; insertS: insertsort)}. $1.0^* \approx 1$ means the situation where the proposed analysis tightens the inter-core cache interference compared to~\cite{zhang2022precise}, however, such interference is negligible to the total WCET.}
\label{table:full_wcet}
\setlength{\tabcolsep}{1.5pt} 
\renewcommand{\arraystretch}{1.2}
\resizebox{.98\textwidth}{!}{

\begin{tabular}
{|c|l|c|c|c|c|c|c|>{\hspace{5pt}}c<{\hspace{5pt}}|c|c|>{\hspace{4pt}}c<{\hspace{4pt}}|>{\hspace{5pt}}c<{\hspace{5pt}}|c|>{\hspace{4pt}}c<{\hspace{4pt}}|}

\hline
Classification &Benchmark & adpcm\_dec & binaryS & cover & fir2dim & fmref & huff\_dec & iir & insertS & jfdctint & ndes & st & statemate & \textbf{Avg.} \\
\hline
\multirow{26}{*}{\makecell[c]{Tasks that\\incur evictions on\\the L2 cache due to\\the interfering task}} 
&adpcm\_dec & 0.93 & 0.93 & 0.85 & 0.50 & 0.46 & 0.61 & 0.93 & 0.93 & 0.90 & 0.63 & 0.50 & 0.97 & 0.76 \\
\cline{2-15}
&binarysearch & 1.0 & 0.98 & 0.96 & 0.60 & 0.76 & 0.76 & 1.0 & 1.0 & 0.96 & 0.75 & 0.71 & 0.94 & 0.87 \\
\cline{2-15}
&fir2dim & 0.97 & 0.99 & 0.99 & 0.80 & 0.85 & 0.88 & 0.99 & 0.96 & 0.98 & 0.90 & 0.90 & 0.99 & 0.93 \\
\cline{2-15}
&fmref & 0.99 & 0.96 & 0.97 & 0.92 & 0.76 & 0.79 & 0.96 & 0.91 & 0.89 & 0.39 & 0.56 & 0.95 & 0.84 \\
\cline{2-15}
&huff\_dec & 1.0* & 1.0 & 1.0 & 0.99 & 0.99 & 0.97 & 0.99 & 1.0* & 0.98 & 0.97 & 0.90 & 0.99 & 0.98 \\
\cline{2-15}
&iir & 1.0 & 1.0 & 0.81 & 0.91 & 0.85 & 0.91 & 0.94 & 0.85 & 0.89 & 0.73 & 0.49 & 0.92 & 0.86 \\
\cline{2-15}
&insertsort & 0.86 & 0.80 & 0.96 & 0.96 & 0.41 & 0.90 & 0.70 & 0.98 & 0.96 & 0.91 & 0.80 & 0.98 & 0.85 \\
\cline{2-15}
&jfdctint & 0.97 & 0.74 & 0.75 & 0.45 & 0.48 & 0.46 & 0.86 & 0.87 & 0.72 & 0.49 & 0.32 & 0.90 & 0.67 \\
\cline{2-15}
&ndes & 0.92 & 0.96 & 0.93 & 0.78 & 0.61 & 0.84 & 0.91 & 0.72 & 0.93 & 0.74 & 0.56 & 0.96 & 0.82 \\
\cline{2-15}
&st & 1.0* & 0.99 & 0.99 & 0.96 & 0.96 & 0.98 & 0.99 & 0.99 & 0.98 & 0.97 & 0.93 & 1.0* & 0.98 \\
\cline{2-15}
&statemate & 0.28 & 0.97 & 0.73 & 0.07 & 0.63 & 0.10 & 0.92 & 0.82 & 0.88 & 0.67 & 0.67 & 0.98 & 0.64 \\
\cline{2-15}
&audiobeam & 1.0* & 1.0* & 1.0* & 1.0* & 1.0* & 1.0* & 1.0 & 1.0 & 1.0* & 1.0* & 1.0* & 1.0 & 0.83 \\

\cline{2-15}
&cjpeg\_transupp & 1.0* & 1.0* & 1.0* & 1.0* & 1.0* & 1.0* & 1.0* & 1.0* & 1.0* & 1.0* & 1.0* & 1.0 & 1.0* \\
\cline{2-15}
&dijkstra & 1.0 & 1.0 & 1.0* & 1.0 & 1.0 & 1.0 & 1.0 & 1.0 & 1.0 & 1.0 & 1.0 & 1.0 & 1.0* \\
\cline{2-15}
&g723\_enc & 1.0* & 1.0* & 1.0* & 0.99 & 0.99 & 1.0 & 1.0 & 1.0 & 1.0 & 1.0 & 1.0 & 1.0 & 1.0* \\
\cline{2-15}
&gsm\_dec & 1.0* & 1.0* & 1.0* & 1.0* & 1.0* & 1.0* & 1.0* & 1.0* & 1.0* & 1.0* & 1.0* & 1.0 & 1.0* \\
\cline{2-15}
&gsm\_enc & 1.0* & 1.0* & 1.0* & 1.0* & 1.0* & 1.0* & 1.0* & 1.0* & 1.0* & 1.0* & 1.0* & 1.0* & 1.0* \\
\cline{2-15}
&h264\_dec & 0.74 & 0.52 & 0.93 & 0.11 & 0.81 & 0.21 & 0.96 & 0.96 & 0.96 & 0.29 & 0.80 & 0.98 & 0.69 \\
\cline{2-15}
&lift & 0.88 & 0.95 & 0.82 & 0.63 & 0.92 & 0.86 & 0.91 & 0.96 & 0.98 & 0.68 & 0.44 & 0.94 & 0.83 \\
\cline{2-15}
&md5 & 1.0 & 1.0 & 1.0* & 1.0* & 1.0* & 1.0* & 1.0* & 1.0* & 1.0* & 1.0* & 1.0* & 1.0* & 1.0* \\
\cline{2-15}
&minver & 0.97 & 0.94 & 0.78 & 0.65 & 0.33 & 0.40 & 0.97 & 0.93 & 0.85 & 0.70 & 0.50 & 0.90 & 0.74 \\
\cline{2-15}
&petrinet & 1.0 & 1.0* & 1.0 & 1.0 & 1.0 & 1.0 & 1.0 & 1.0 & 1.0 & 1.0 & 1.0 & 1.0 & 1.0* \\
\cline{2-15}
&pm & 1.0* & 1.0* & 1.0* & 0.99 & 1.0* & 0.99 & 1.0* & 1.0* & 1.0* & 0.99 & 0.99 & 1.0* & 1.0* \\
\cline{2-15}
&powerwindow & 0.89 & 0.89 & 0.74 & 0.44 & 0.50 & 0.53 & 0.83 & 0.95 & 0.80 & 0.50 & 0.32 & 0.91 & 0.69 \\
\cline{2-15}
&prime & 0.99 & 0.73 & 0.90 & 0.77 & 0.31 & 0.74 & 0.94 & 0.94 & 0.92 & 0.77 & 0.62 & 0.97 & 0.80 \\
\cline{2-15}
&sha & 1.0* & 1.0* & 1.0* & 1.0* & 1.0* & 1.0* & 1.0* & 1.0* & 1.0 & 1.0* & 1.0* & 1.0 & 1.0* \\

\hline\hline
\multirow{3}{*}{\makecell[c]{Tasks that always\\hit the L1 cache\\except the first miss}} 

&bsort & 1.0 & 1.0 & 1.0 & 1.0 & 1.0 & 1.0 & 1.0 & 1.0 & 1.0 & 1.0 & 1.0 & 1.0 & 1.0 \\
\cline{2-15}
&cover & 1.0 & 1.0 & 1.0 & 1.0 & 1.0 & 1.0 & 1.0 & 1.0 & 1.0 & 1.0 & 1.0 & 1.0 & 1.0 \\
\cline{2-15}
&complex\_updates & 1.0 & 1.0 & 1.0 & 1.0 & 1.0 & 1.0 & 1.0 & 1.0 & 1.0 & 1.0 & 1.0 & 1.0 & 1.0 \\

\hline\hline
\multirow{7}{*}{\makecell[c]{Tasks that cannot\\hit the L2 cache\\according to\\the intra-core\\cache analysis}} 
&cjpeg\_wrbmp & 1.0 & 1.0 & 1.0 & 1.0 & 1.0 & 1.0 & 1.0 & 1.0 & 1.0 & 1.0 & 1.0 & 1.0 & 1.0 \\
\cline{2-15}
&countnegative & 1.0 & 1.0 & 1.0 & 1.0 & 1.0 & 1.0 & 1.0 & 1.0 & 1.0 & 1.0 & 1.0 & 1.0 & 1.0 \\
\cline{2-15}
&fft & 1.0 & 1.0 & 1.0 & 1.0 & 1.0 & 1.0 & 1.0 & 1.0 & 1.0 & 1.0 & 1.0 & 1.0 & 1.0 \\
\cline{2-15}
&filterbank & 1.0 & 1.0 & 1.0 & 1.0 & 1.0 & 1.0 & 1.0 & 1.0 & 1.0 & 1.0 & 1.0 & 1.0 & 1.0 \\
\cline{2-15}
&lms & 1.0 & 1.0 & 1.0 & 1.0 & 1.0 & 1.0 & 1.0 & 1.0 & 1.0 & 1.0 & 1.0 & 1.0 & 1.0 \\
\cline{2-15}
&ludcmp & 1.0 & 1.0 & 1.0 & 1.0 & 1.0 & 1.0 & 1.0 & 1.0 & 1.0 & 1.0 & 1.0 & 1.0 & 1.0 \\
\cline{2-15}
&matrix1 & 1.0 & 1.0 & 1.0 & 1.0 & 1.0 & 1.0 & 1.0 & 1.0 & 1.0 & 1.0 & 1.0 & 1.0 & 1.0 \\

\hline

\end{tabular}}
\end{table*}

\textbf{Detailed Comparison of Each Evaluated Task.} Third, Fig.~\ref{fig:adpcm_dec} to~\ref{fig:matrix1} present the absolute value of the inter-core cache interference and the total WCET estimations of the 36 tasks under different interfering tasks. 
The results highlight the portion of the inter-core cache interference in the total WCET of each task, which validates the results presented above and the effectiveness of the proposed analysis. 

\begin{table}[htbp]
\centering
\caption{Tasks of program FLY-BY-WIRE and AUTOPILOT in the PapaBench suite.}
\label{tab:task_frequencies}
\resizebox{.5\textwidth}{!}{
\begin{tabular}{|l|l|}
\hline
Tasks in AUTOPILOT  & Tasks in FLY-BY-WIRE  \\
\hline\hline
altitude\_control\_task &check\_failsafe\_task \\
climb\_control\_task & check\_mega128\_values\_task \\
link\_fbw\_send & \textbf{send\_data\_to\_autopilot\_task}  \\
stabilisation\_task&servo\_transmit   \\
\textbf{radio\_control\_task} & test\_ppm\_task   \\
\hline
\end{tabular}
}
\vspace{-5pt}
\end{table}


\textbf{Evaluation of Papabench.} Finally, the tasks in Papabench are analysed using the proposed analysis and Zhang2022. 
The Papabench produces two executable files (Fly-By-Wire and Autopilot), representing two applications running on two processors. 
Each application invokes a set of different tasks, as shown in Tab.~\ref{tab:task_frequencies}. 
However, we observe that most tasks (8 out of 10) in Papabench can always hit the L1 cache (except the first miss). Thus, similar to the second group of the TacleBench tasks evaluated above, both the proposed analysis and the analysis in~\cite{zhang2022precise} return an inter-core cache interference of zero with an identical WCET, and hence, their results are omitted. 

The task \texttt{radio\_control\_task} (in AUTOPILOT) and task \texttt{send\_data\_to\_autopilot\_task} (in FLY-BY-WIRE) would incur L2 cache evictions due to interfering tasks. 
Fig.~\ref{fig:papa_radio}
presents the inter-core cache interference and the total WCET for \texttt{radio\_control\_task} when being executed under the interference from each of the five tasks in FLY-BY-WIRE.
Fig.~\ref{fig:papa_send} shows the results of \texttt{send\_data\_to\_autopilot\_task} under the interference of each task in Autopilot.
As observed, for task \texttt{radio_control_task}, the proposed analysis outperforms Zhang2022 by reducing the inter-core cache interference and total WCET estimations by 87.1\% and 7.56\%, respectively. 
For task \texttt{send\_data\_to\_autopilot\_task}, the proposed analysis outperforms Zhang2022 by reducing the inter-core cache interference and total WCET estimations by 62.42\% and 2.78\%, respectively.

\clearpage
\twocolumn
\raggedbottom

\begin{figure}[t]
    \centering
    \includegraphics[width=1\linewidth]{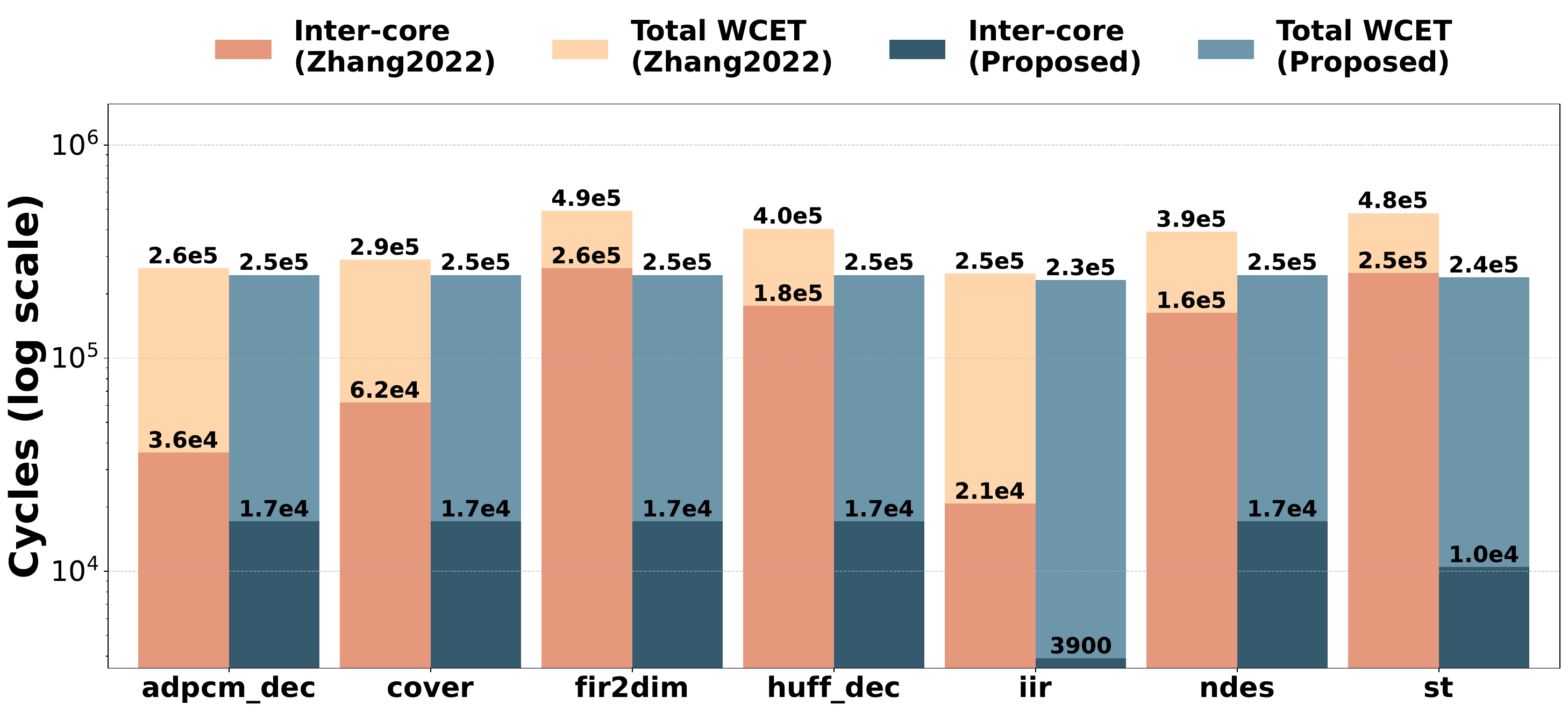}
    \caption{Inter-core cache interference and total WCET  of \texttt{adpcm\_dec} with $\kappa=2$ \textit{(y-axis: interfering task; dark colours: inter-core cache interference; light colours: WCET)}.}
    \label{fig:adpcm_dec}
\end{figure}

\begin{figure}[t]
    \centering
    \includegraphics[width=1\linewidth]{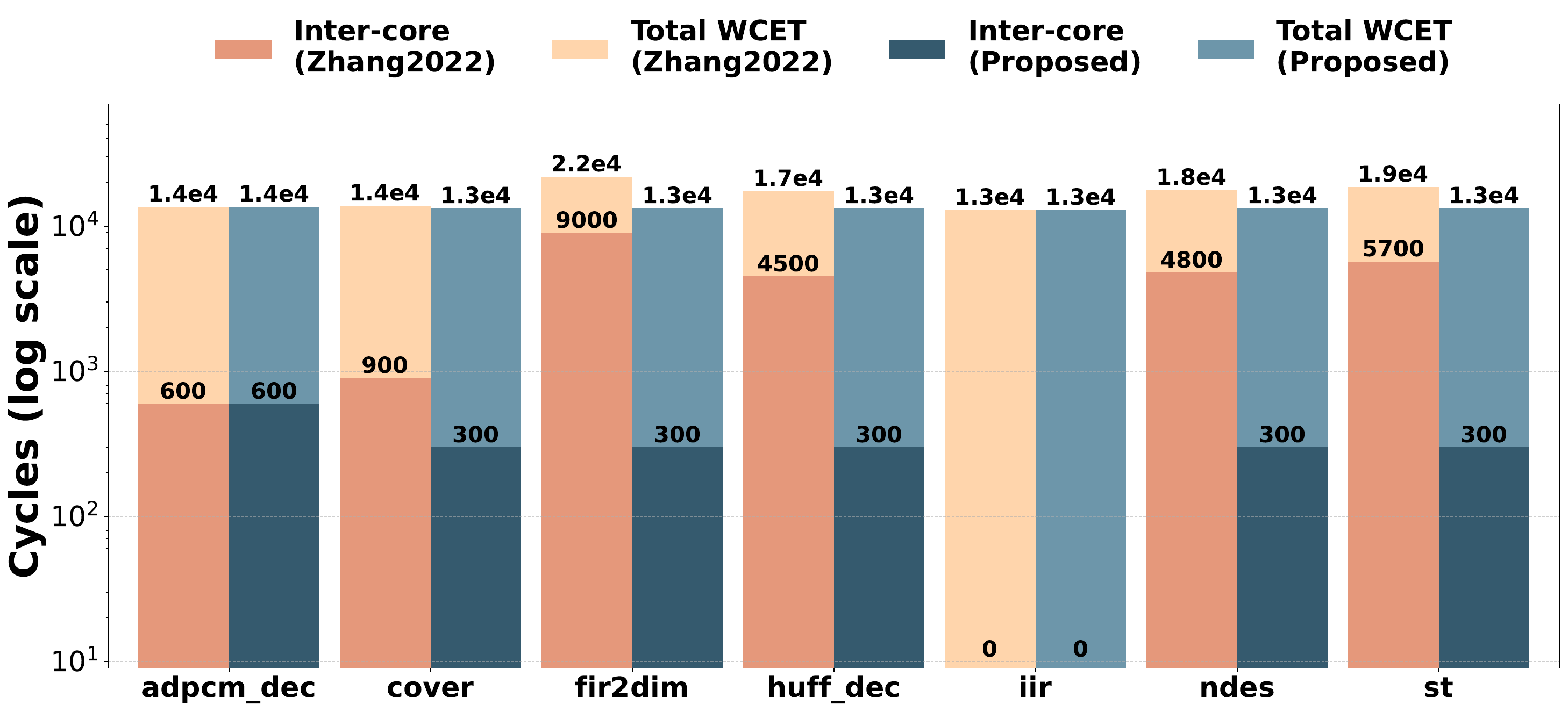}
    \caption{Inter-core cache interference and total WCET of \texttt{binarysearch} with $\kappa=2$ \textit{(y-axis: interfering task; dark colours: inter-core cache interference; light colours: WCET)}.}
    \label{fig:binarysearch}
\end{figure}

\begin{figure}[t]
\centering
\includegraphics[width=1\linewidth]{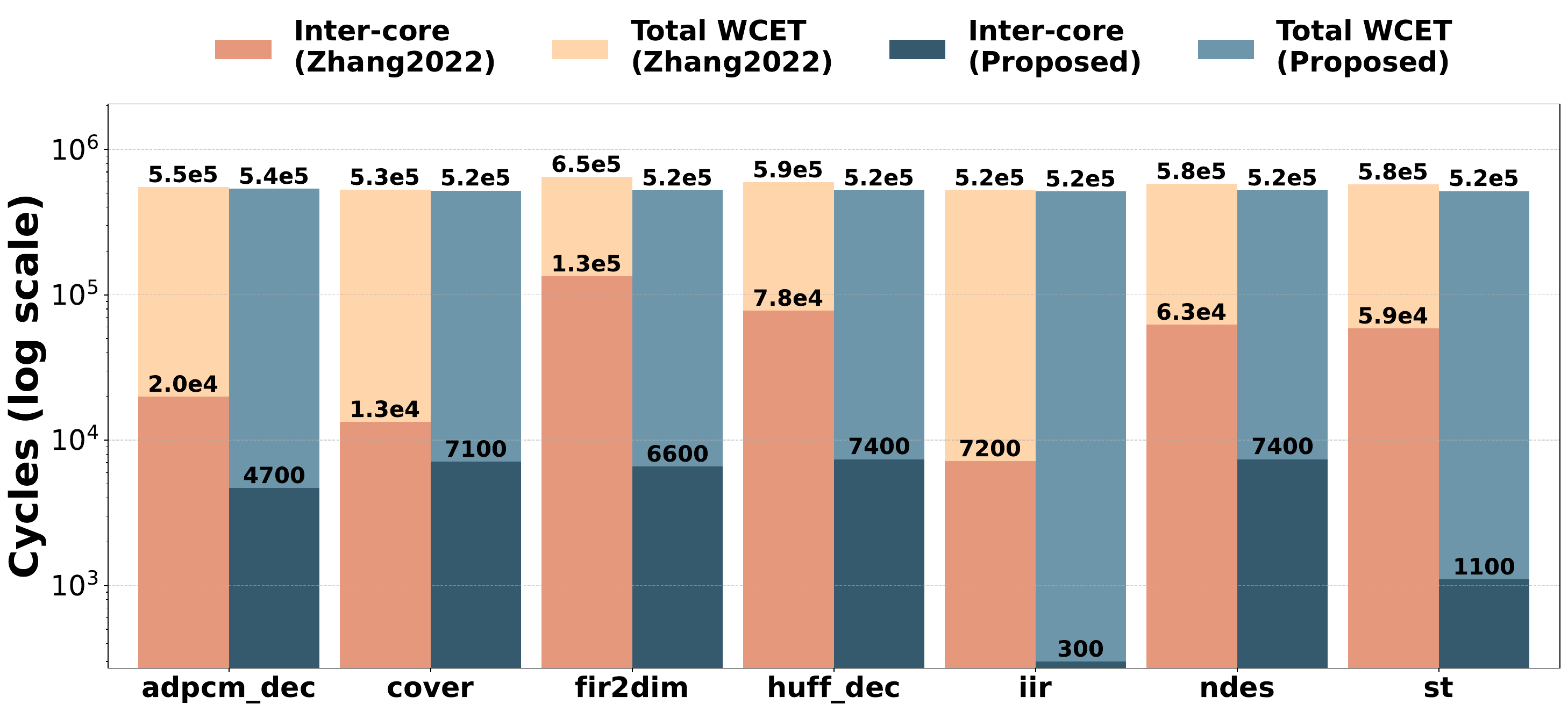}
\caption{Inter-core cache interference and total WCET  of \texttt{fir2dim} with $\kappa=2$ \textit{(y-axis: interfering task; dark colours: inter-core cache interference; light colours: WCET)}.}
\label{fig:fir2dim}
\end{figure}

\begin{figure}[t]
    \centering
    \includegraphics[width=1\linewidth]{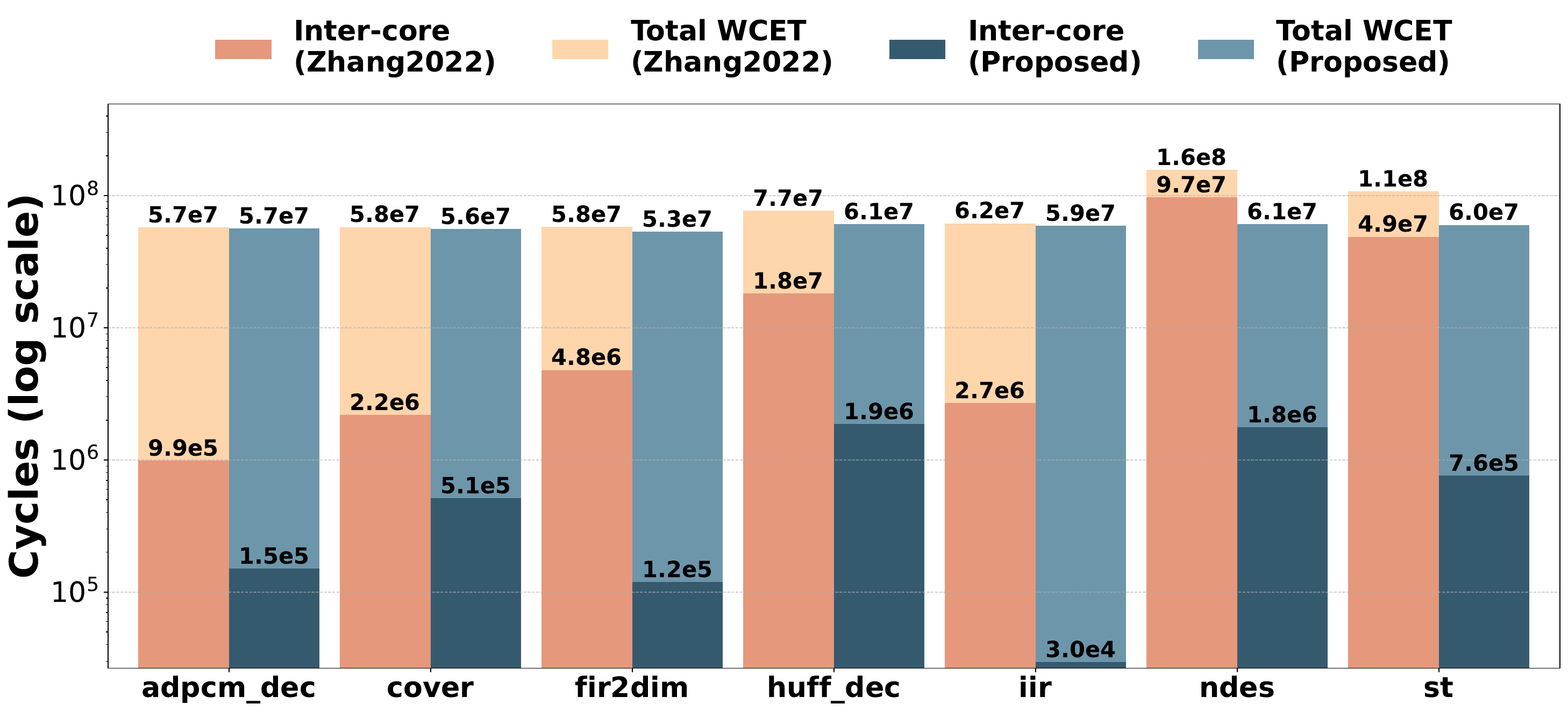}
    \caption{Inter-core cache interference and total WCET  of \texttt{fmref} with $\kappa=2$ \textit{(y-axis: interfering task; dark colours: inter-core cache interference; light colours: WCET)}.}
    \label{fig:fmref}
\end{figure}

\begin{figure}[t]
    \centering
    \includegraphics[width=1\linewidth]{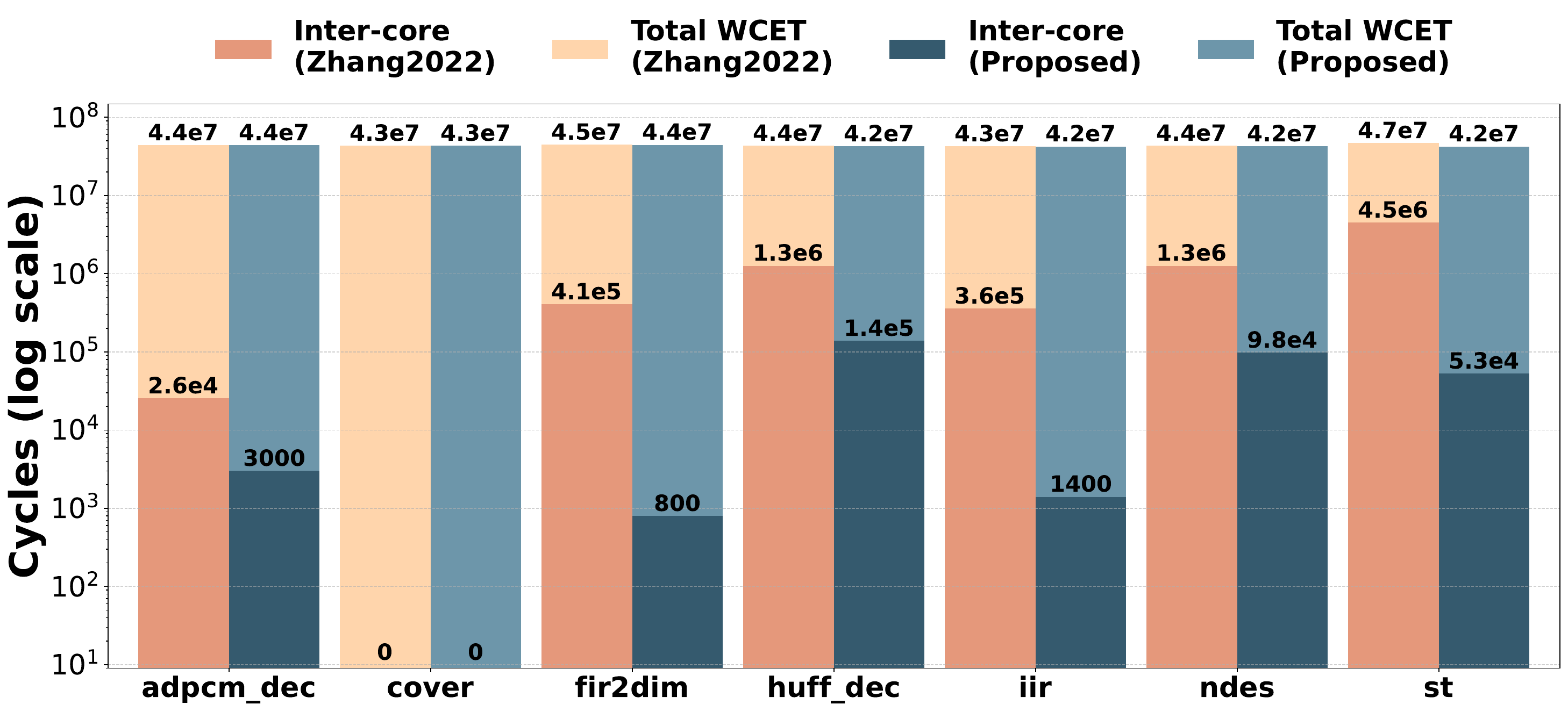}
    \caption{Inter-core cache interference and total WCET  of \texttt{huff\_dec} with $\kappa=2$ \textit{(y-axis: interfering task; dark colours: inter-core cache interference; light colours: WCET)}.}
    \label{fig:huff_dec}
\end{figure}

\begin{figure}[t]
    \centering
    \includegraphics[width=1\linewidth]{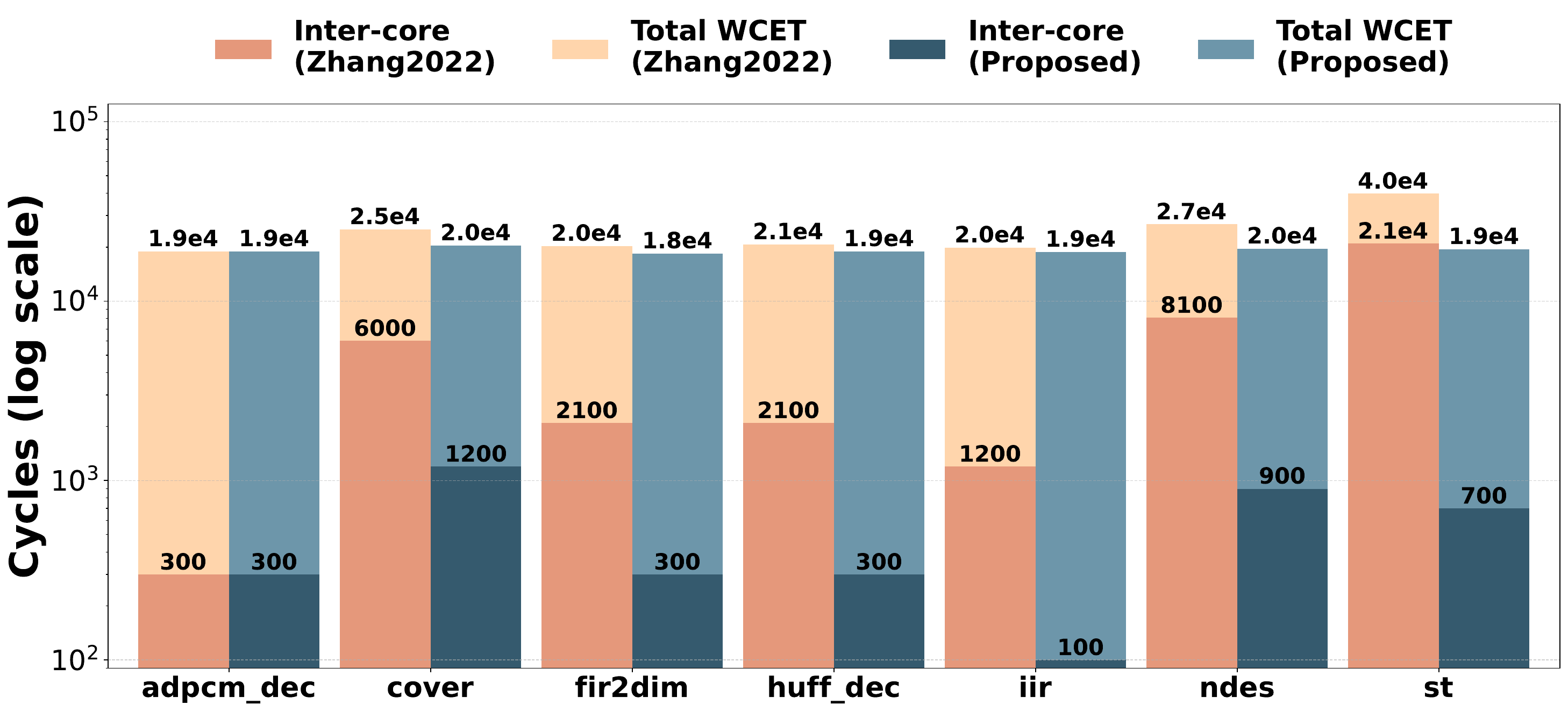}
    \caption{Inter-core cache interference and total WCET  of \texttt{iir} with $\kappa=2$ \textit{(y-axis: interfering task; dark colours: inter-core cache interference; light colours: WCET)}.}
    \label{fig:iir}
\end{figure}

\begin{figure}[t]
    \centering
    \includegraphics[width=1\linewidth]{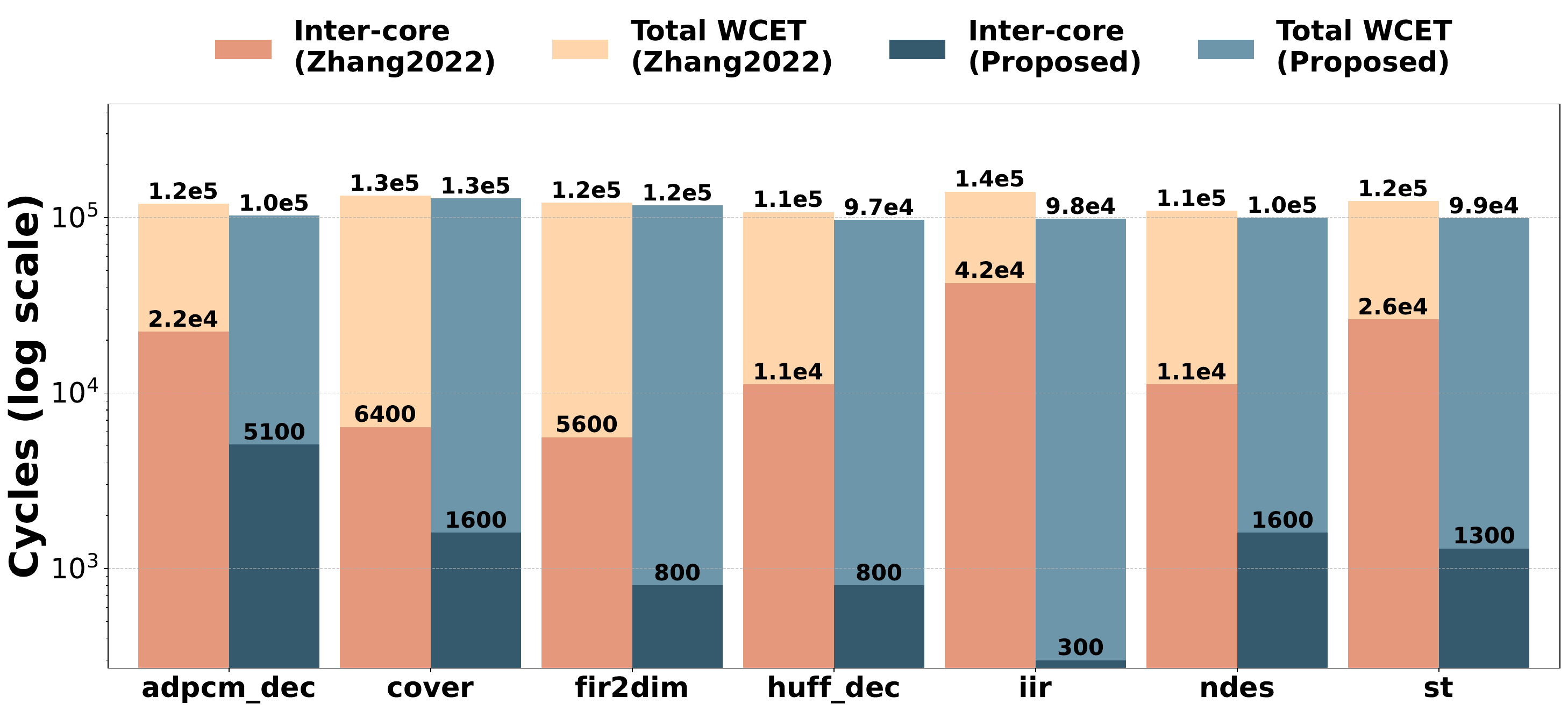}
    \caption{Inter-core cache interference and total WCET  of \texttt{insertsort} with $\kappa=2$ \textit{(y-axis: interfering task; dark colours: inter-core cache interference; light colours: WCET)}.}
    \label{fig:insertsort}
\end{figure}

\begin{figure}[t]
    \centering
    \includegraphics[width=1\linewidth]{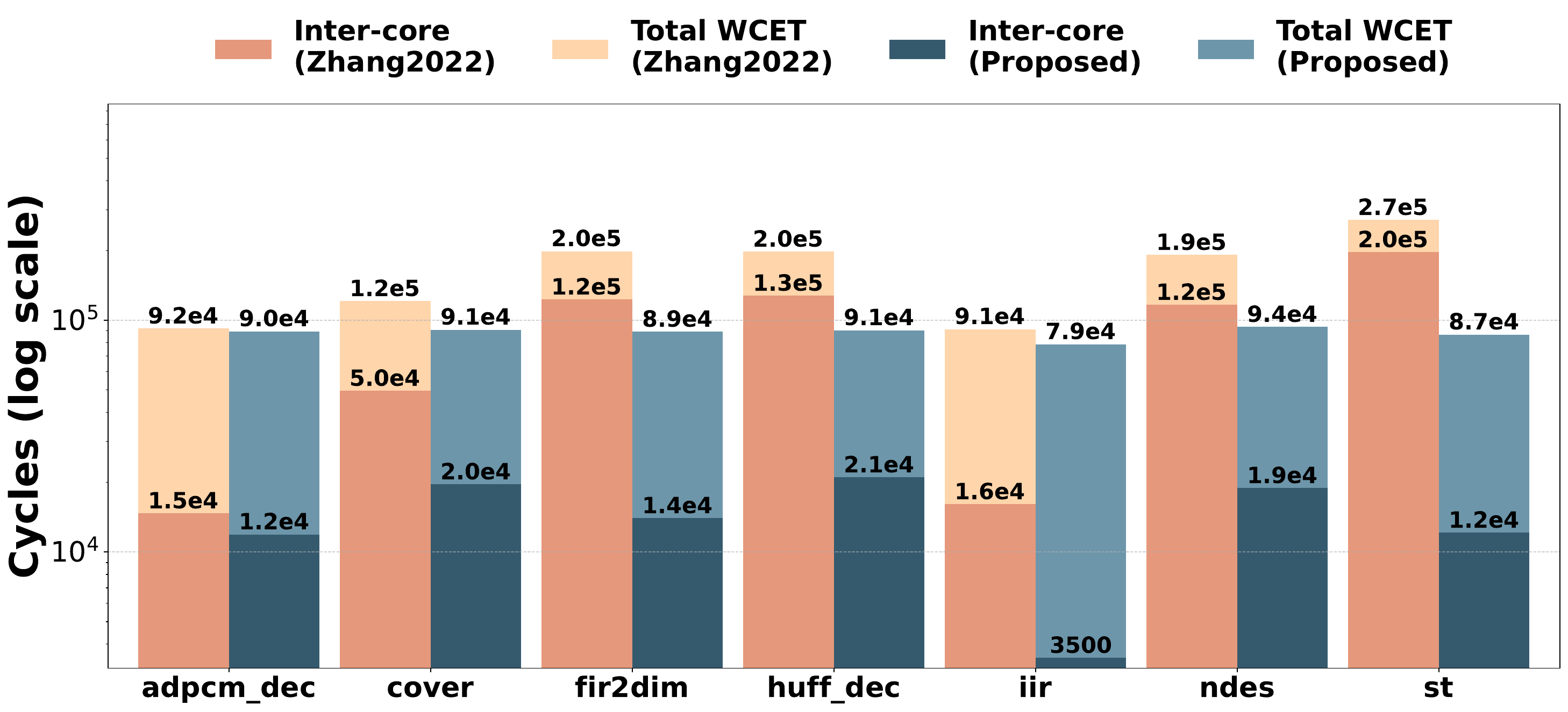}
    \caption{Inter-core cache interference and total WCET  of \texttt{jfdctint} with $\kappa=2$ \textit{(y-axis: interfering task; dark colours: inter-core cache interference; light colours: WCET)}.}
    \label{fig:jfdctint}
\end{figure}

\begin{figure}[t]
    \centering
    \includegraphics[width=1\linewidth]{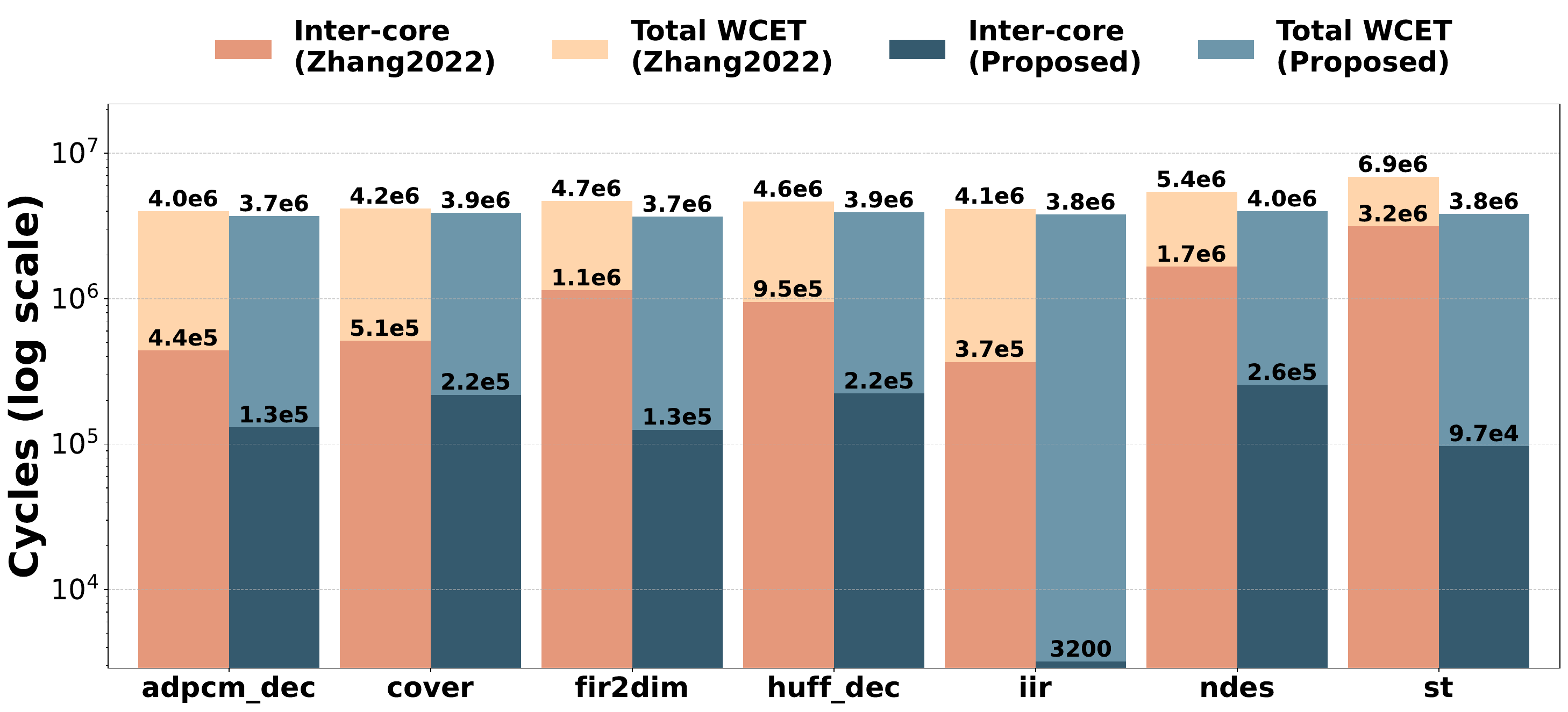}
    \caption{Inter-core cache interference and total WCET  of \texttt{ndes} with $\kappa=2$ \textit{(y-axis: interfering task; dark colours: inter-core cache interference; light colours: WCET)}.}
    \label{fig:ndes}
\end{figure}

\begin{figure}[t]
    \centering
    \includegraphics[width=1\linewidth]{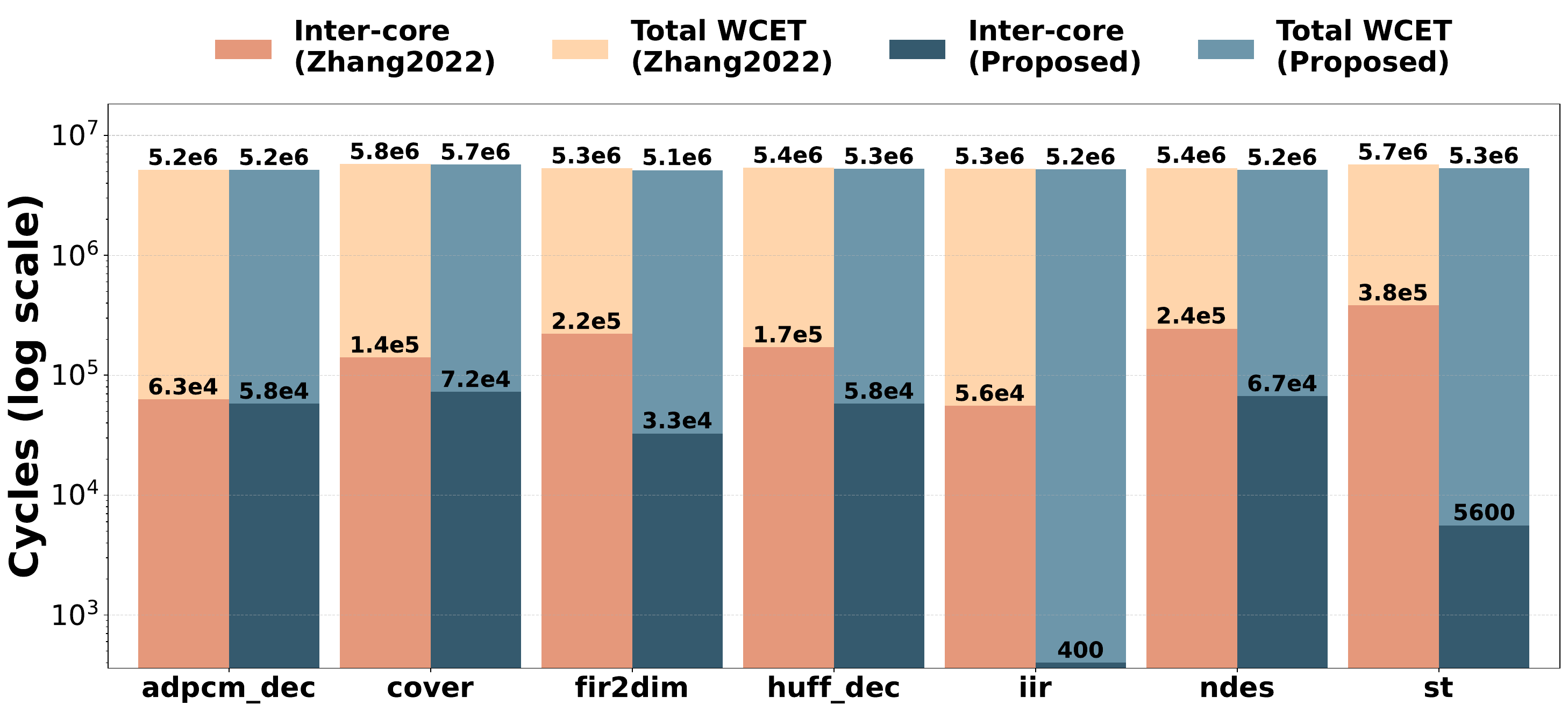}
    \caption{Inter-core cache interference and total WCET  of \texttt{st} with $\kappa=2$ \textit{(y-axis: interfering task; dark colours: inter-core cache interference; light colours: WCET)}.}
    \label{fig:st}
\end{figure}

\begin{figure}[t]
    \centering
    \includegraphics[width=1\linewidth]{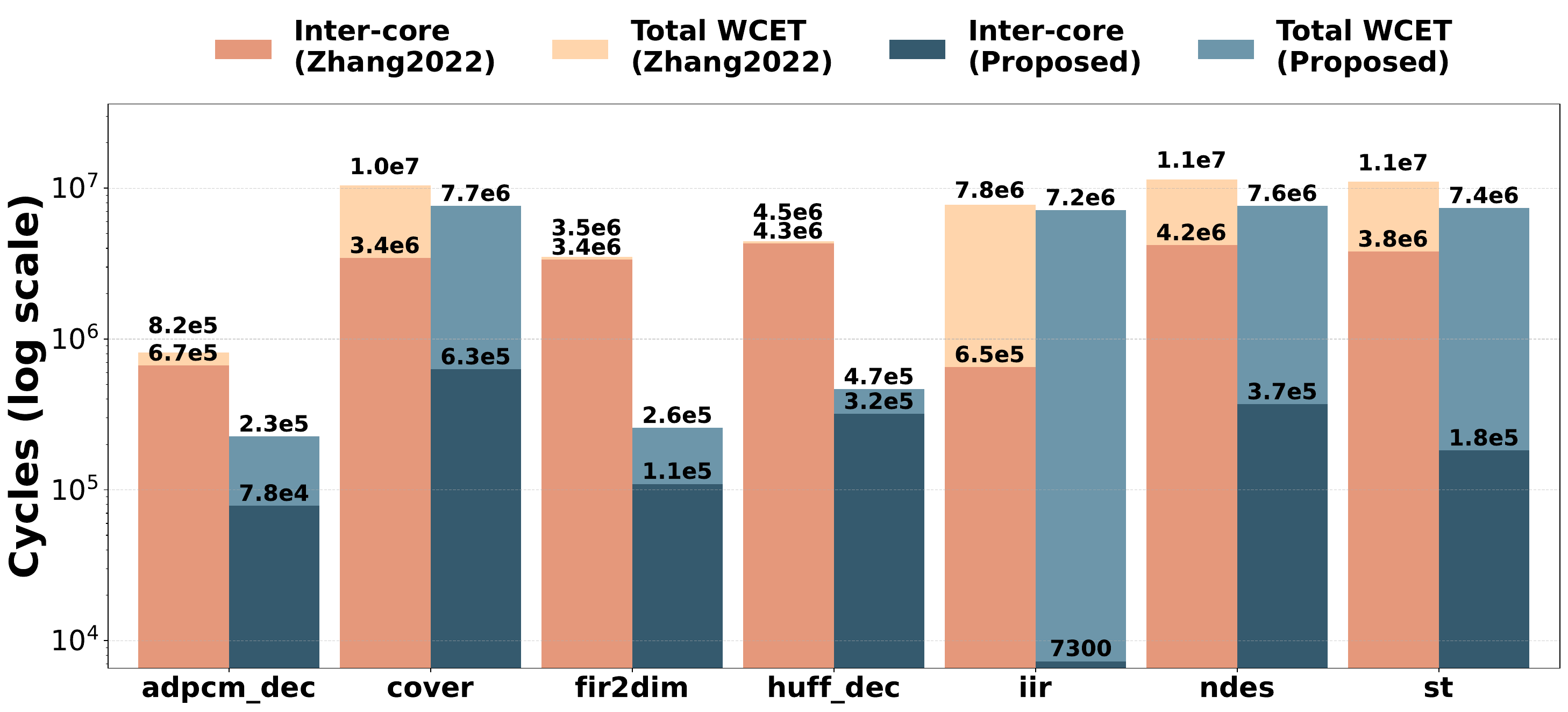}
    \caption{Inter-core cache interference and total WCET  of \texttt{statemate} with $\kappa=2$ \textit{(y-axis: interfering task; dark colours: inter-core cache interference; light colours: WCET)}.}
    \label{fig:statemate}
\end{figure}

\begin{figure}[t]
    \centering
    \includegraphics[width=1\linewidth]{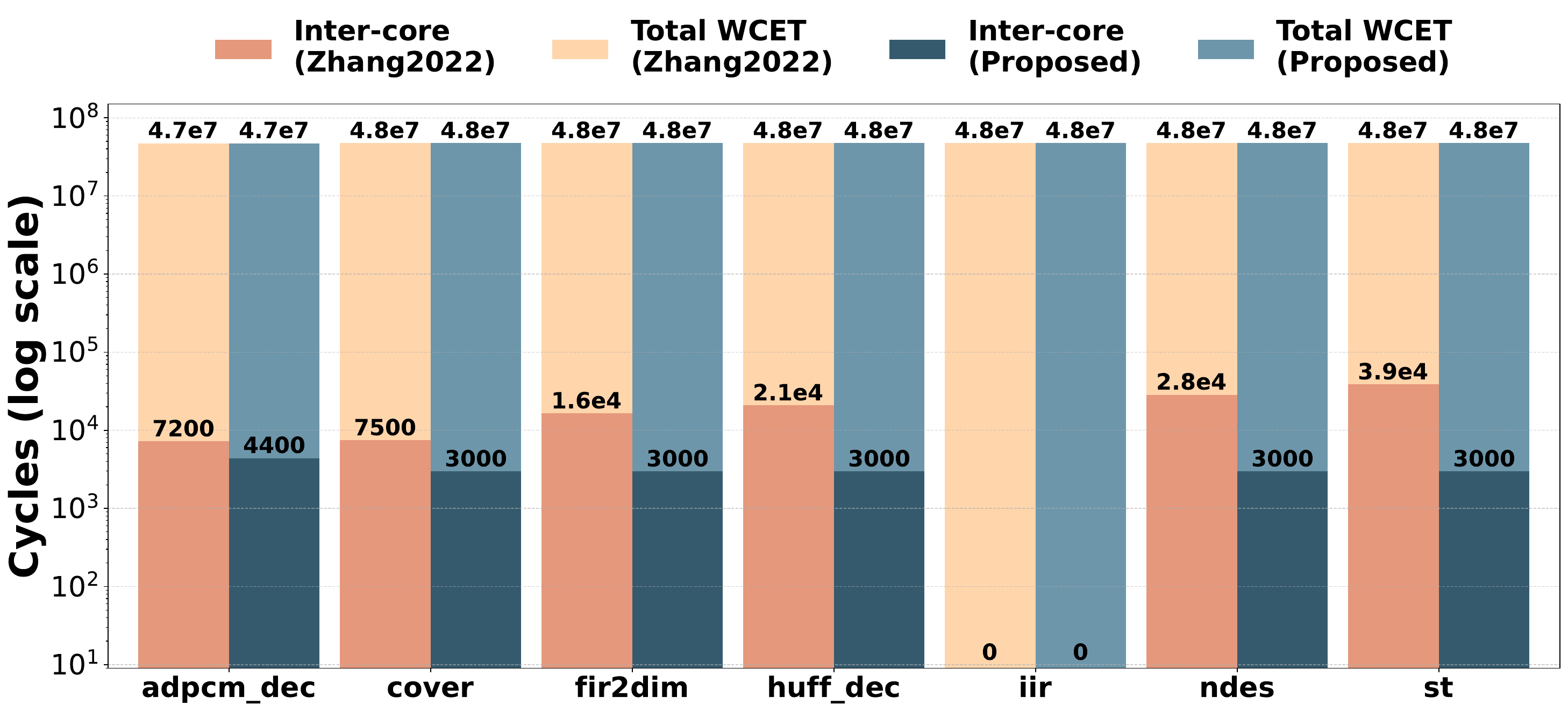}
    \caption{Inter-core cache interference and total WCET  of \texttt{audiobeam} with $\kappa=2$ \textit{(y-axis: interfering task; dark colours: inter-core cache interference; light colours: WCET)}.}
    \label{fig:audiobeam}
\end{figure}

\begin{figure}[t]
    \centering
    \includegraphics[width=1\linewidth]{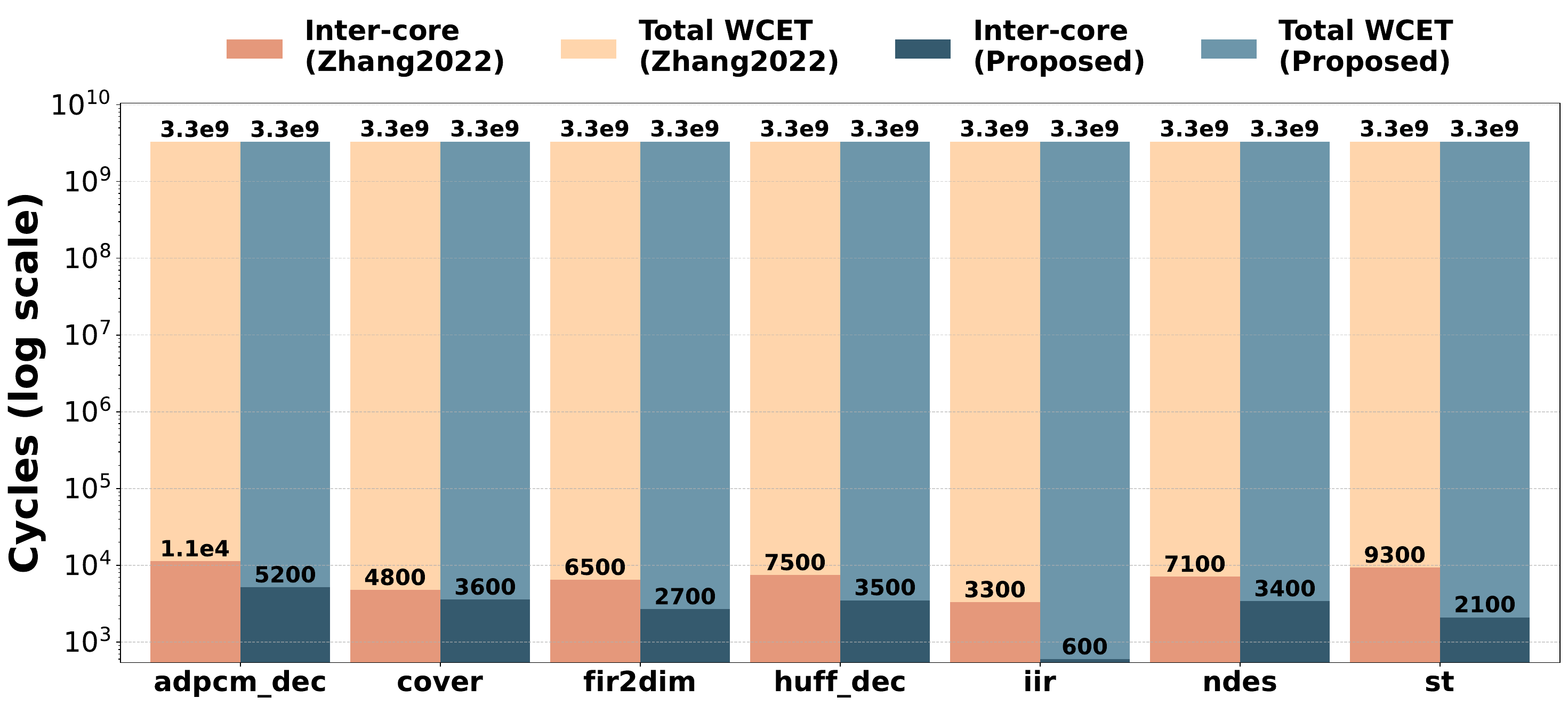}
    \caption{Inter-core cache interference and total WCET  of \texttt{cjpeg\_transupp} with $\kappa=2$ \textit{(y-axis: interfering task; dark colours: inter-core cache interference; light colours: WCET)}.}
    \label{fig:cjpeg_transupp}
\end{figure}

\begin{figure}[t]
    \centering
    \includegraphics[width=1\linewidth]{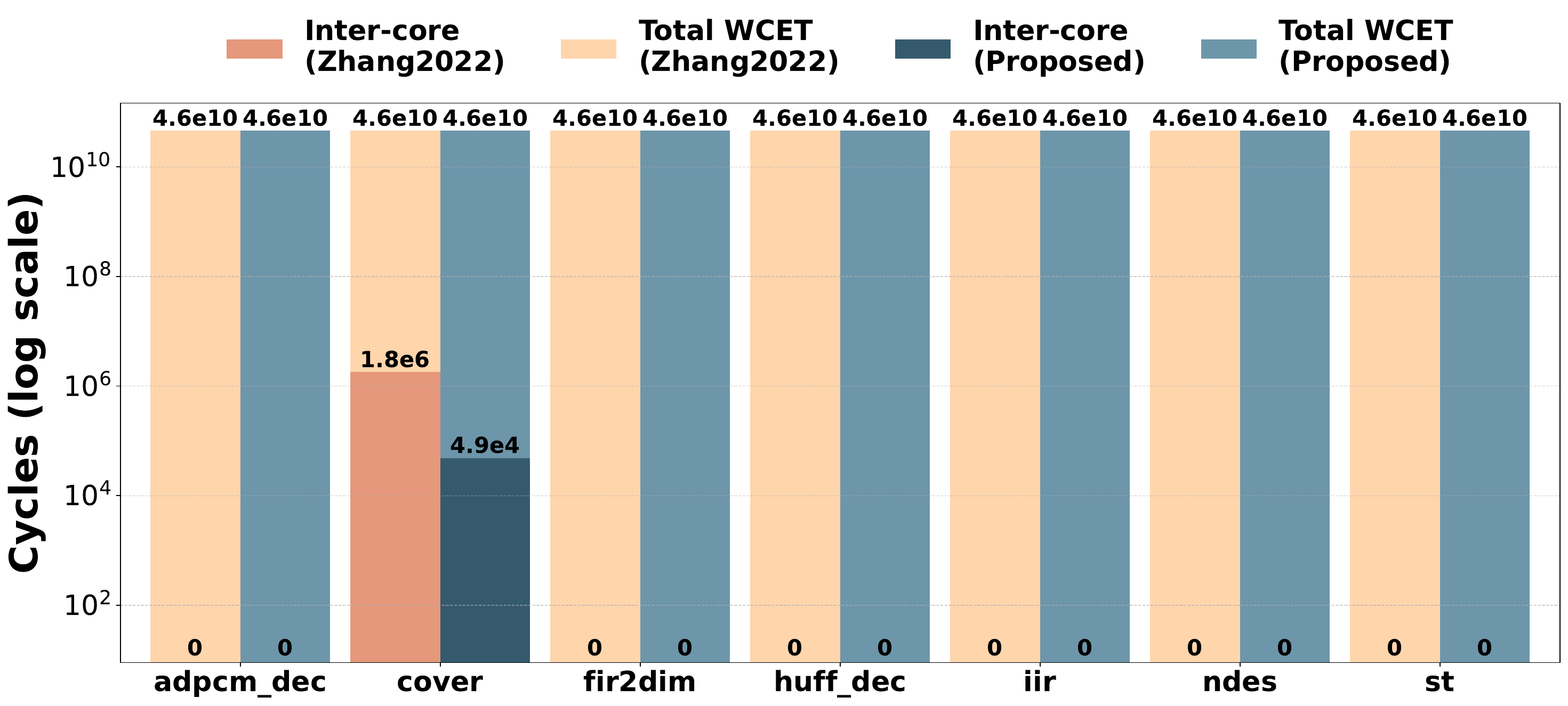}
    \caption{Inter-core cache interference and total WCET  of \texttt{dijkstra} with $\kappa=2$ \textit{(y-axis: interfering task; dark colours: inter-core cache interference; light colours: WCET)}.}
    \label{fig:dijkstra}
\end{figure}

\begin{figure}[t]
    \centering
    \includegraphics[width=1\linewidth]{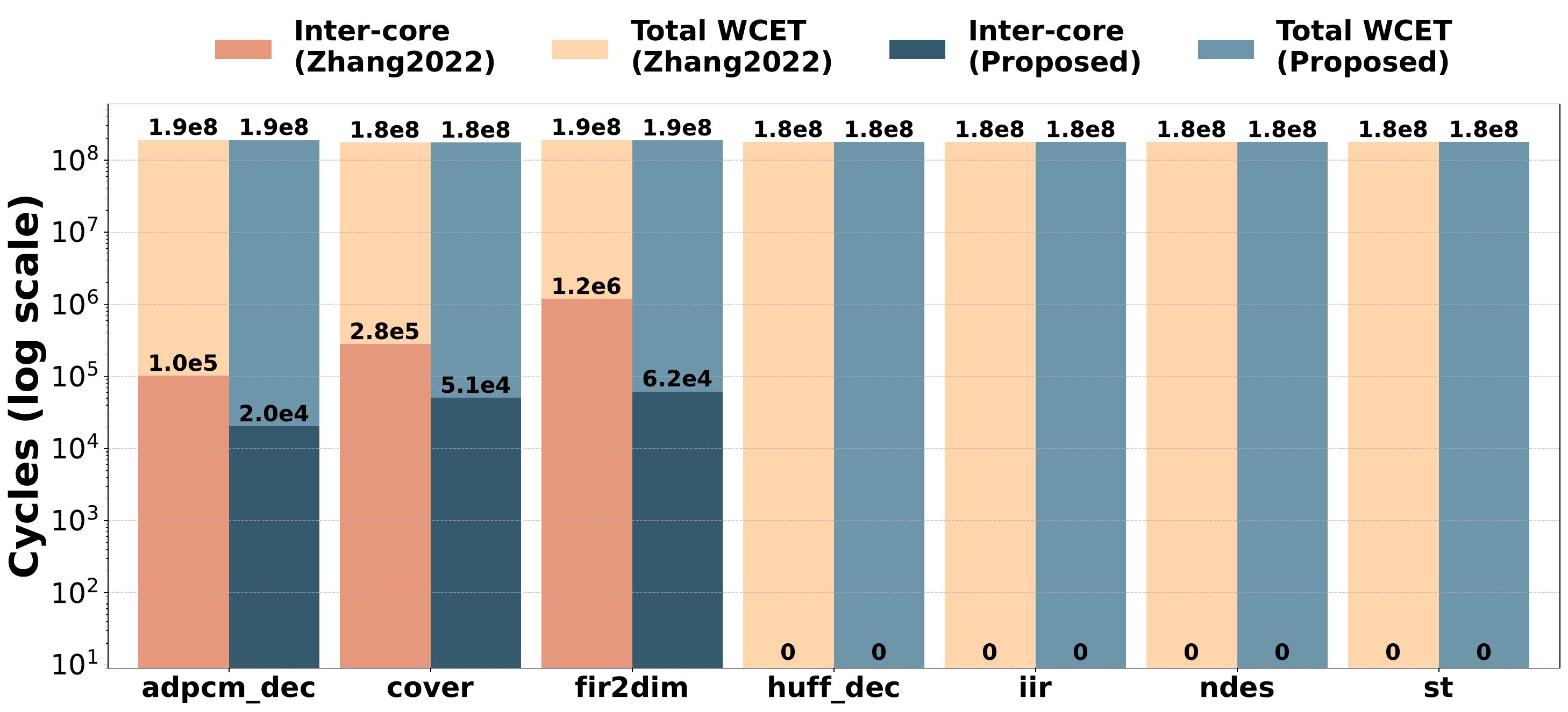}
    \caption{Inter-core cache interference and total WCET  of \texttt{g723\_enc} with $\kappa=2$ \textit{(y-axis: interfering task; dark colours: inter-core cache interference; light colours: WCET)}.}
    \label{fig:g723_enc}
\end{figure}

\begin{figure}[t]
    \centering
    \includegraphics[width=1\linewidth]{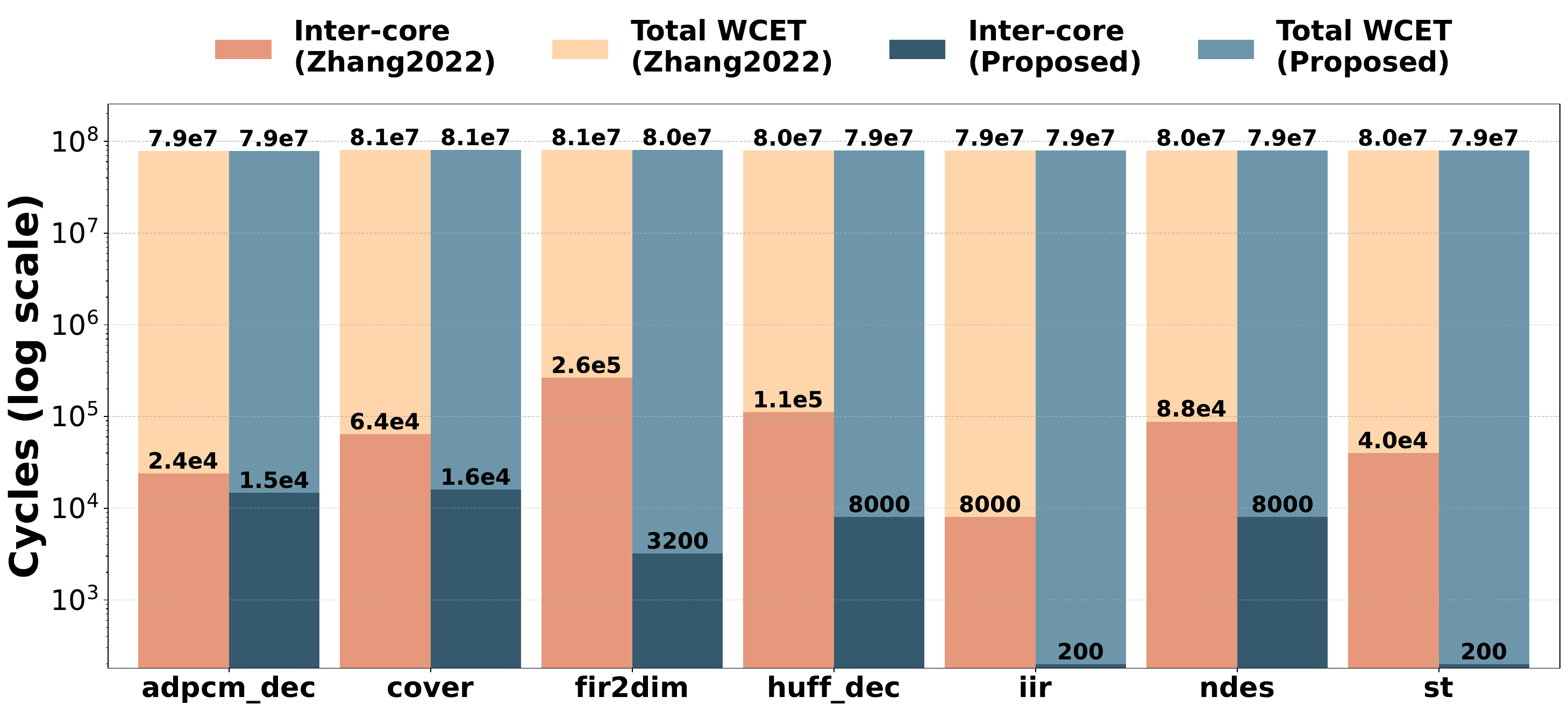}
    \caption{Inter-core cache interference and total WCET  of \texttt{gsm\_dec} with $\kappa=2$ \textit{(y-axis: interfering task; dark colours: inter-core cache interference; light colours: WCET)}.}
    \label{fig:gsm_dec}
\end{figure}

\begin{figure}[t]
    \centering
    \includegraphics[width=1\linewidth]{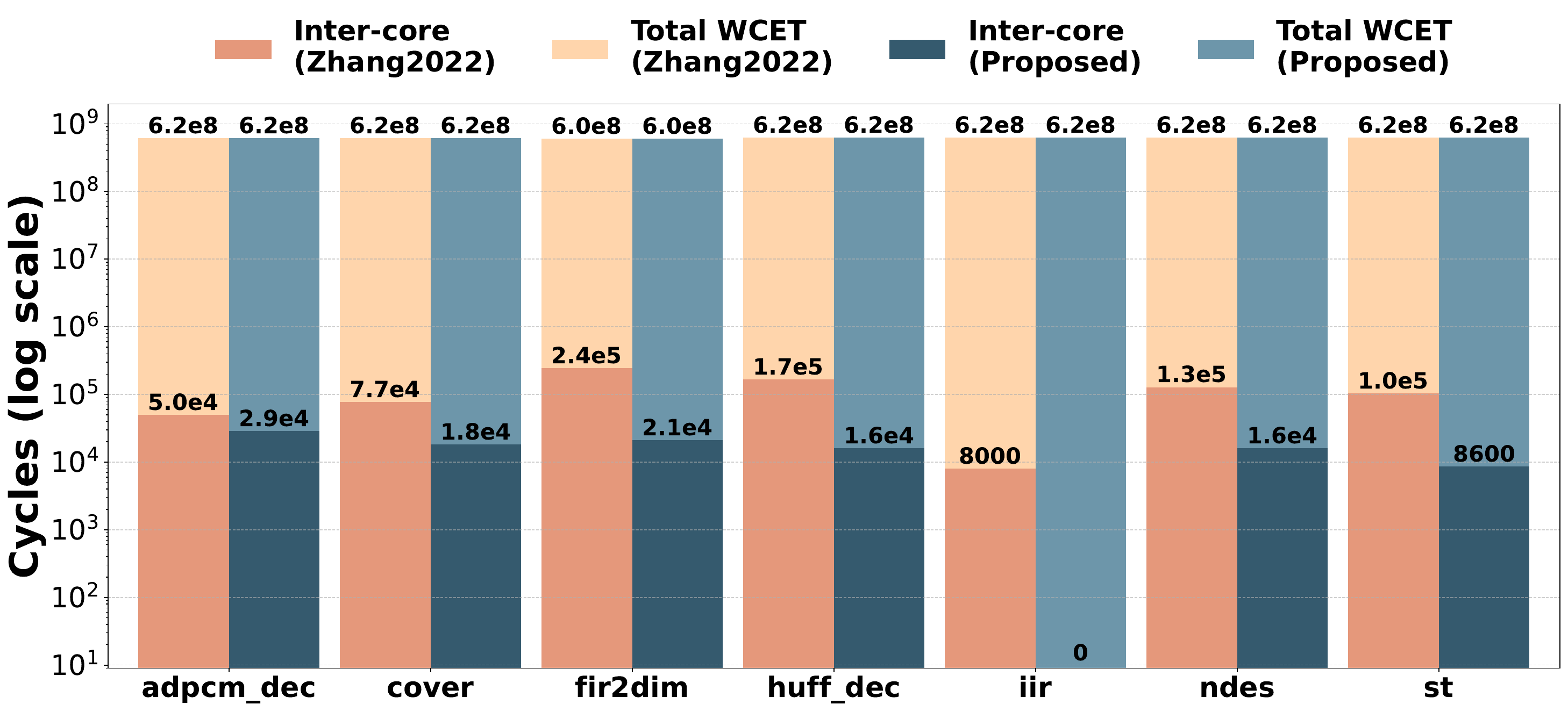}
    \caption{Inter-core cache interference and total WCET  of \texttt{gsm\_enc} with $\kappa=2$ \textit{(y-axis: interfering task; dark colours: inter-core cache interference; light colours: WCET)}.}
    \label{fig:gsm_enc}
\end{figure}

\begin{figure}[t]
    \centering
    \includegraphics[width=1\linewidth]{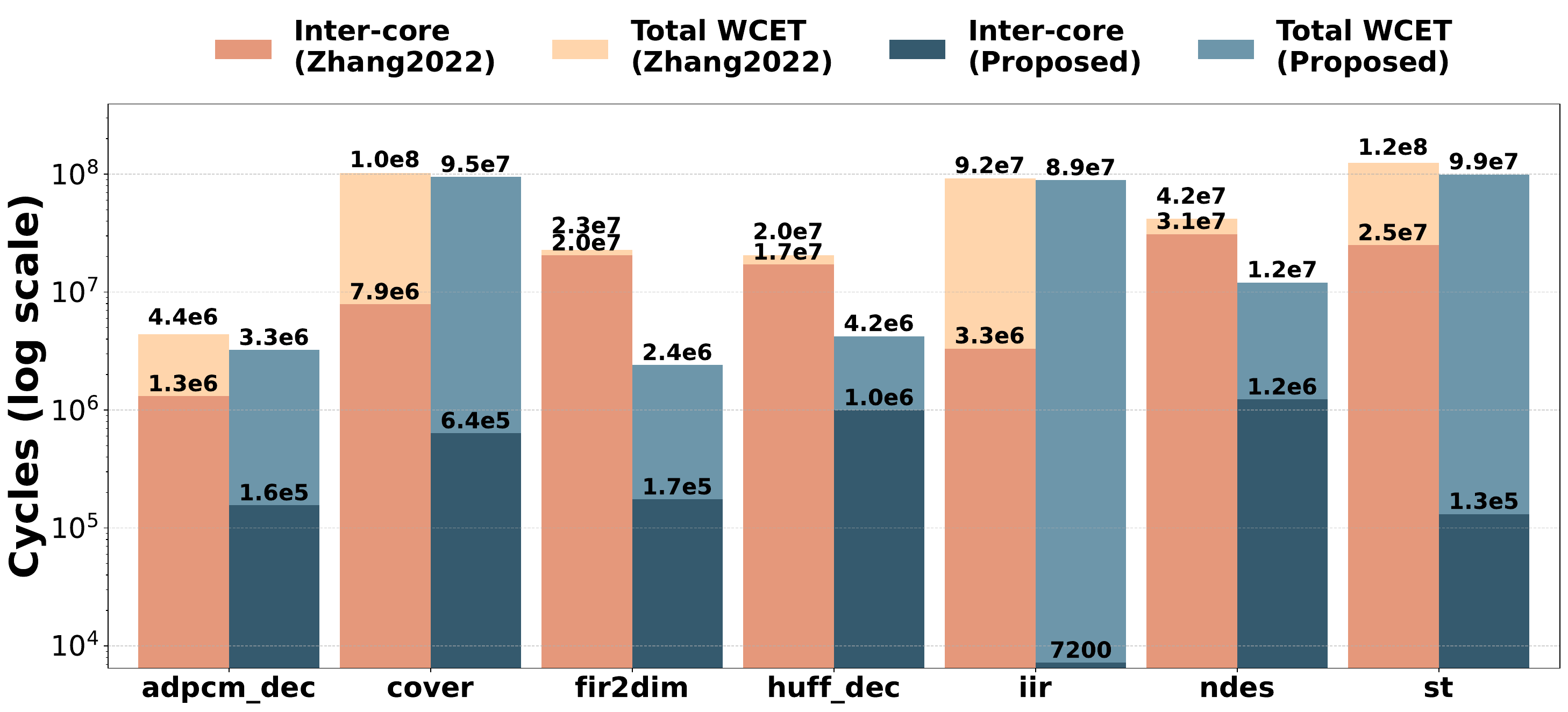}
    \caption{Inter-core cache interference and total WCET  of \texttt{h264\_dec} with $\kappa=2$ \textit{(y-axis: interfering task; dark colours: inter-core cache interference; light colours: WCET)}.}
    \label{fig:h264_dec}
\end{figure}

\begin{figure}[t]
    \centering
    \includegraphics[width=1\linewidth]{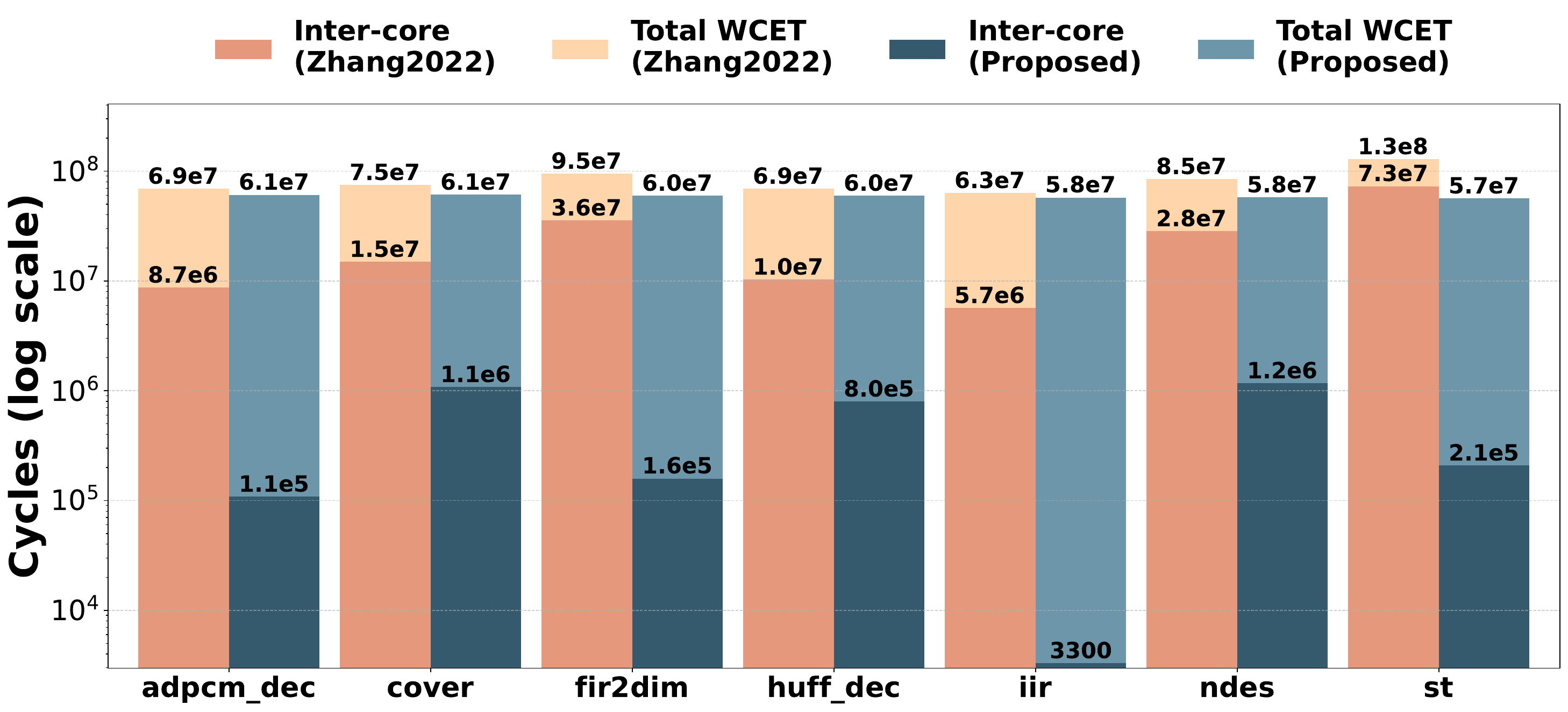}
    \caption{Inter-core cache interference and total WCET  of \texttt{lift} with $\kappa=2$ \textit{(y-axis: interfering task; dark colours: inter-core cache interference; light colours: WCET)}.}
    \label{fig:lift}
\end{figure}

\begin{figure}[t]
    \centering
    \includegraphics[width=1\linewidth]{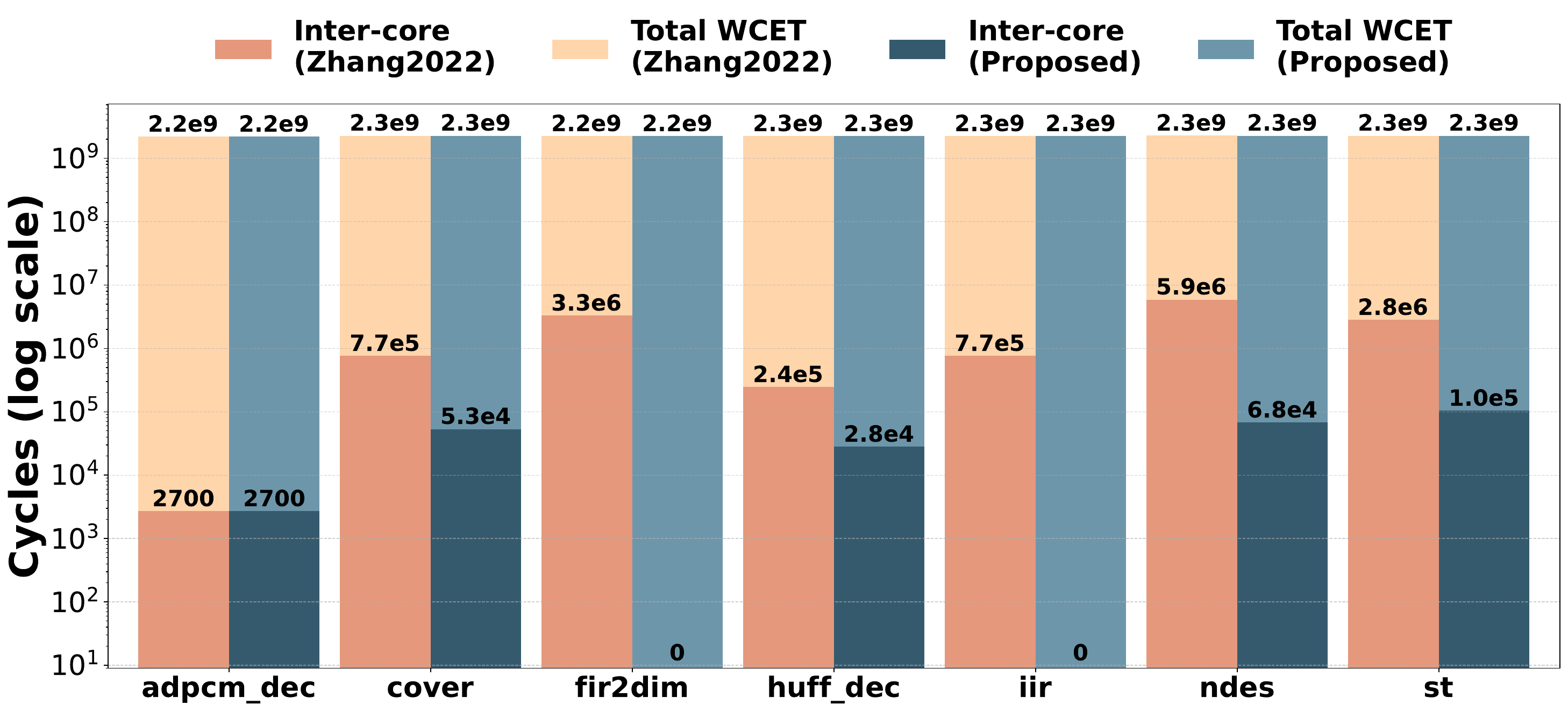}
    \caption{Inter-core cache interference and total WCET  of \texttt{md5} with $\kappa=2$ \textit{(y-axis: interfering task; dark colours: inter-core cache interference; light colours: WCET)}.}
    \label{fig:md5}
\end{figure}

\begin{figure}[t]
    \centering
    \includegraphics[width=1\linewidth]{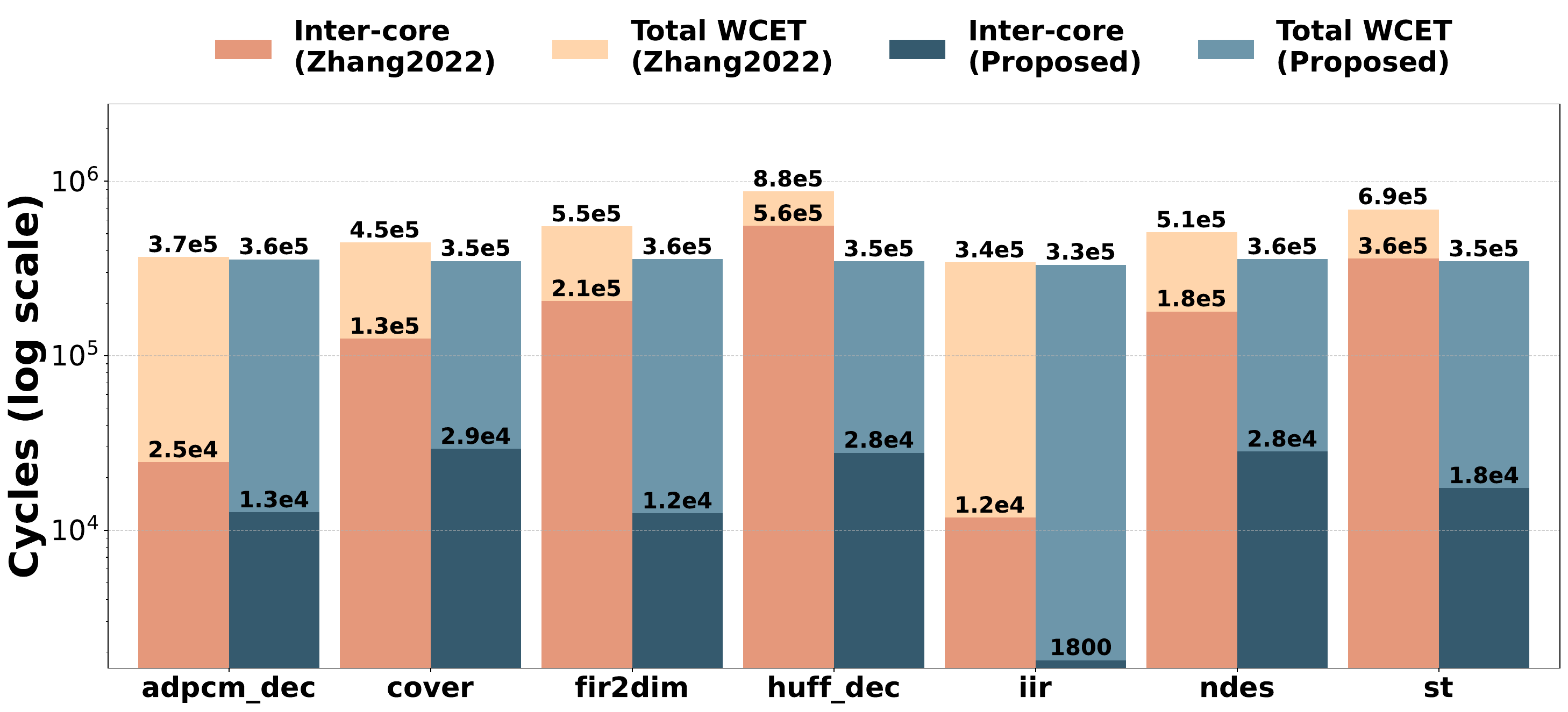}
    \caption{Inter-core cache interference and total WCET  of \texttt{minver} with $\kappa=2$ \textit{(y-axis: interfering task; dark colours: inter-core cache interference; light colours: WCET)}.}
    \label{fig:minver}
\end{figure}

\begin{figure}[t]
    \centering
    \includegraphics[width=1\linewidth]{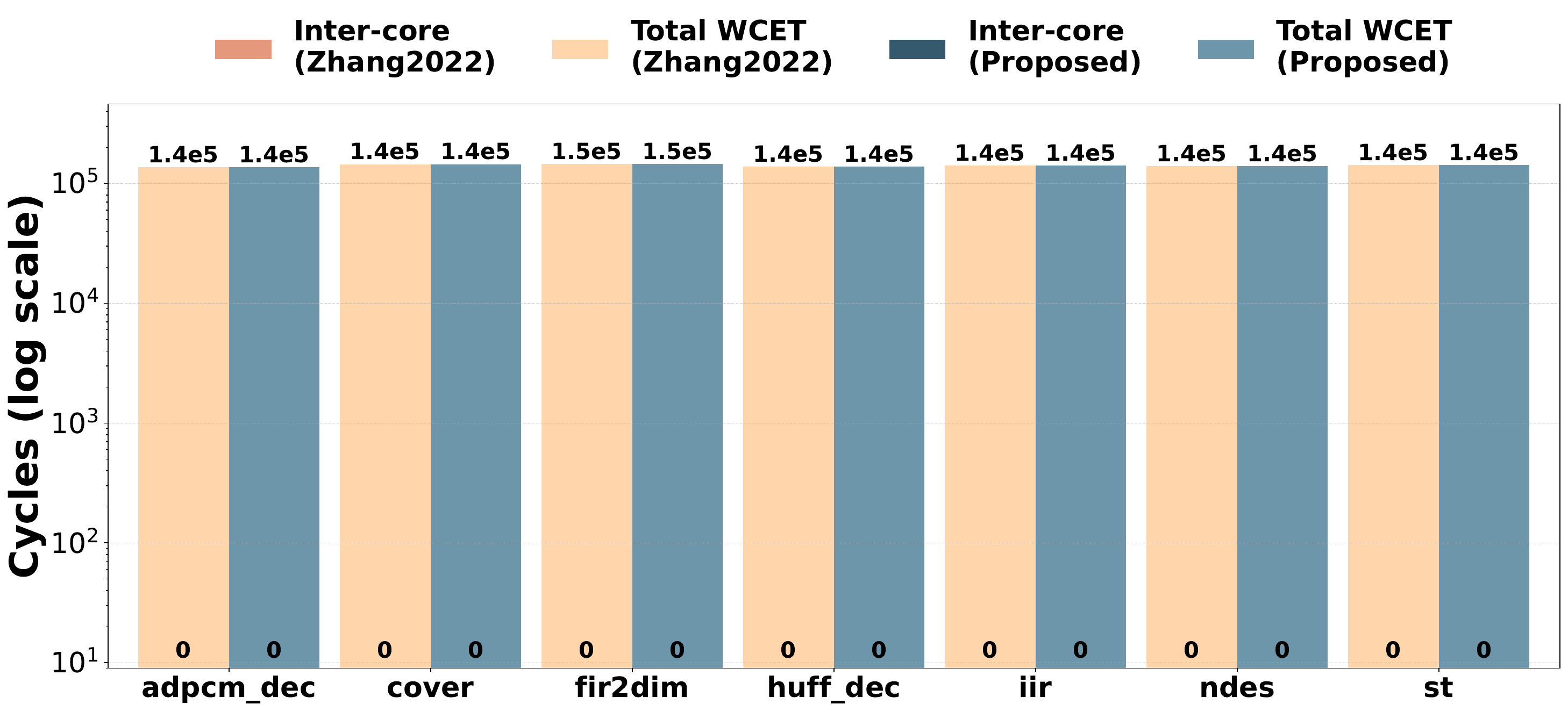}
    \caption{Inter-core cache interference and total WCET  of \texttt{petrinet} with $\kappa=2$ \textit{(y-axis: interfering task; dark colours: inter-core cache interference; light colours: WCET)}.}
    \label{fig:petrinet}
\end{figure}

\begin{figure}[t]
    \centering
    \includegraphics[width=1\linewidth]{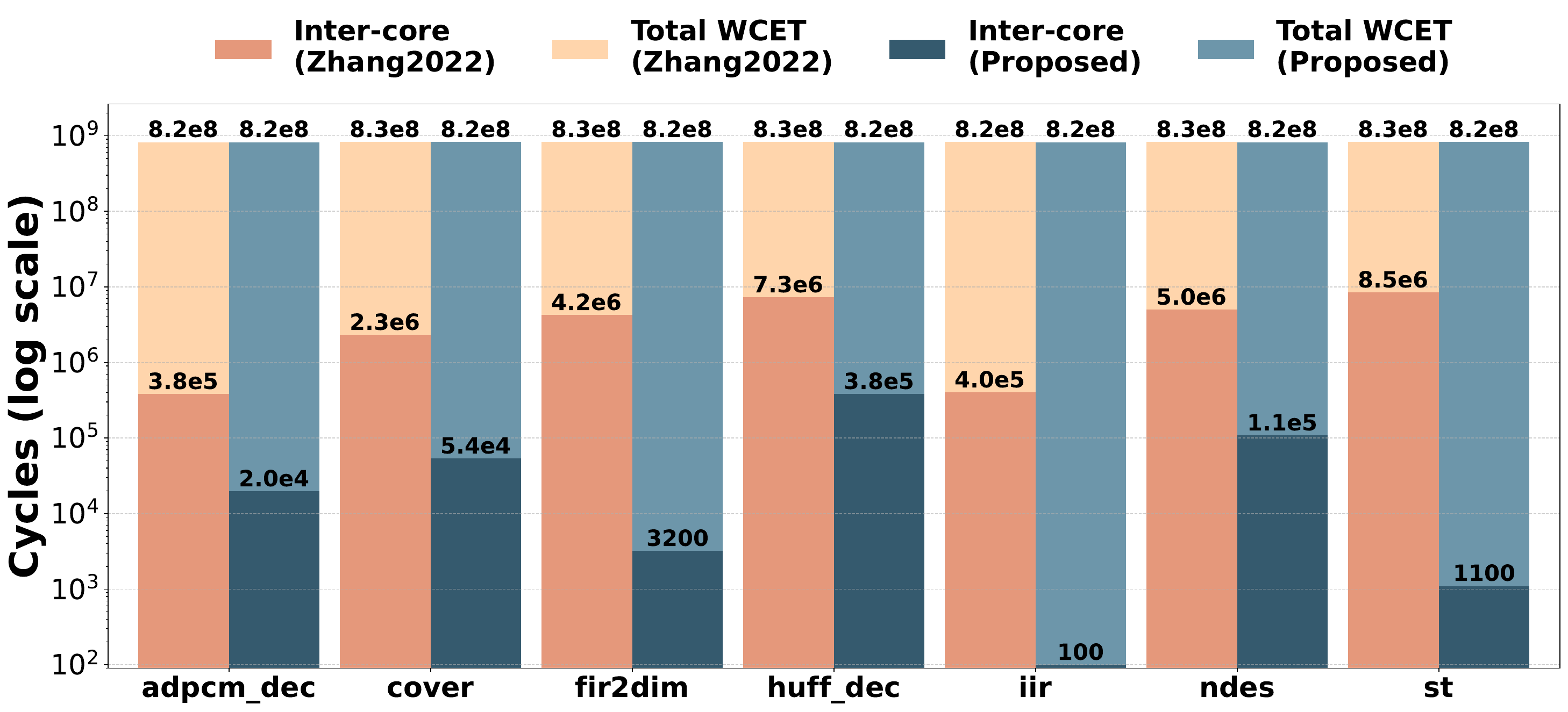}
    \caption{Inter-core cache interference and total WCET  of \texttt{pm} with $\kappa=2$ \textit{(y-axis: interfering task; dark colours: inter-core cache interference; light colours: WCET)}.}
    \label{fig:pm}
\end{figure}

\begin{figure}[t]
    \centering
    \includegraphics[width=1\linewidth]{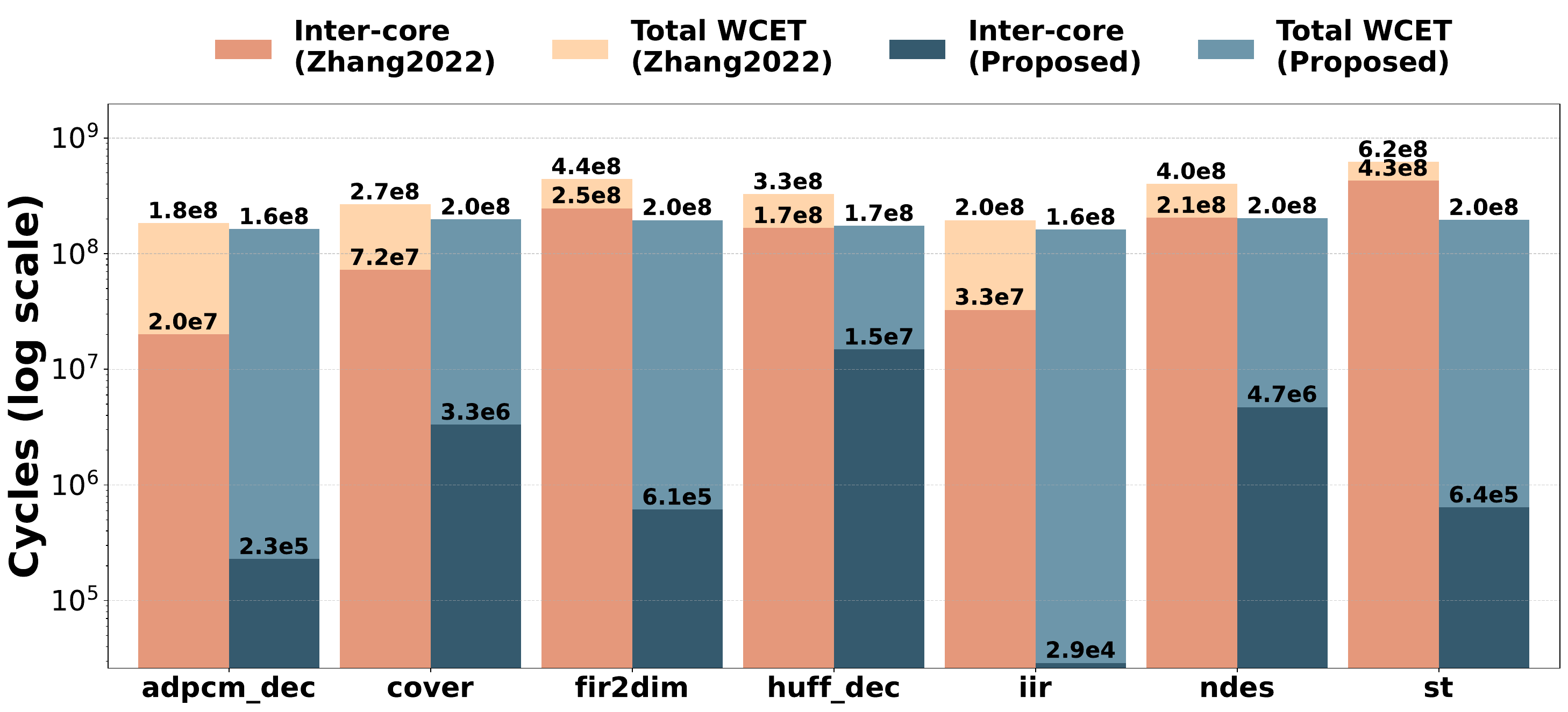}
    \caption{Inter-core cache interference and total WCET  of \texttt{powerwindow} with $\kappa=2$ \textit{(y-axis: interfering task; dark colours: inter-core cache interference; light colours: WCET)}.}
    \label{fig:powerwindow}
\end{figure}

\begin{figure}[t]
    \centering
    \includegraphics[width=1\linewidth]{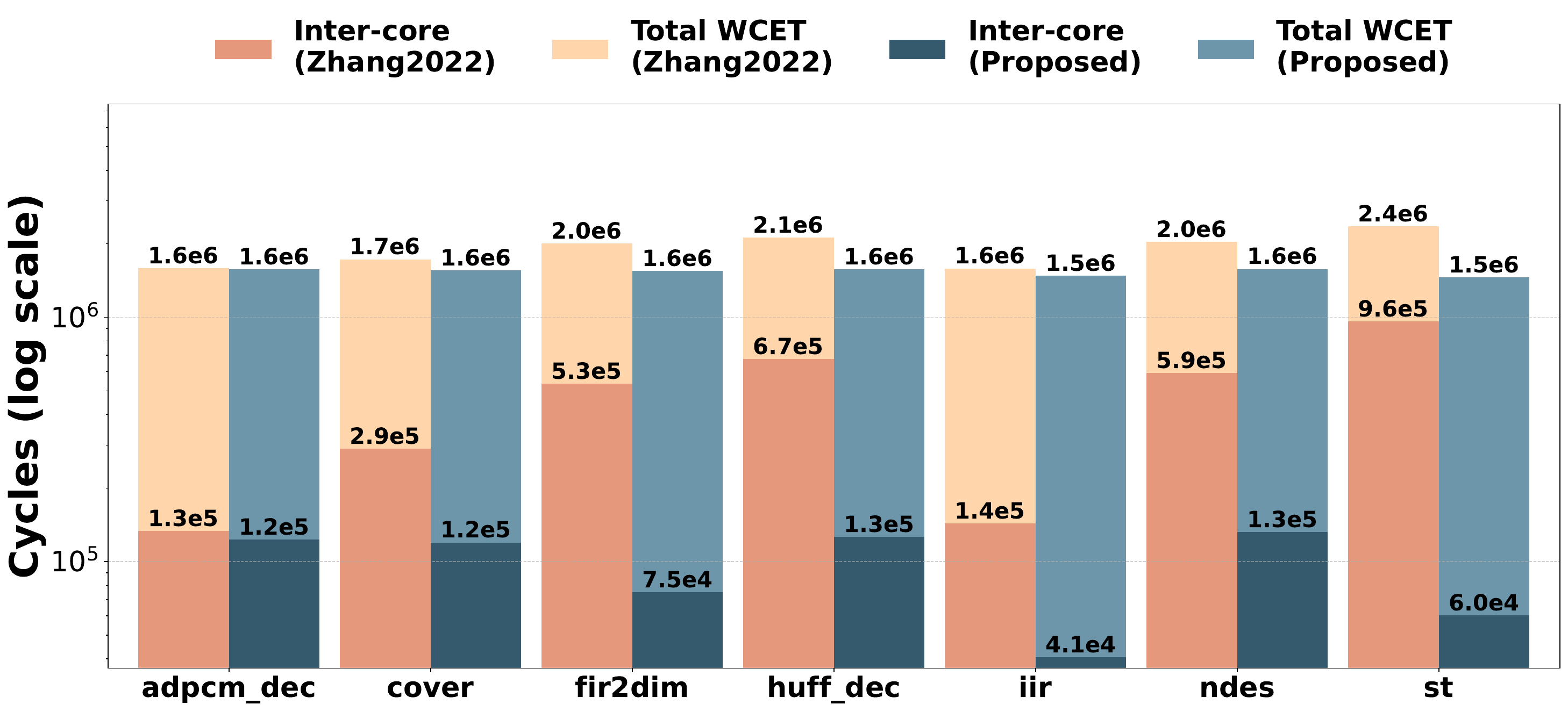}
    \caption{Inter-core cache interference and total WCET  of \texttt{prime} with $\kappa=2$ \textit{(y-axis: interfering task; dark colours: inter-core cache interference; light colours: WCET)}.}
    \label{fig:prime}
\end{figure}

\begin{figure}[t]
    \centering
    \includegraphics[width=1\linewidth]{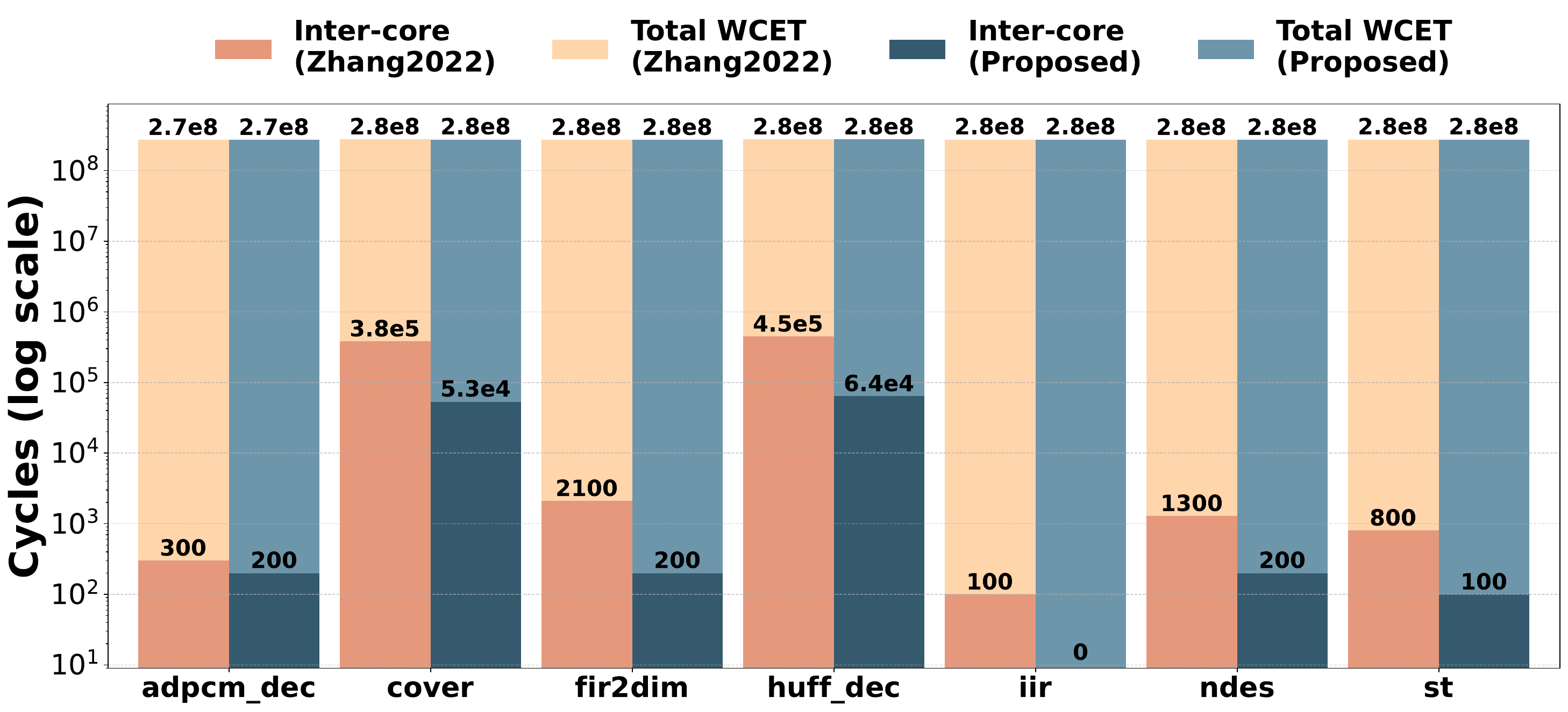}
    \caption{Inter-core cache interference and total WCET  of \texttt{sha} with $\kappa=2$ \textit{(y-axis: interfering task; dark colours: inter-core cache interference; light colours: WCET)}.}
    \label{fig:sha}
\end{figure}

\begin{figure}[t]
    \centering
    \includegraphics[width=1\linewidth]{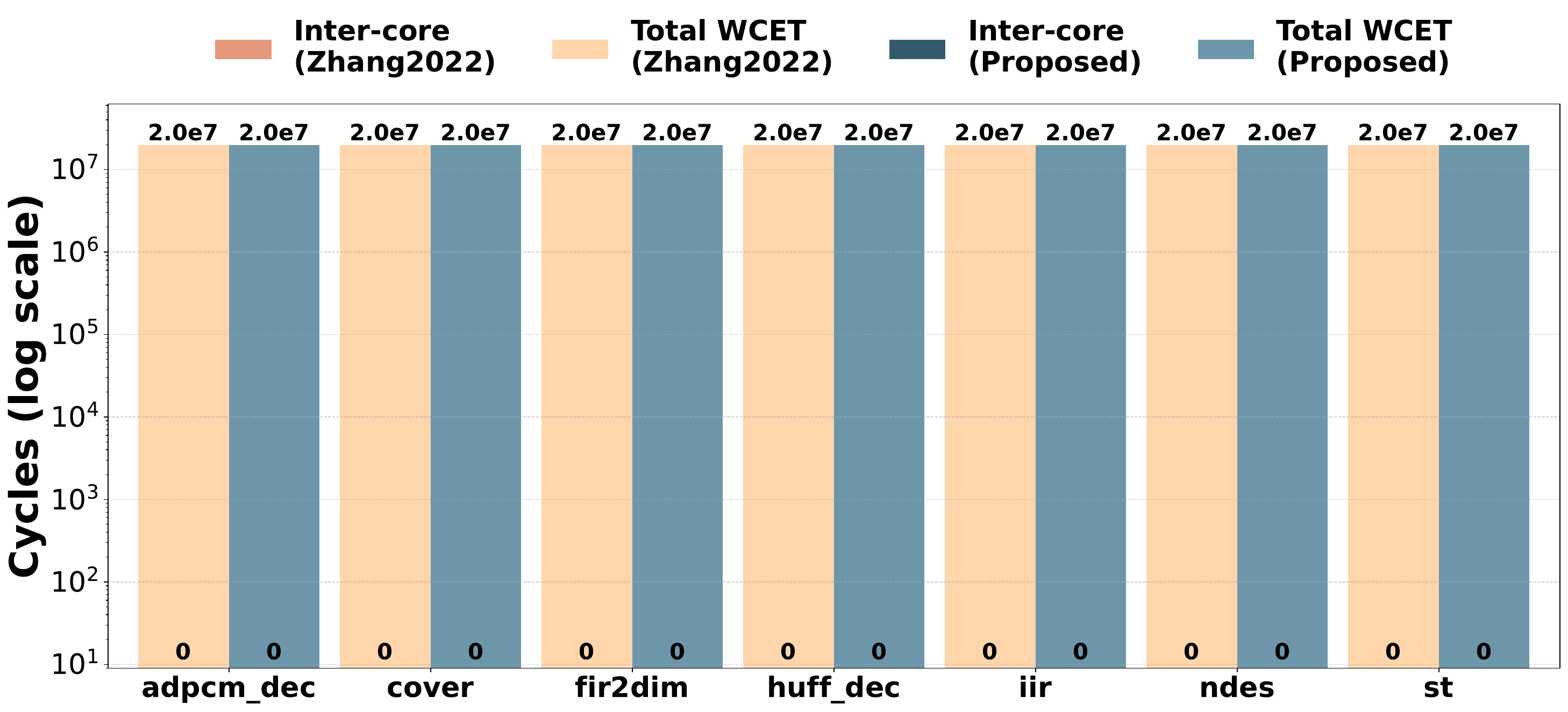}
    \caption{Inter-core cache interference and total WCET  of \texttt{bsort} with $\kappa=2$ \textit{(y-axis: interfering task; dark colours: inter-core cache interference; light colours: WCET)}.}
    \label{fig:bsort}
\end{figure}

\begin{figure}[t]
    \centering
    \includegraphics[width=1\linewidth]{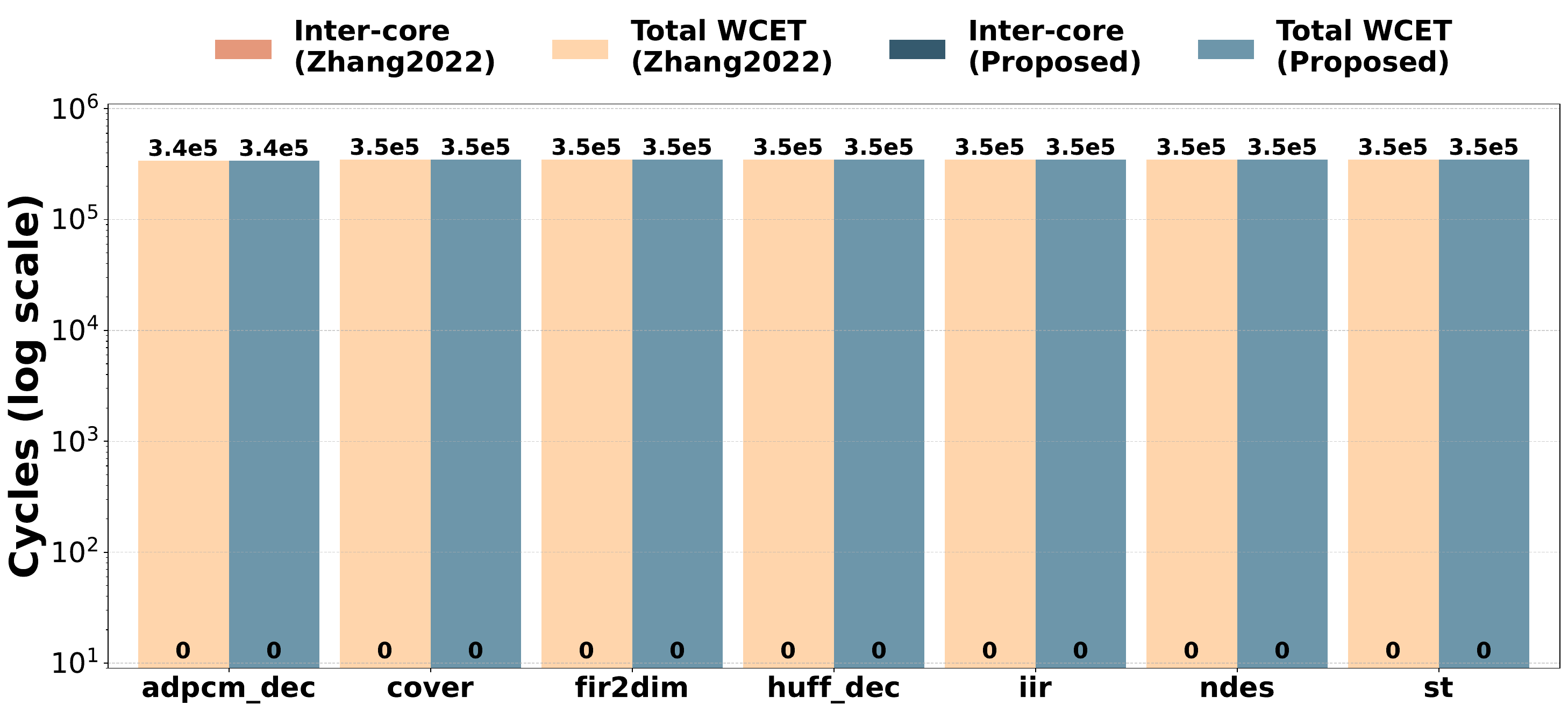}
    \caption{Inter-core cache interference and total WCET  of \texttt{cover} with $\kappa=2$ \textit{(y-axis: interfering task; dark colours: inter-core cache interference; light colours: WCET)}.}
    \label{fig:cover}
\end{figure}

\begin{figure}[t]
    \centering
    \includegraphics[width=1\linewidth]{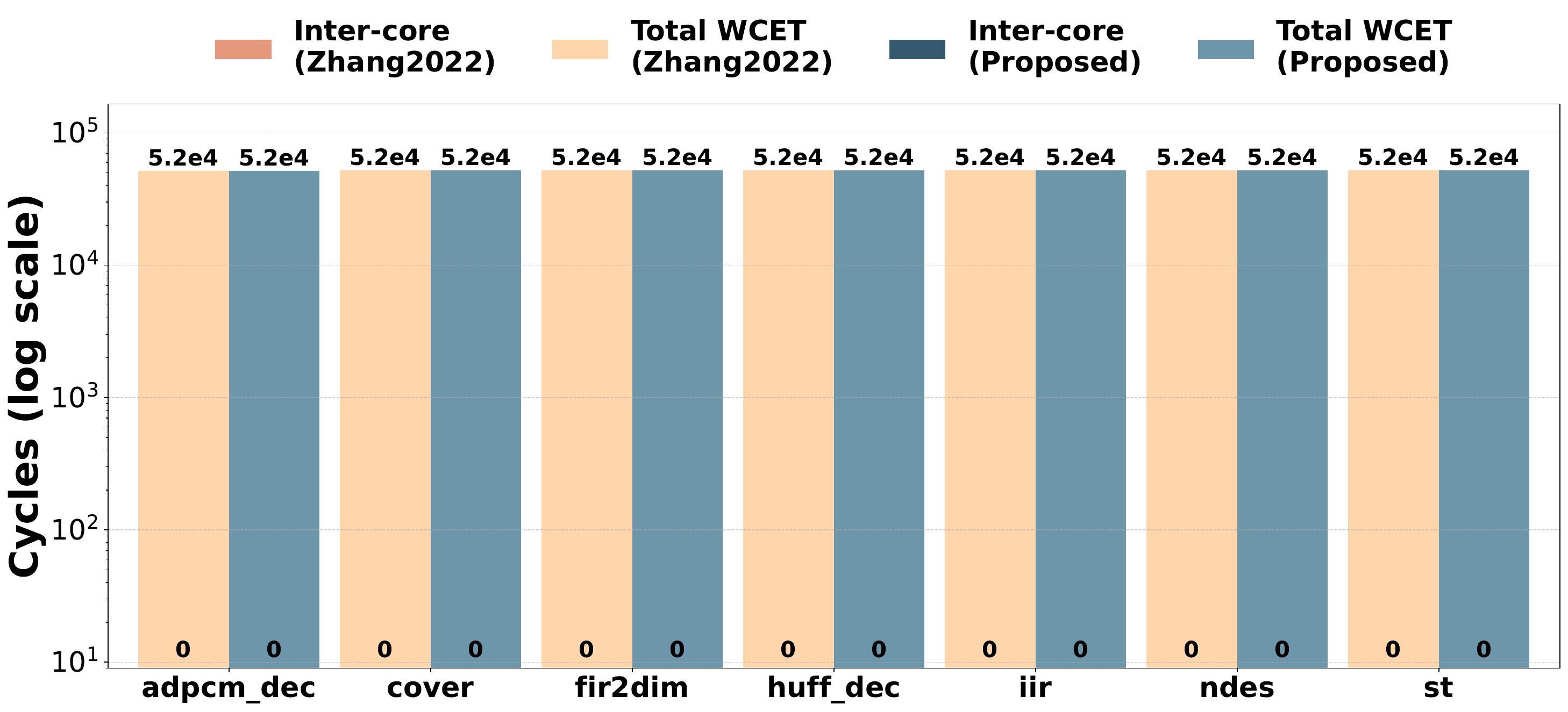}
    \caption{Timing estimations of \texttt{complex\_updates} with $\kappa=2$ \textit{(y-axis: interfering task; dark colours: inter-core interference; light colours: WCET)}.}
    \label{fig:complex_updates}
\end{figure}

\begin{figure}[t]
    \centering
    \includegraphics[width=1\linewidth]{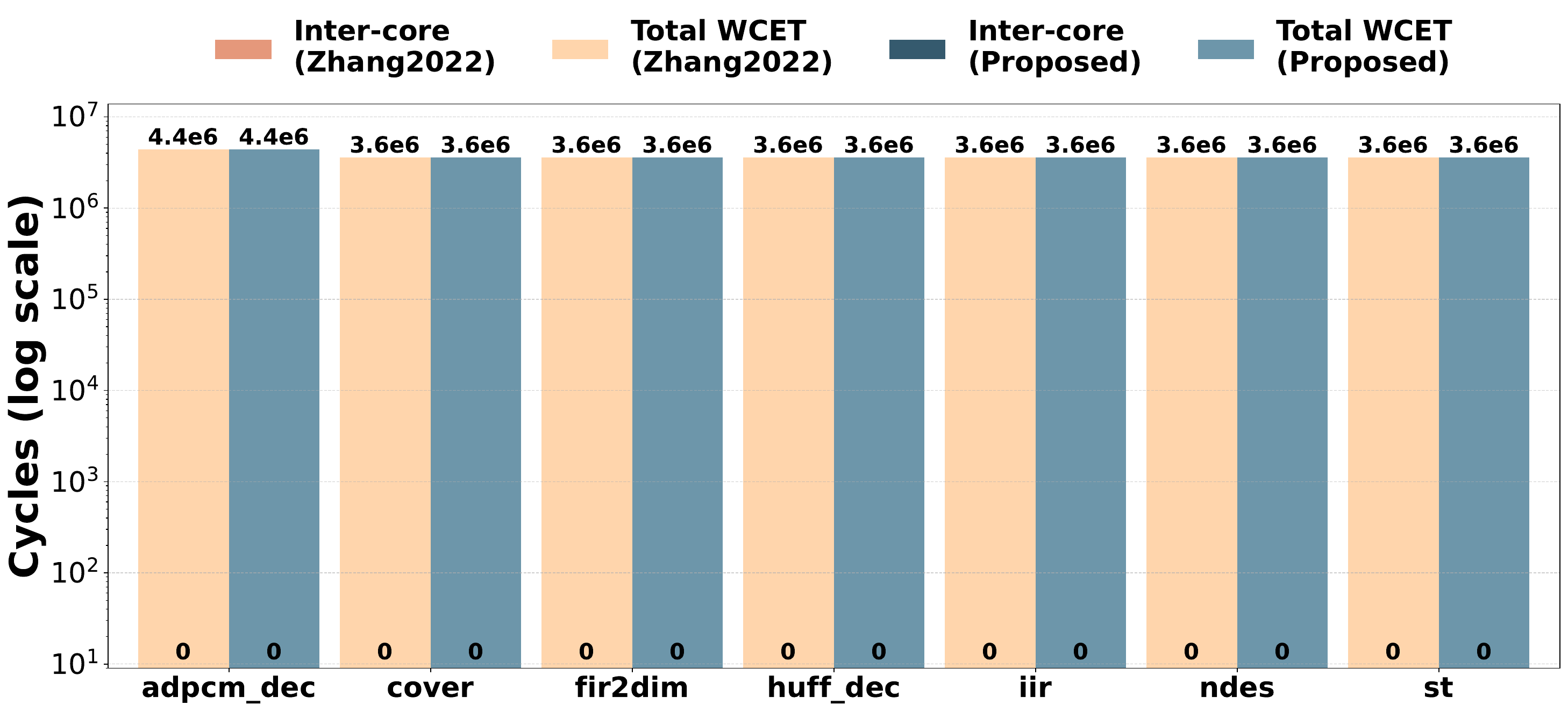}
    \caption{Inter-core cache interference and total WCET  of \texttt{cjpeg\_wrbmp} with $\kappa=2$ \textit{(y-axis: interfering task; dark colours: inter-core cache interference; light colours: WCET)}.}
    \label{fig:cjpeg_wrbmp}
\end{figure}

\begin{figure}[t]
    \centering
    \includegraphics[width=1\linewidth]{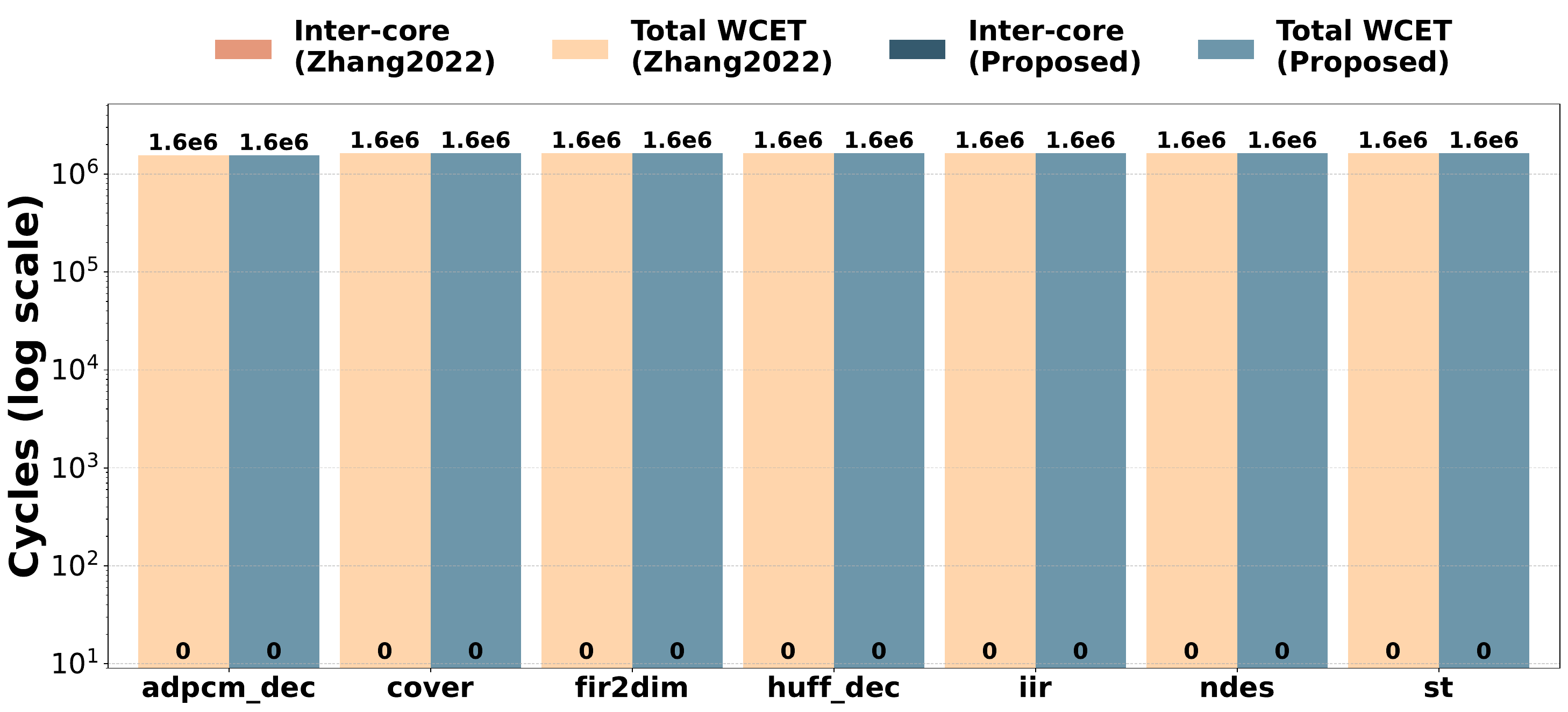}
    \caption{Inter-core cache interference and total WCET  of \texttt{countnegative} with $\kappa=2$ \textit{(y-axis: interfering task; dark colours: inter-core cache interference; light colours: WCET)}.}
    \label{fig:countnegative}
\end{figure}

\begin{figure}[t]
    \centering
    \includegraphics[width=1\linewidth]{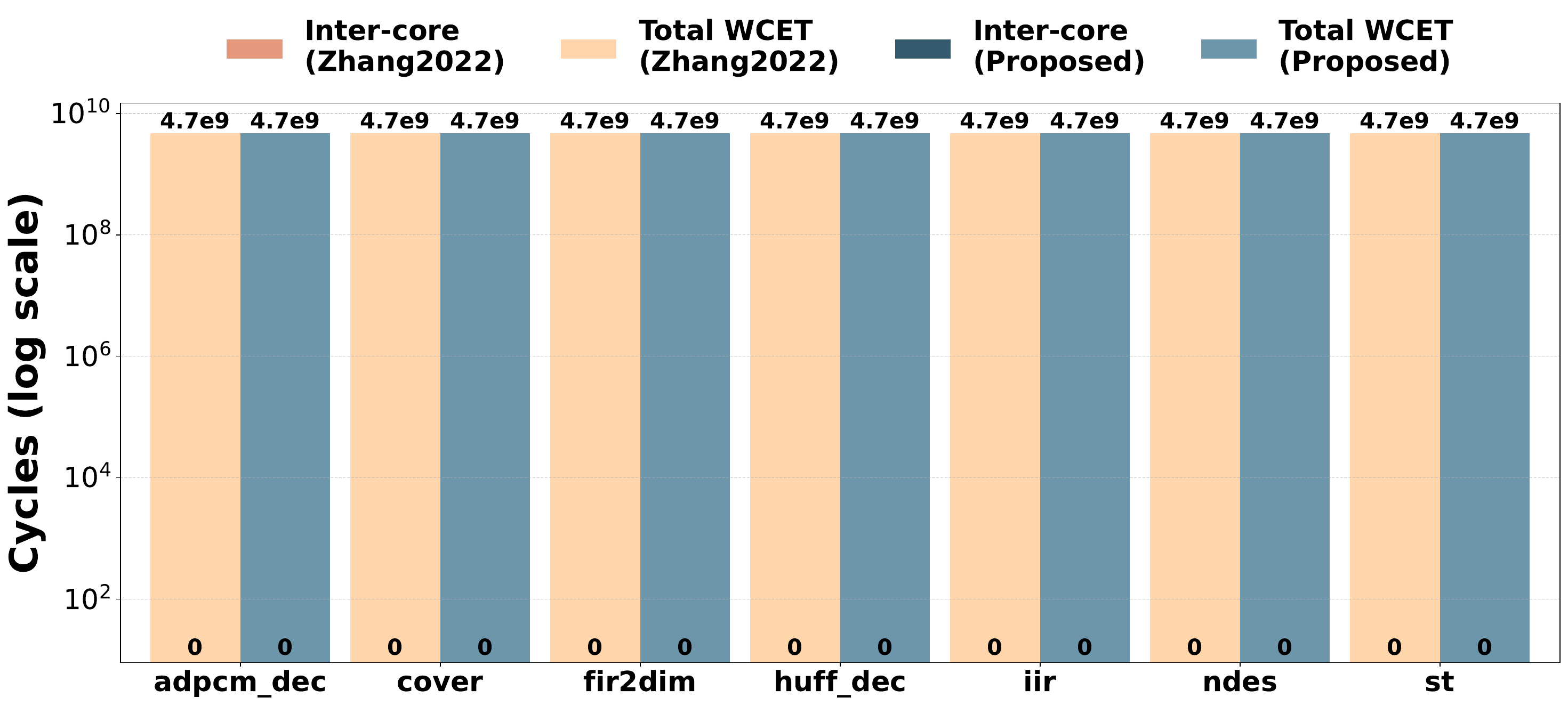}
    \caption{Inter-core cache interference and total WCET  of \texttt{fft} with $\kappa=2$ \textit{(y-axis: interfering task; dark colours: inter-core cache interference; light colours: WCET)}.}
    \label{fig:fft}
\end{figure}

\begin{figure}[t]
    \centering
    \includegraphics[width=1\linewidth]{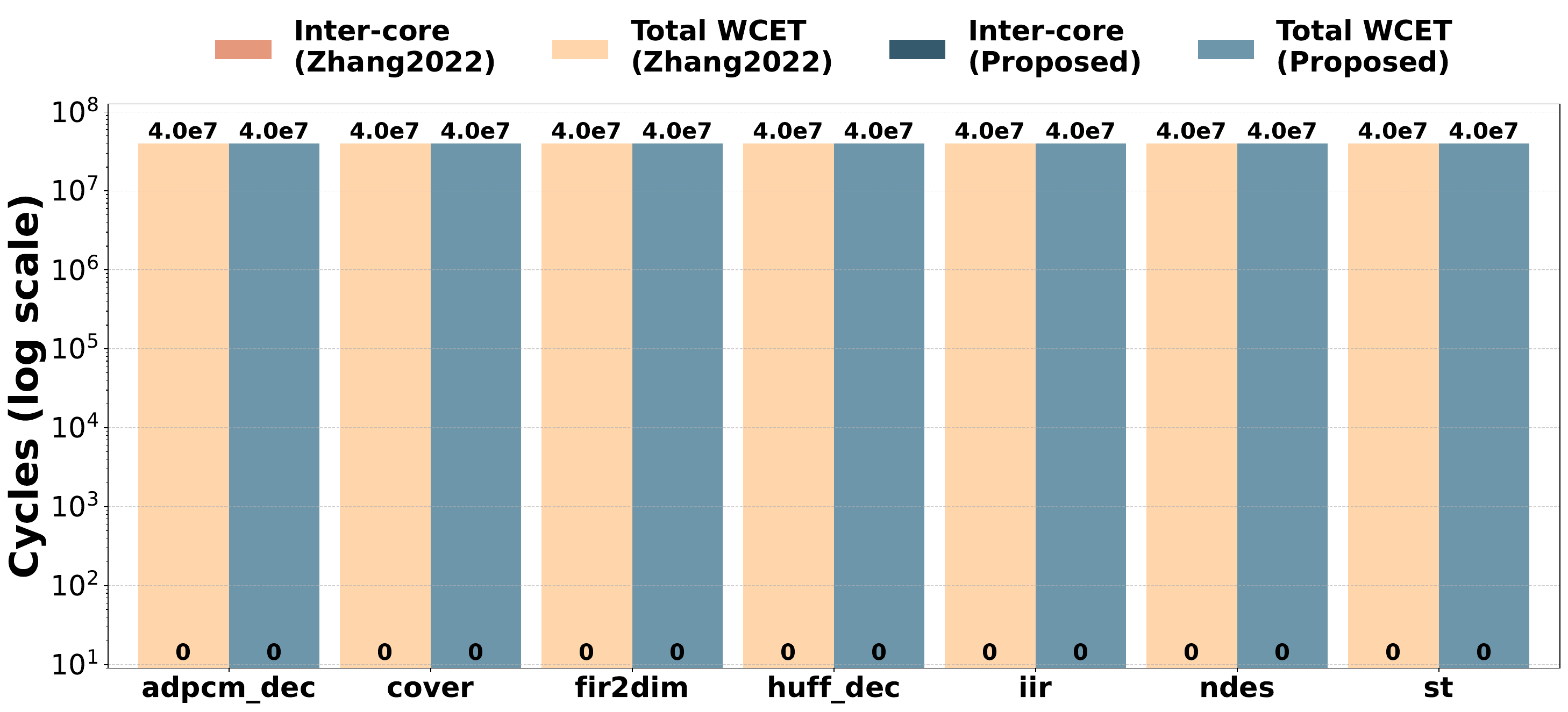}
    \caption{Inter-core cache interference and total WCET  of \texttt{filterbank} with $\kappa=2$ \textit{(y-axis: interfering task; dark colours: inter-core cache interference; light colours: WCET)}.}
    \label{fig:filterbank}
\end{figure}

\begin{figure}[t]
    \centering
    \includegraphics[width=1\linewidth]{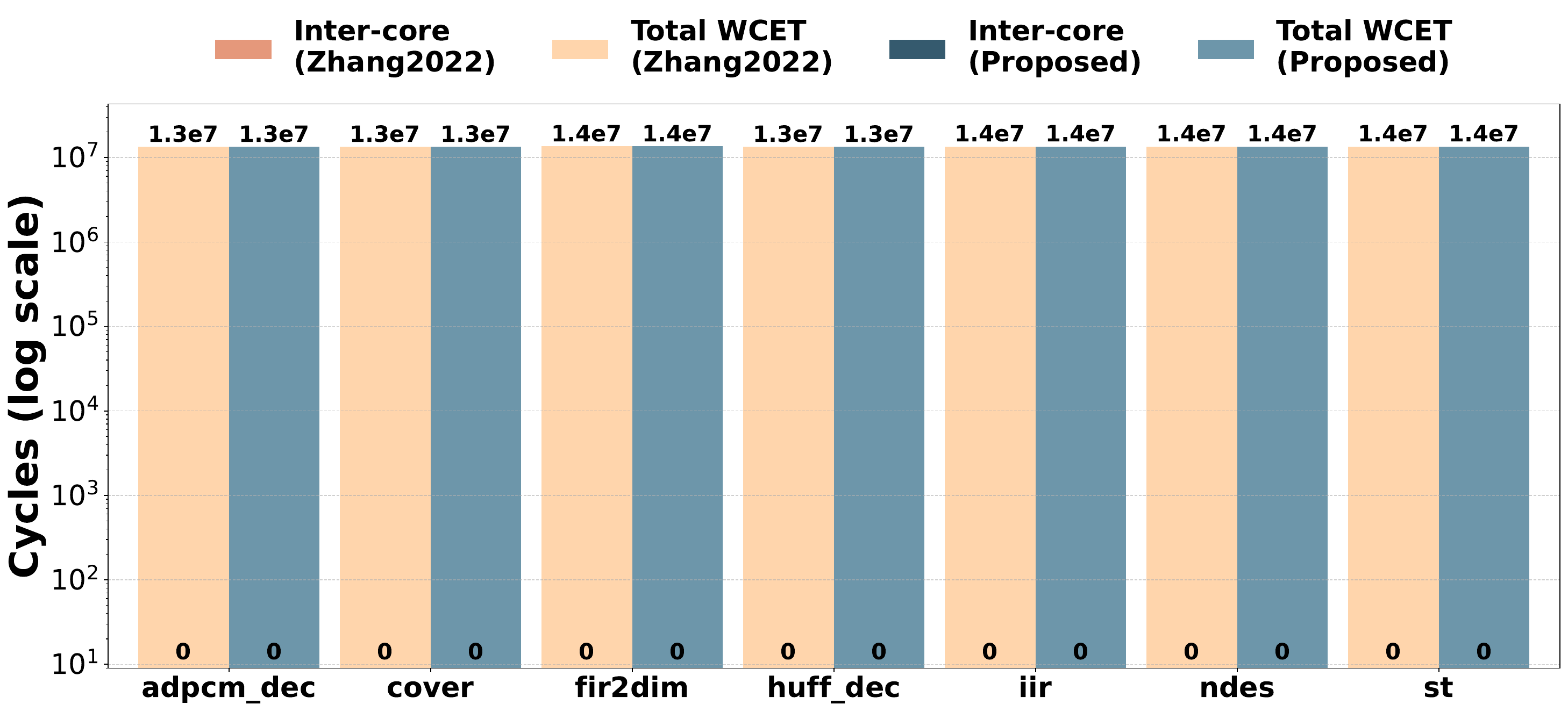}
    \caption{Inter-core cache interference and total WCET  of \texttt{lms} with $\kappa=2$ \textit{(y-axis: interfering task; dark colours: inter-core cache interference; light colours: WCET)}.}
    \label{fig:lms}
\end{figure}

\begin{figure}[t]
    \centering
    \includegraphics[width=1\linewidth]{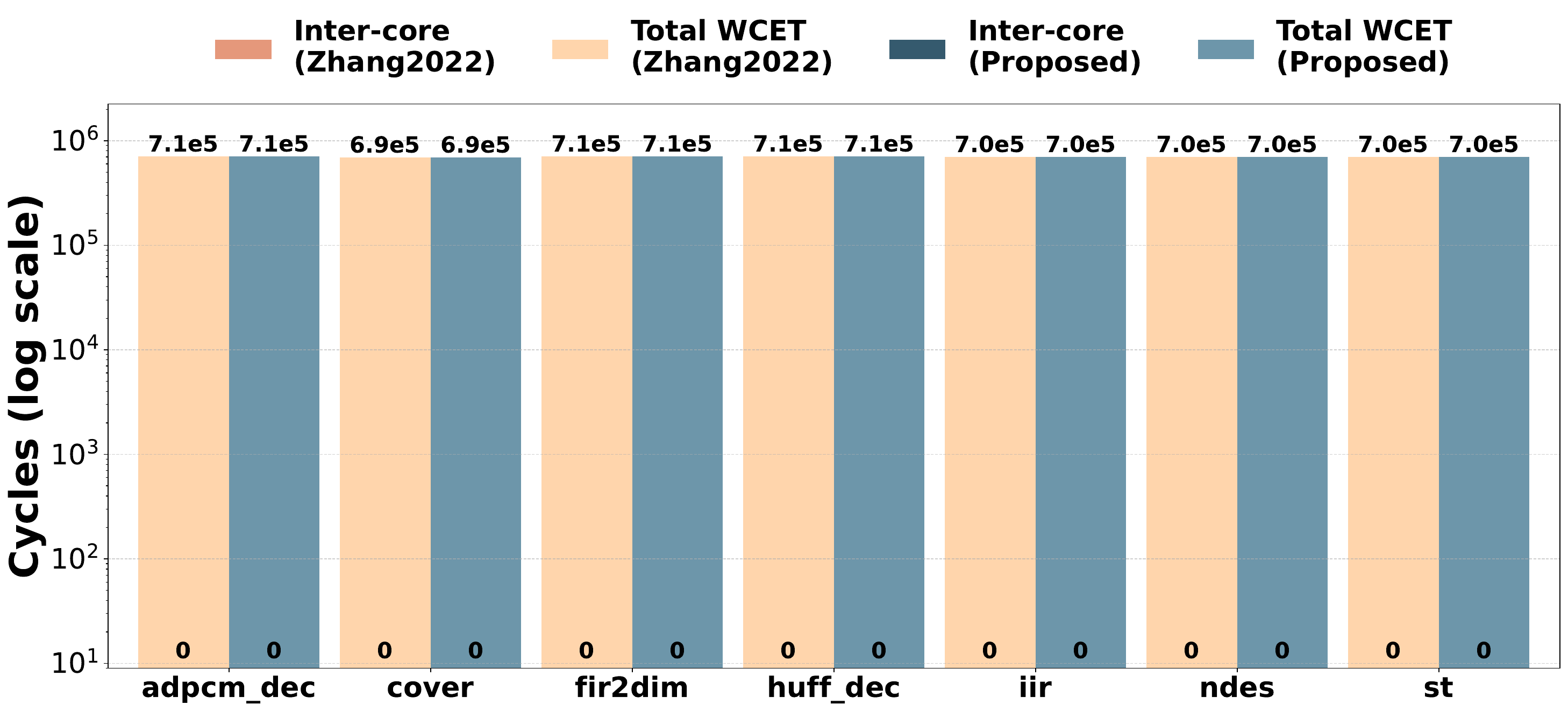}
    \caption{Inter-core cache interference and total WCET  of \texttt{ludcmp} with $\kappa=2$ \textit{(y-axis: interfering task; dark colours: inter-core cache interference; light colours: WCET)}.}
    \label{fig:ludcmp}
\end{figure}

\begin{figure}[t]
    \centering
    \includegraphics[width=1\linewidth]{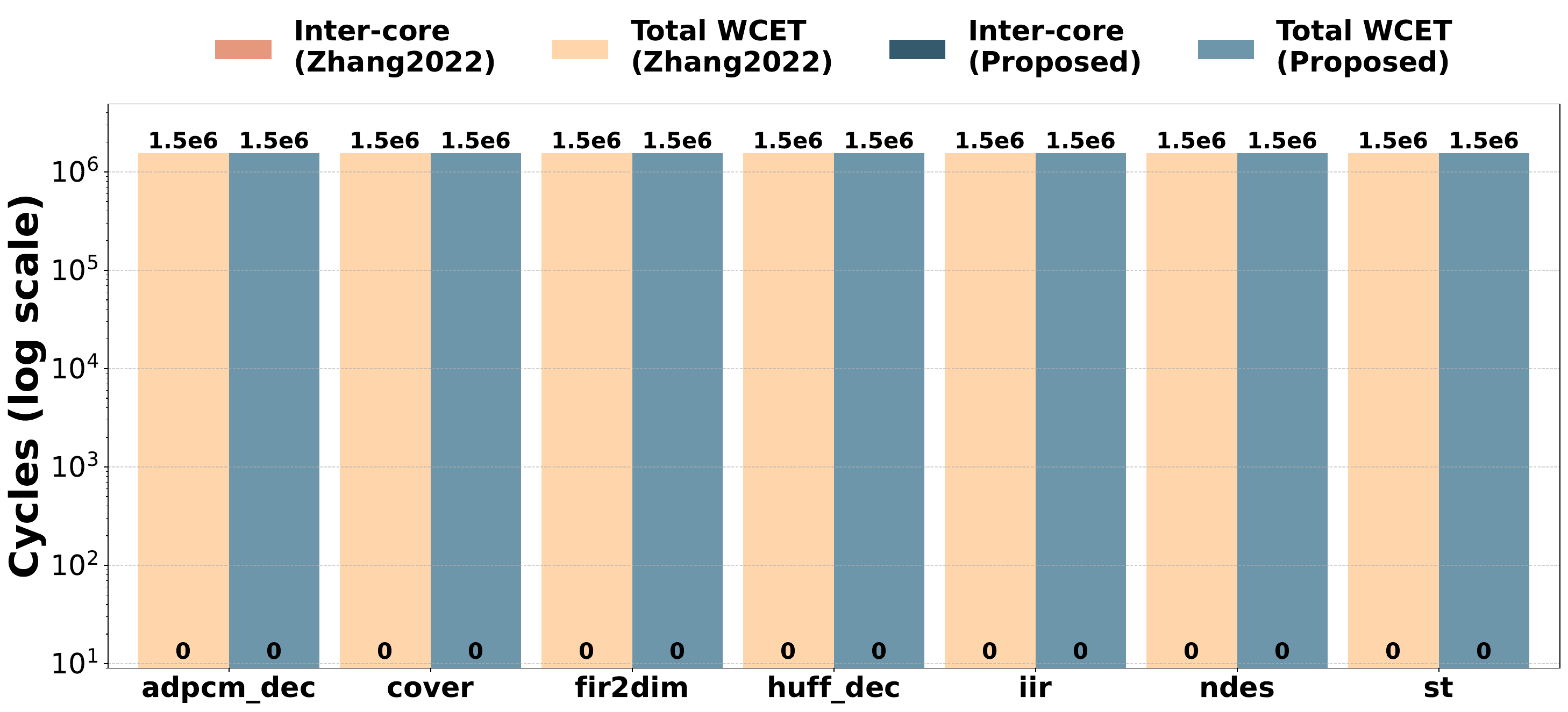}
    \caption{Inter-core cache interference and total WCET  of \texttt{matrix1} with $\kappa=2$ \textit{(y-axis: interfering task; dark colours: inter-core cache interference; light colours: WCET)}.}
    \label{fig:matrix1}
\end{figure}

\begin{figure}[t]
    \centering
    \includegraphics[width=1\linewidth]{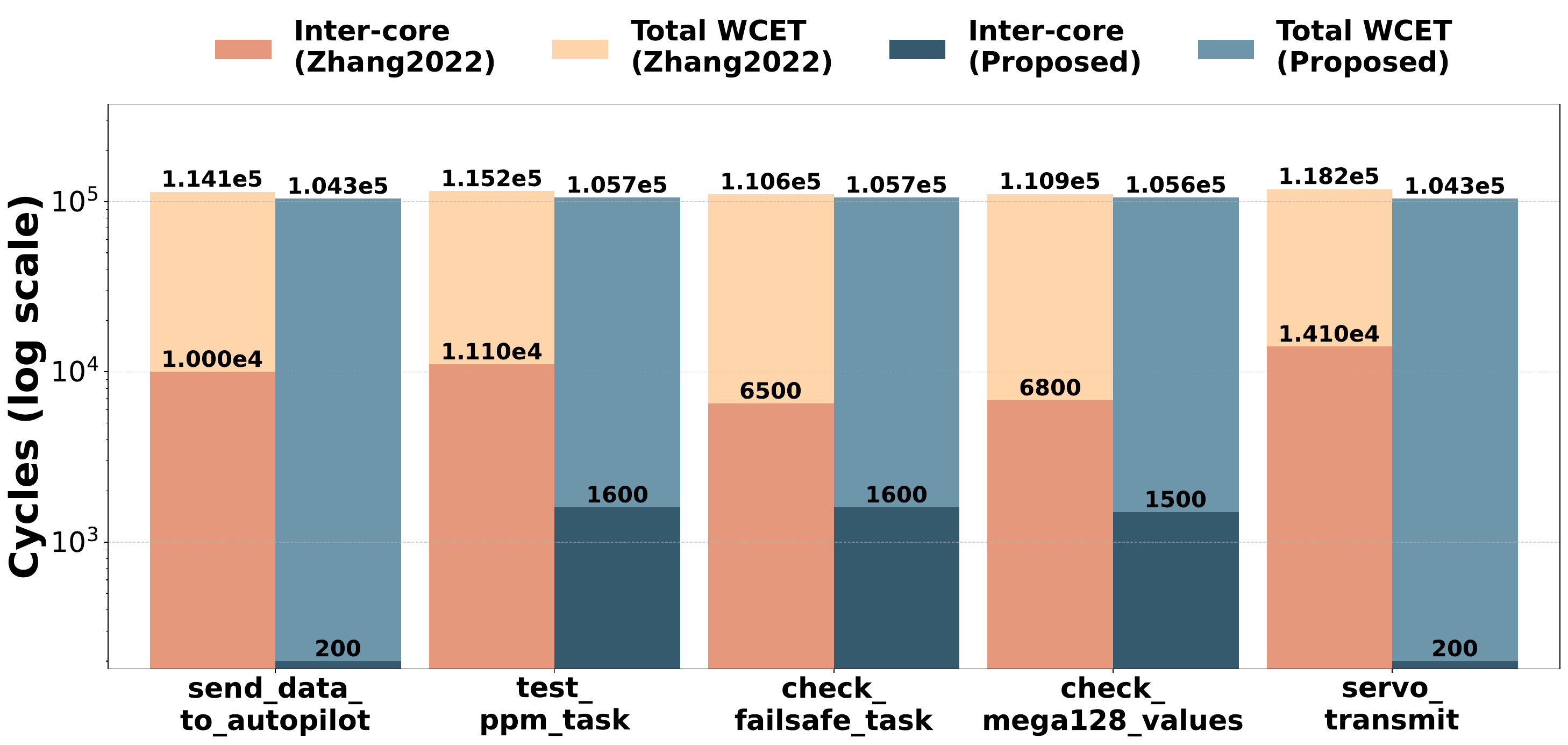}
    \caption{Inter-core cache interference and WCET  of the PapaBench task \texttt{radio\_control\_task} (AUTOPILOT) with $\kappa=2$ \textit{(y-axis: interfering task in FLY-BY-WIRE program; dark colours: inter-core cache interference; light colours: WCET)}.}
    \label{fig:papa_radio}
\end{figure}

\begin{figure}[t]
    \centering
    \includegraphics[width=1\linewidth]{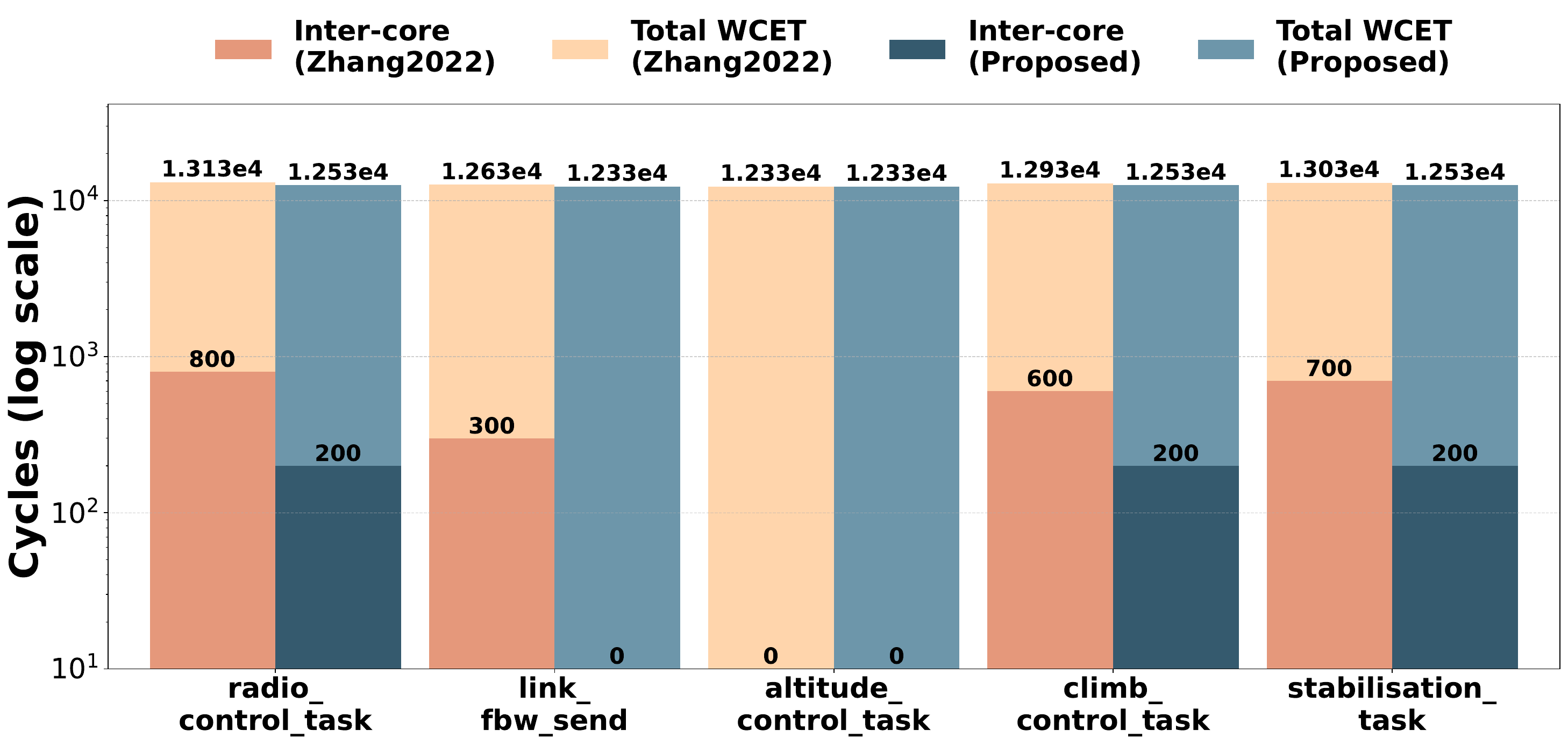}
    \caption{Inter-core cache interference and total WCET of the PapaBench task \texttt{send\_data\_to\_autopilot\_task} (FLY-BY-WIRE) with $\kappa=2$ \textit{(y-axis: interfering task in AUTOPILOT program; dark colours: inter-core cache interference; light colours: WCET)}.}
    \label{fig:papa_send}
\end{figure}

\begin{figure}[t]
    \centering
    \vfill
    \vspace{280pt}
\end{figure}

\end{document}